\documentclass[11pt]{article}
\usepackage{comment}
\usepackage[hidelinks]{hyperref}
\usepackage{enumitem}
\usepackage{amsmath,amssymb,epsf,cite,graphicx,subfigure}
\usepackage{amsmath,amsthm, tikzsymbols, dsfont}
\usepackage{comment}
\usepackage{xcolor}
\usepackage{graphicx}
\usepackage{dsfont}
\usepackage{soul}
\usepackage[T1]{fontenc} 
\usepackage{subfigure}
\usepackage[export]{adjustbox}
\usepackage{mathrsfs}  
\usepackage{mathtools}

\numberwithin{equation}{section}

\definecolor{airforceblue}{rgb}{0.36, 0.54, 0.66}
\newcommand{\beq}{\begin{equation}}
\newcommand{\eeq}{\end{equation}}

\newtheorem{definition}{Definition}
\newtheorem{thm}{Theorem}
 
  \newtheorem{lemma}{Lemma}
  \newtheorem{corollary}{Corollary}
  
  \newtheorem{claim}{Claim}
  \newtheorem*{CM3p2}{Theorem 3.2 of~\cite{CollinsMatsumoto1}}
  \newtheorem*{CM4p11}{Theorem 4.11 of~\cite{CollinsMatsumoto1}}
  \newtheorem*{CM4p10}{Theorem 4.10 of~\cite{CollinsMatsumoto1}}
  \theoremstyle{definition}
    \newtheorem{remark}{Remark}
  \newtheorem{fact}{Fact}

\newcommand{\bra}[1]{\left\langle #1\right|}

\newcommand{\ket}[1]{\left| #1\right\rangle}
\DeclarePairedDelimiterX\braket[2]{\langle}{\rangle}{#1 \delimsize\vert #2}

\newcommand{\fontH}[1]{\textsf{\textbf{#1}}}

\newcommand{\ignore}[1]{}

\DeclareMathOperator{\LO}{\textsf{LO}}
\DeclareMathOperator{\LOP}{\textsf{LOP}}
\DeclareMathOperator{\LOQ}{\textsf{LOQ}}

\DeclareMathOperator{\QUALM}{\textsf{QUALM}}

\DeclareMathOperator{\Task}{\textsf{Task}}

\begin{document}
\baselineskip=15.5pt
\pagestyle{plain}
\setcounter{page}{1}

\begin{center}
{\LARGE \bf Quantum Algorithmic Measurement}
\vspace{0.5cm}

\textbf{Dorit Aharonov$^{1,a}$, Jordan Cotler$^{2,3,b}$, Xiao-Liang Qi$^{3,c}$}

\vspace{0.4cm}

{\it ${}^1$ School of Computer Science and Engineering, The Hebrew University of Jerusalem, \\ The Edmond J. Safra Campus, 9190416 Jerusalem, Israel\\}

\vspace{0.1cm}

{\it ${}^2$ Society of Fellows, Harvard University, Cambridge, MA 02138 USA \\}

\vspace{0.1cm}

{\it ${}^3$ Stanford Institute for Theoretical Physics, Stanford University, \\ Stanford, CA 94305 USA \\}

\vspace{0.1cm}

{\tt  ${}^a$doria@cs.huji.ac.il, ${}^b$jcotler@fas.harvard.edu, ${}^c$xlqi@stanford.edu\\}

\vspace{.1cm}

\end{center}

\begin{center}
{\bf Abstract}
\end{center}
 We initiate the systematic study of experimental quantum physics from the perspective of computational complexity. To this end, we define the framework of quantum algorithmic measurements (QUALMs), a hybrid of black box quantum algorithms and interactive protocols.  We use the QUALM framework to study two important experimental problems in quantum many-body physics: determining whether a system's Hamiltonian is time-independent or time-dependent, and determining the symmetry class of the dynamics of the system. We study abstractions of these problem and show for both cases that if the experimentalist can use her experimental samples coherently (in both space and time), a provable exponential speedup is achieved compared to the standard situation in which each experimental sample is accessed separately. Our work suggests that quantum computers can provide a new type of exponential advantage: exponential savings in resources in quantum {\it experiments}.

\newpage

\tableofcontents

\section{Introduction}

Since the early days of physics, innovative methods have been invented to interrogate physical systems via {\it experiments}.  By example, some experiments measure constants of nature, such as the speed of light or the charge the electron; others quantify dynamical properties of systems, such as rates of chemical reactions; yet others infer structural properties, like the symmetry group of a crystal.  Often experiments seek to learn more abstract information, such as the chain of chemical reactions that comprise photosynthesis, or whether Yang-Mills theory describes the strong force.  We can ask: what exactly \textit{is} an experiment, in its full scope of generality? 

Over the past two decades,
we are witnessing a new era in this respect, in which ingredients, ideas and concepts originating from the world of quantum computation are being incorporated into the experimental physics toolbox. Examples are beautiful and very diverse. For instance, in quantum state tomography\cite{tomography}
compressed sensing allows for much more efficient tomography protocols than the straightforward ones, in many cases~\cite{compressedsensing1,compressedsensing2,compressedsensing3,compressedsensing4}. A restricted type of tomography coined ``shadow tomography'' was proposed by Aaronson \cite{aaronsonshadow}, where ideas from classical learning theory were used to provide sample-efficient protocols which allow making restricted predictions of exponentially many future measurements.  Recently these results were made gate-efficient under some non-trivial restrictions on the set of observables to be predicted \cite{CW1, Evans1, preskillshadow, Yu1}. In the area of metrology and sensing there are also very interesting examples, including super-resolution using large-scale entanglement of NOON states \cite{noon}, 
improved spectroscopy \cite{RetzkerScience}, 
and the application of restricted error correction during measurements to prolong signal decay \cite{Arrad,Dur,Ozeri,Lukin, PreskillZhou}.
In \cite{AtiaAharonov} quantum computational tools were shown to enable  {\it exponential} accuracies of energy measurements, in (highly) restricted settings. 
In yet another direction, a plethora of theoretical schemes for {\it verifying} that a quantum computer evolves as planned were suggested \cite{QPIP,BFK09,surveyverification,RUV}; they require mild quantum computational power to apply. 
One such example was tested experimentally \cite{PhilipBlind}.
Finally, in the gedanken (and highly theoretical) experiment suggested by Hayden and Preskill to estimate a black hole’s information retention time \cite{HaydenPreskill}, 
a full-fledged quantum computer is used. 

The above examples, among many others we have not listed,
constitute strong evidence that leveraging
{\it quantum computational resources} to manipulate and measure physical systems may dramatically enhance experimental capabilities.  But what is the general scope of leveraging quantum computers for experiments? And what are the limitations of such protocols?   We will argue that these developments both challenge and clarify the paradigm of experimentation itself, and its relation to computation.  Accordingly, we will develop a computational framework for the most general kind of quantum experiment which can be implemented in the physical world. 

The first goal of this paper is to provide a theoretical framework in which one can rigorously state and study these questions. 
To do this, we make use of the language of {\it computational complexity} and {\it communication complexity}. In other words, here we initiate the systematic study of quantum experiments from a complexity-theoretic point of view.  Initial seeds for such an approach were given in \cite{AtiaAharonov, CJQW1}.

\subsection{A computational complexity point of view on quantum experiments} 
Our starting point is this: we can view a physical experiment, very roughly, as a generalization of an algorithm. 
 An algorithm (classical or quantum) computes a function whose argument is a {\it classical} input string (e.g.~computing the permanent for a matrix given as a bit string). In a quantum experiment, we can view the {\it input} to the process as a \textit{quantum system} (a state, a quantum process, or both) to which we have some form of ``oracle''
access in the lab.  The goal of a quantum experiment is to implement a function mapping (very roughly) $\{\text{physical quantum system}\} \to \{\text{classical bit string}\}$. 
 Thus to define a quantum experiment, we turn to the insights underlying the theory of quantum 
 algorithms. 

Underlying the theory of quantum computation is 
the {\it quantum Church-Turing thesis}\cite{vazirani}. 
The  thesis suggests that any physical process can be efficiently (i.e., with only polynomial overhead) simulated by a \textit{quantum} algorithm applying local quantum gates and local measurements. This observation not only constitutes the pillar on which the entire theory of quantum algorithms and quantum computational complexity stands, but it has also had a profound impact on our understanding of quantum physics in the past two decades (see, e.g.~\cite{Gharibian14}). 
The thesis is widely conjectured to be true, and has been verified in many contexts, see e.g.~\cite{preskilljordan}. 
We take this insight one step further to the setting of quantum {\it experiments}; here we provide a computational complexity framework for such processes. We argue that:
\begin{itemize} 
\item Experiments should be viewed as generalizations of quantum algorithms. They can be studied and designed abstractly, using quantum gates and circuits.  
\item One can study the {\it computational complexity} of quantum experiments, as an extension of the way the computational complexity of quantum algorithms is studied. 
\end{itemize} 

 We thus first make a simple but crucial proposition: In the most general quantum experiments, the experimentalist can in principle apply arbitrarily complicated quantum computations during the experiment. For example, she can prepare many copies of the physical system and entangle them with one another, pre- and post-process results of intermediate measurements using a quantum computing register on the side, and so on. By the quantum Church-Turing thesis, we would model the physical experimental process as some sort of a quantum {\it algorithm}, composed of {\it quantum gates}. This will indeed capture the most general possible quantum processes that the experimentalist controls, putting aside the details of implementation.

The computational point of view on quantum experiments naturally enables us to quantify the resources required to conduct the experiment, in a meaningful way. To do this, we again take insight from the quantum Church Turing thesis, which tells us that the resources used during the experiment can be quantified, by viewing the experimental process as a quantum computation made of local gates and local measurements. This enables us to estimate the {\it computational complexity}
of an experiment, by counting the number 
of elementary, local quantum operations (e.g.~gates or measurements) required to implement the quantum experiment. Thus we will arrive at a characterization of the complexity of quantum experiments in a manner \textit{independent} of the physical systems and apparatuses by which they are instantiated (as a direct generalization of the independence of the complexity of quantum algorithms from the exact details of physical implementation). This approach enables the study of quantum experiments and the resources they require, from an abstract, mathematical point of view, independent of exact physical details. This is exactly our purpose.

We thus use the language of {\it computational complexity} to define an abstract model of general experiments, which we call {\it quantum algorithmic measurements}, or QUALMs, in which the complexity of quantum experiments can be quantified in a meaningful way. Relying on the quantum Church Turing thesis, we believe that any physical implementation of a quantum experiment can be simulated in this framework with at most polynomial overhead; in other words, we hypothesize that the QUALM framework is {\it universal} for quantum experiments.  Initial seeds for our approach were given in~\cite{AtiaAharonov, CJQW1}.

We note that the QUALM framework allows, in principle, to include a full-fledged quantum computer as part of the experimental apparatus, and have it interact coherently in time and space with the physical system to be measured.  Of course, full-fledged quantum computers are not yet available, and so in current quantum experiments the quantum manipulations are limited, but our goal is to develop a {\it theory} of the most general quantum experiments which are in principle possible in the quantum world, so that we can study their limitations and possibilities.  As we will see later in the paper, starting from such a general model, we can then rigorously define and study various interesting restrictions. We can choose these restricted variants to be close to the physical reality we want to understand.  

We will next define these abstract experimental protocols, which we call \textit{quantum algorithmic measurements}, abbreviated as QUALMs, in detail. 

\subsection{The quantum algorithmic measurement framework} 

Our starting point is the postulate that the goal of any physical experiment is to compute a {\it function} from an {\it input} physical system to a classical outcome. The value of the function holds the information that the experimentalist wishes to extract about that physical system. In contrast to standard (quantum or classical) algorithms, the input for a physical experiment is a physical system; the experimentalist is not given a full classical description of it. Instead, access to the physical input system is mediated by quantum operations and measurements, which in general provide only limited information.

A first natural attempt is to model experiments as ``black box'' quantum algorithms: 
queries to the physical system (namely applications of the unknown superoperator describing the system) are interlaced with controlled {\it quantum computations} (or quantum manipulations, modeled by sequences of quantum gates) applied by the experimentalist.  However, it turns out that this model is not general enough to describe all quantum experiments; in particular, it does not allow the physical system being studied to maintain its own inaccessible (or private) quantum memory.

Towards defining a universal model of experiments, consider the concrete example of an X-ray diffraction experiment, performed to determine the crystal structure of a material. The experiment involves a crystal sample, X-ray photons which exhibit an electromagnetic interaction with the crystal, and a camera as well as other lab equipment which only interact with the photons (see Figure~\ref{fig:interaction}(b)).  This is a very general situation: 
in a physical experiment, the experimentalist usually cannot
fully interact with all degrees of freedom of the physical system she desires to measure.  
We thus model a general experimental system as consisting of three subsystems (registers).
  The first is called ``Nature'', denoted by $\fontH{N}$, which we view as the system that Nature holds secretly, and to which the experimentalist has no direct access in the experiment (this is the crystal in the above example).  Our apparatus in the lab 
  is denoted by 
   $\fontH{W}$ for ``work space'' (e.g.~the camera and data processors in the X-ray example). 
   The degrees of freedom which the experimentalist does have access to, 
   but 
  which couple to the hidden degrees of freedom of $\fontH{N}$, comprise the ``lab'' register $\fontH{L}$ (e.g., the X-ray photons).
  
  In the X-ray example, the {\it input physical system}, which the experimentalist would like to measure or learn about, can be described by the combination of the (unknown) state and structure of the crystal, together with the (unknown) interaction between the crystal and the photons (it is unknown since it is a function of the unknown properties of the crystal). More generally, we model an input physical system by a 
  {\it lab oracle}. The lab oracle is defined as a pair, consisting of the initial state of the hidden degrees of freedom $\fontH{N}$, as well as a superoperator modeling the dynamics of $\fontH{N}$ and the interaction between $\fontH{N}$ and $\fontH{L}$. 
\begin{definition}\label{def:laboracleroughly0}(Roughly) A lab oracle is defined as a pair $\LO(\fontH{N},\fontH{L}) = (\mathcal{E}_{\fontH{NL}},\rho_{\fontH{N}})$, where $\mathcal{E}_{\fontH{NL}}$ is a superoperator acting jointly on $\fontH{N}$ and $\fontH{L}$, and $\rho_{\fontH{N}}$ is the initial state of the $\fontH{N}$ system. 
\end{definition} 
\noindent
Our general model of a physical experiment is 
described in Figure~\ref{fig:interaction}(a). We model a physical experiment as an interactive protocol 
  applied between the work space $\fontH{W}$ and
Nature $\fontH{N}$; these two registers communicate using the lab register $\fontH{L}$ which serves as a ``message'' register. 
The superoperator
$\mathcal{E}_{\fontH{NL}}$ describing the interaction between $\fontH{L}$ and  $\fontH{N}$, given by the physical system to be measured or probed, 
is unknown and is viewed as a 
{\it black box} which can be ``queried''; namely, it can be applied at will by the experimentalist.  
Note that we have no 
control over this interaction; it is dictated by physics, but we assume that the experimentalist can apply it at her will.  
The addition of the Nature register $\fontH{N}$ allows us to arrive at a model which is both quite simple as well as, we believe, capable of describing the most general experiments.

We next introduce a notion parallel to that of a {\it computational problem} in the algorithmic world. It is called a $\Task$, and it encapsulates the experimental goal that the experimentalist wishes to achieve. The $\Task$ consists of the information which the experimentalist wishes to extract, expressed as a function from lab oracles (capturing physical systems) to classical outputs. It also includes the constraints on the experiment due to various limitations in the lab, which are specified by the {\it admissible gate set}. Importantly, 
these gates are constrained to not act on $\fontH{N}$, and they can also express additional constraints in the labratory such as geometric restrictions on the interactions.
\begin{definition}\label{def:taskroughly}(Roughly) A task is a tuple $\Task = (\fontH{S}_{in}, \fontH{S}_{out}, f,\mathcal{G})$, associated with a given system $\fontH{N}\otimes\fontH{L}\otimes\fontH{W}$. Here, $\fontH{S}_{in}, \fontH{S}_{out}\subseteq \fontH{W}$ consist of $p$ and $q$ qubits respectively;   
$f$ is 
a function
$$f : \{\LO_0, \LO_1, \LO_2,...\} \times \{0,1\}^p \longrightarrow \{0,1\}^q\,,$$
and $\mathcal{G}$ is a set of admissible gates on $\fontH{L}\otimes\fontH{W}$.  In the domain of $f$, $\{\LO_0, \LO_1, \LO_2,...\}$ is a set of lab oracles.
\end{definition} 

Note that the function $f$ in the above definition has as its inputs not only the lab oracle, but also a classical bit string.  The latter should be thought of as additional parameters that describe the desired experimental output, e.g.~the specification of the temperature at which the experimentalist may want to perform a certain experiment. While in the above definition the output of $f$ is deterministic, a very natural generalization is for the output to be a {\it probability distribution} over classical output strings; this allows expressing approximated tasks or taking into account finiteness of precision (see discussion in the state tomography 
Example 
in Subsection~\ref{subsec:ExampleQUALMs}. 

Finally we can define a QUALM; this specifies an experimental protocol or experimental design; namely, 
a way to implement the experiment 
such that it achieves the desired task. 
\begin{definition}(Roughly) A $\QUALM$(\fontH{N},\fontH{L},\fontH{W}) is a specification of a sequence of admissable gates from a set $\mathcal{G}$ on the subsystems  $\fontH{L},\fontH{W}$, interlaced with applications of a black box operator $\square$ acting on  $\fontH{N},\fontH{L}$; some of the qubits (i.e., those in $\fontH{S}_{in}$) in the register $\fontH{W}$ are marked as inputs and some (i.e., those in $\fontH{S}_{out}$) as outputs. 
\end{definition} 
\noindent 
We note that this definition is tightly related to the definitions of quantum strategies~\cite{watrous} and quantum combs~\cite{comb1} introduced in the context of quantum interactive protocols or games. 
A QUALM is designed in order to achieve a specific experimental $\Task$. 
To see which $\Task$ is achieved by the 
QUALM, we can view the QUALM naturally as a map from the input lab oracle 
 $(\mathcal{E}_{\fontH{NL}},\rho_{\fontH{N}})$ to 
 a standard quantum circuit, acting on $\fontH{N}\otimes \fontH{L}\otimes \fontH{W}$, whose input bits are in
 $\fontH{S}_{\text{in}}$ and output bits are in
 $\fontH{S}_{\text{out}}$\,. 
 $\fontH{N}$ is initialized to $\rho_{\fontH{N}}$, all qubits in $\fontH{W}$ and $\fontH{L}$ except for $\fontH{S}_{in}$ are initialized  to $0$, and the input to the circuit is given in $\fontH{S}_{in}$\,. The circuit applies to this initial state the gate sequence of the QUALM, where whenever  $\square$ appears, $\mathcal{E}_{\fontH{NL}}$ is applied. The output of the circuit is given by measuring $\fontH{S}_{\text{out}}$ in the computational basis.
We say that a QUALM  {\it achieves} a given $\Task$ if (i) the sets $\mathcal{G}, \fontH{S}_{in}, \fontH{S}_{out}$ are the same for the QUALM and the $\Task$, and (ii) for every lab oracle  $\LO$ and input string $x$ to $\fontH{S}_{in}$, the result of the measurement of the output qubits $\fontH{S}_{out}$ (after the application of the corresponding circuit) is equal (or close, in cases of approximations) to $f(\LO,x)$, with $f$ being the $\Task$ function. 
 
The computational complexity of a QUALM is the number of gates plus the number of lab oracle applications; the QUALM complexity of the $\Task$ is that of the most efficient QUALM that achieves it. We propose that QUALM complexity is a standardized way to quantify and study the (asymptotic behavior of the) cost of achieving an experimental task in the laboratory.  

In Section \ref{subsec:ExampleQUALMs}, we provide a versatile set of examples for how different experimental tasks can be viewed as $\Task$s and be realized by QUALMs. We hypothesis that the QUALM framework is a {\it universal} model for quantum experiments, in the sense that any physical process realizing an experimental task can be simulated {\it efficiently} (i.e., with at most polynomial overhead in all resources) by a QUALM; in other words, we speculate that the quantum Church Turing thesis can be extended from computational problems to {\it experiments}, by generalizing quantum algorithms to QUALMs. 
We thus arrive at a
new framework which allows us to initiate a rigorous study of the resources required in order to perform 
physical experimental tasks.

\begin{figure}
\begin{center}
    \includegraphics[width=6.75in]{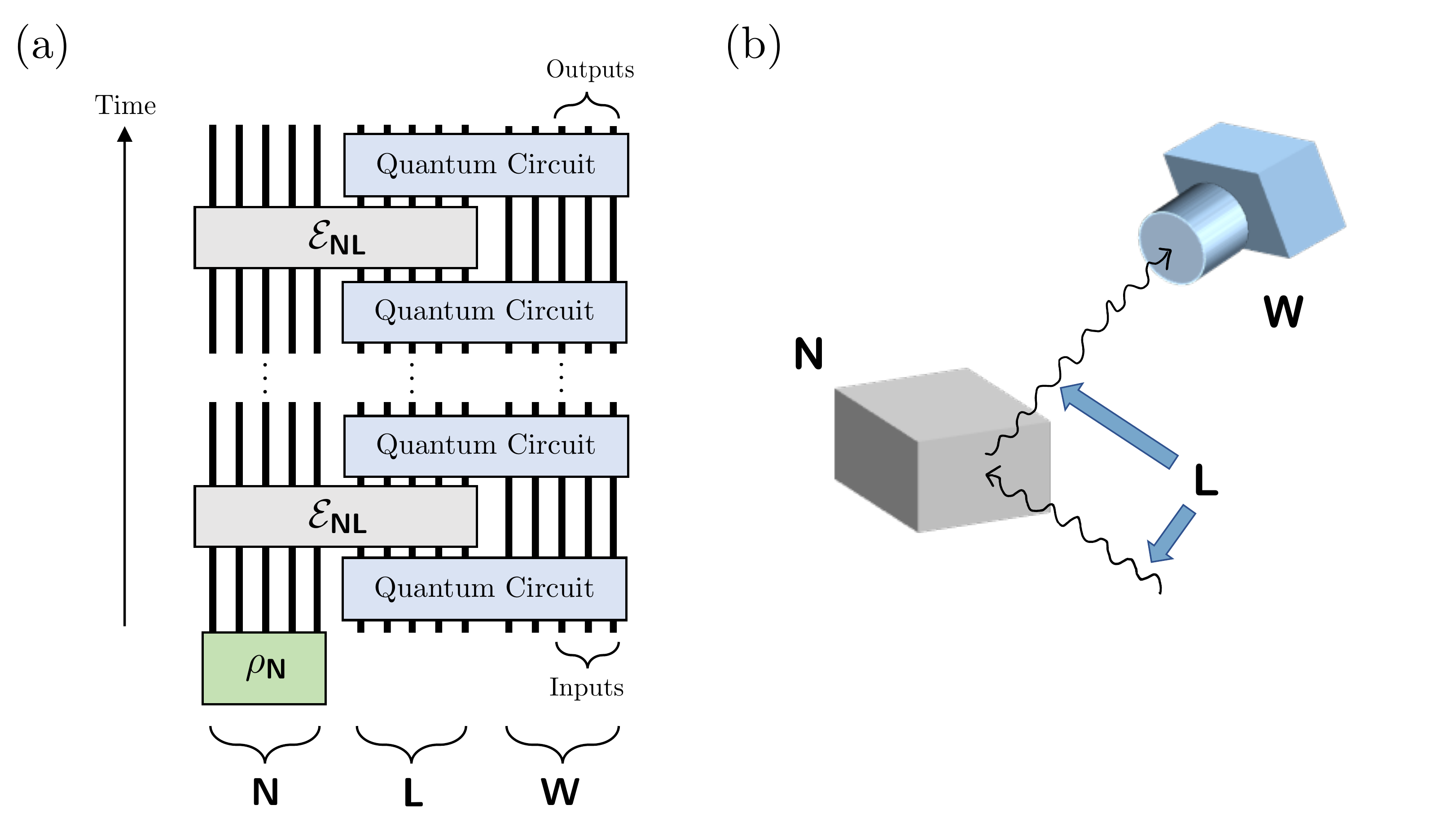}
    \caption{(a) Schematic illustrating the structure of a QUALM as an interaction between Nature and the experimentalist's controlled degrees of freedom.  Here $\fontH{N}$ represents the `Nature' register, $\fontH{L}$ is the `lab' register, and $\fontH{W}$ is the `working space' register. The experimentalist does not have direct measurement access to the $\fontH{N}$ register which should be thought of as the ``hidden'' degrees of freedom of the physical system on which the experiment is conducted.  The initial state on $\fontH{N}$ is $\rho_{\fontH{N}}$, and the input and output subsets of $\fontH{W}$ are specified.  Each thick wire corresponds to a `doubled wire' representing both a Hilbert space and its dual, since  the input is a density matrix.  For further explanation see Appendix ~\ref{App:diagrams}. (b) Illustration of an X-ray diffraction experiment, where $\fontH{N}$ is the crystal sample, $\fontH{L}$ consists of the X-ray photons (including the incoming and outgoing ones), and \fontH{W} contains the camera and other lab equipment for taking and processing the image. }\label{fig:interaction}
\end{center}
\end{figure}


\subsection{Exponential advantage from quantum coherence} \label{sec:results}


In physics experiments, it is clear that the complexity/difficulty of a task depends on the experimental setup. For example, measuring the size of a vortex in a two-dimensional superconductor is much easier for a scanning tunneling microscope (STM) than for a standard charge transport experiment, because the former has a much higher spatial resolution. The framework of QUALMs allows us to define the capability of an experiment, and study the complexity of a task with respect to that. 

In this work, we are interested in particular in the role of quantum entanglement and quantum correlations in determining the complexity of QUALMs. In many physics experiments, the experimental setup is not capable of introducing or manipulating quantum entanglement between the lab system $\fontH{L}$ and the working space $\fontH{W}$. For example, in a standard X-ray diffraction experiment, the photon source does not have the ability to create EPR pairs of photons. Even if an outgoing photon might in general be in a quantum entangled state with the crystal sample in $\fontH{N}$, the outgoing photons are measured by a camera, after which we gain classical information about the photons' locations, and all other information about their quantum coherence is gone. By contrast, one can consider a more advanced experimental setup in which photons can be prepared in EPR pairs, and after one member of a pair interacts with the sample in $\fontH{N}$, further quantum gates can be applied jointly to the pair. Perhaps the photons can be sent back to interact with the sample again, without losing coherence. One would expect that the second setup has an advantage over the first one, at least for some tasks. Most traditional experiments in condensed matter and atomic physics are like the first case, where only classical information is extracted from the physical system being studied, and measurements may destroy the quantum state. As described in the introduction, modern experimental techniques have enabled more and more examples of the second kind.

Motivated by these physics considerations, in this work we compare the following two types of QUALMs, describing experimental setups with different access capabilities. The first type allows the application of arbitrary quantum circuits to $\fontH{L}$ and $\fontH{W}$, which will therefore be able to create and manipulate quantum entanglement between them. We refer to such QUALMs as {\it coherent access} QUALMs, or simply, general or  unrestricted QUALMs. The second type only allows local operations and classical communication (LOCC) between the lab system $\fontH{L}$ and the working space $\fontH{W}$. In addition, each time after applying the lab oracle, the state of $\fontH{L}$ is measured such that no \textit{coherent} quantum information can survive and be acted upon by the lab oracle again. We refer to these QUALMs as {\it incoherent access} QUALMs, 
or in short, incoherent QUALMs.

It seems to be very natural to ask: how advantageous is the second type of experiments 
(coherent QUALMs) over the first (incoherent) type, which seems to include all conventional quantum experiments performed today? What do we gain by allowing coherent access to the physical system that we are trying to study?  

 Our main technical contribution is to provide two physically motivated problems, for which coherent access QUALMs can be {\it exponentially} more efficient than their incoherent counterparts, even if the incoherent access QUALM is allowed to employ {\it adaptive} strategies. 
 
We start by considering the following problem. We informally formulate the problem below, to be made precise later:

\begin{definition}[\bf The fixed unitary problem] (Roughly) Consider two lab oracles $\LOQ_\ell$ and $\LOP_\ell$, corresponding to two physical systems.  The first lab oracle $\LOQ_\ell$ picks a random unitary, remembers it (forever), and then subsequently applies that \textit{same} unitary to $\fontH{L}$ each time the oracle is called. ($\ell$ labels the number of qubits in $\fontH{L}$.) By contrast, the second lab oracle $\LOP_\ell$ applies a new random unitary to $\fontH{L}$ each time the oracle is called.  The goal is to distinguish between $\LOQ_\ell$ and $\LOP_\ell$ with success probability at least $0.5+\epsilon$ for some constant $\epsilon>0$.
\end{definition} 
We formally describe the aforementioned lab oracles in Section~\ref{sec:mainresult}. Physically, this problem is a toy model for distinguishing a Floquet system from a  quantum system governed by a Hamiltonian whose time dependence is generic. The time evolution operator in the former has a discrete time translation symmetry, while the time evolution operator of the latter has no time translation symmetry.  We have that $\LOQ_\ell$ implements evolution by a Floquet system with period $T$, which is only accessed every time $T$. On the other hand, $\LOP_\ell$ implements evolution by a random quantum circuit where each time evolution step is enacted by an independent Haar random unitary. Though the respective superoperators of $\LOQ_\ell$ and $\LOP_\ell$ only seem to apply unitaries on $\fontH{L}$, the Nature system $\fontH{N}$ is needed in order to `remember' the unitary applied, or pick a new (random) unitary each time.

There is a very simple coherent access QUALM that distinguishes between $\LOQ_\ell$ and $\LOP_\ell$. We first prepare an entangled state between $\fontH{L}$ and $\fontH{W}$, namely the maximally entangled state\footnote{In fact, the QUALM we present for the task (in  the proof of Theorem \ref{thm:mainswap}), does not require the preparation of a maximally entangled state but starts simply with the all $0$ state (and hence is even simpler to realize in the lab); however one of our follow-up results, Theorem~\ref{thm:symmetryswap}, does require it.} 
\begin{equation}
    |\Psi\rangle = \frac{1}{\sqrt{2^\ell}}\sum_{x \in \{0,1\}^\ell} |x\rangle_{\fontH{L}} \otimes |x\rangle_{\fontH{W}}\,.
\end{equation}
We only need to utilize $\ell$ of the qubits in the $\fontH{W}$ subsystem here.  We first call the lab oracle to act on the $\fontH{L}$ part of $|\Psi\rangle$, giving us $(U_1 \otimes \mathds{1}) |\Psi\rangle$.  Then we swap wires so that the $\fontH{L}$ part of the state is now in $\fontH{W}$, and the $\fontH{W}$ part of the state is now in $\fontH{L}$.  We subsequently call the lab oracle again, and it acts on the $\fontH{L}$ subsystem, giving us
\begin{equation}
    \frac{1}{\sqrt{2^\ell}}\sum_{x \in \{0,1\}^\ell} U_2|x\rangle_{\fontH{L}} \otimes U_1|x\rangle_{\fontH{W}}\,.
\end{equation}
Then we implement a SWAP test to essentially calculate $|\text{tr}(U_1^\dagger U_2)|^2$.  If $U_1 = U_2$, then the value of this trace is large, the lab oracle must have been $\LOQ_\ell$.  On the other hand, if $|\text{tr}(U_1^\dagger U_2)|^2$ is near zero, then $U_1 \not = U_2$ and the lab oracle was $\LOP_\ell$.  Thus, the QUALM complexity in the entangled access model is $\mathcal{O}(\ell)$ -- see Theorem~\ref{thm:mainswap}.

It is technically much harder to show that the QUALM complexity in the case of incoherent QUALMs must be exponential. This is given by our first main technical theorem: 
\begin{thm}[\bf Exponential lower bound for incoherent adaptive QUALMs for the fixed unitary problem]\label{thm:main1} (Roughly)
For any incoherent access QUALM for the ``fixed unitary problem'' on $\ell$ qubits (i.e., $\fontH{L}$ has $\ell$ qubits), its QUALM complexity is lower bounded by an exponential in $\ell$. 
 This holds even if the QUALM is allowed to employ adaptive  strategies.
\end{thm} 

We note that the upper bound on the QUALM complexity in the coherent access case is very simple; it is the lower bound in the incoherent access case that requires the technical effort. Moreover, we find a fairly straightforward proof of an exponential lower bound for the incoherent access case if we do not insist on allowing the protocol to be adaptive. Once adaptive protocols are considered, the argument becomes far more complicated. We use Weingarten functions and a substantial dose of combinatorics to arrive at our desired result. Physically, our result is consistent with the intuition that chaotic Hamiltonian evolution or Floquet dynamics (for which $\LOQ_\ell$ is a toy example) can emulate stochastic dynamics (for which $\LOP_\ell$ is a toy example), even if the former is deterministic and time-translation invariant.  Details of the setup of this problem, and the proof, are presented in Section~\ref{sec:mainresult}.

The above bounds constitute the first example of a coherent measurement protocol for a  physically motivated experiment, which is provably exponentially more efficient than any incoherent experiment.  Furthermore, this turns out to be true even if the coherent access QUALM is extremely simple, such as the SWAP test; and even if the incoherent QUALM is allowed to be adaptive.  

Two immediate corollaries of the above Theorem \ref{thm:main1} follow. One can consider the \textit{fixed state problem}, in which one is asked to distinguish between the following two lab oracles: the first picks a Haar random $\ell$ qubit state and generates that same state on the $\fontH{L}$ subsystem each time the oracle is called, while the other oracle generates on $\fontH{L}$ a newly picked Haar random state each time the oracle is called.  We have:
 
\begin{corollary}[\bf Exponential lower bound for the fixed state problem]
The QUALM complexity of the fixed state problem is exponential in $\ell$ in the unentangled adaptive access setting. (The QUALM complexity of this 
problem in the 
entangled access case is at most linear in $\ell$.)
\end{corollary}

A related problem, which we call the {\it non-unitarity problem}, is along the following lines. One is asked to distinguish between two lab oracles: the first applies a (fixed) Haar random $\ell$ qubit unitary to the $\fontH{L}$ system, while the other applies a completely depolarizing channel to $\fontH{L}$.  We have the corollary:

\begin{corollary}[\bf Exponential lower bound for the non-unitarity problem]
The QUALM complexity of the non-unitarity problem is exponential in $\ell$ in the incoherent adaptive access setting. (The QUALM complexity of this problem in the coherent access case is at most linear in $\ell$.) 
\end{corollary}

We next turn to another physically-motivated task:
\begin{definition} 
[\bf The symmetry distinction problem]
Let our quantum lab oracle be one of the following: either, it applies a Haar random {\bf unitary} to the $\fontH{L}$ system every time it is called, it applies a Haar random {\bf orthogonal} matrix to $\fontH{L}$ every time it is called, or it applies a Haar random {\bf symplectic} matrix to $\fontH{L}$ every time it is called. In all three cases, once the random matrix is chosen, it is fixed, and the same matrix is applied in every call of the oracle. The problem is to distinguish which of the three oracles we have, with a non-negligible probability of success.
\end{definition}

If one is allowed coherent access, then one can use the SWAP test, with a little more sophistication, to determine the oracle; this requires at most a linear number of gates.  Extending the techniques used in the proof of Theorem \ref{thm:main1}, we prove that for the symmetry distinction problem the incoherent access QUALM complexity is at least exponential, even if adaptive strategies are allowed. 
\begin{thm}[\bf Exponential lower bound for incoherent adaptive QUALMs for the symmetry distinction  problem] \label{thm:main2rough} (Roughly) For any incoherent access QUALM for the symmetry distinction problem on $\ell$ qubits, its QUALM complexity is lower bounded by an exponential in $\ell$. This holds even if the QUALM is allowed to be adaptive.
\end{thm}

Symmetries are essential features of quantum many-body systems, and their measurement is a central desiderata of experimental physics.   Our result establishes, in a prototype setting, that in some cases such symmetries may be discovered with the aid of entanglement using exponentially fewer resources than conventional methods allow in contemporary experiments.

The above result is of potential experimental interest.
In fact, after the first announcement of the results reported in the current paper, there has been experimental follow-up work demonstrating a version of the symmetry distinction problem on the Sycamore processor (see the announcement of the experimental results in the talk~\cite{HuangTalk}).

{~}
\\
\begin{remark}[\textit{Unitary designs}]
We note that in all the above results, we consider Haar random unitaries and states 
for $\ell$ qubits. These are in general  non-physical for large $\ell$. 
To bring our results somewhat 
closer to physics, we note that  
our results also apply if we consider, 
instead of Haar random unitaries, 
random circuits of sufficiently large polynomial depth. The idea is that 
such random quantum circuits on $\ell$ qubits, of depth $t$, become what is known as approximate unitary $k$-designs for $t$ polynomial in $\ell,k$ \cite{HarrowMahraban}.  
Such distributions over unitaries can replace the Haar measure to approximate, with arbitrarily small inverse exponential error, integrals of the form \begin{equation}
\int_{\text{Haar}} dU U^{\otimes k} A (U^\dagger)^{\otimes k}    
\end{equation} which is precisely the form of the equations on which our proofs rely (see, e.g., Eqn.~\eqref{eq:Qk}).  
If our unitaries are drawn from a unitary $k$-designs, for
$k$ the number of lab oracle calls 
(which we assume to be polynomial), then all our results follow, since all 
calculations only incur an exponentially small error.
This brings our results one step closer to physical reality; see Subsection \ref{sec:discussion} for more on this point.
\end{remark}

\subsection{Comparison with related results}
\label{subsec:comparison}
We next put our results in the context of known results in the literature; to this end, we consider known quantum algorithmic advantages and other related results, stated using the QUALM terminology, so that we can compare them with 
our results. It is useful to consider for this matter Table \ref{table:compare}. The columns of the table correspond to 
increasing ``levels of 
coherence''  in  QUALMs. We do not define those rigorously in this paper,
since they are not directly related to our results, 
but loosely speaking, we can consider the following access types: 

\begin{center}
\begin{table}[]
$$$$

\scriptsize
\begin{tabular}{|c|l|l|l|l|l|}
\hline
\textbf{Oracle/} & \textbf{Binary} & \textbf{Local-local} & \textbf{Single register} &\textbf{Incoherent}&
\textbf{Coherent Access} \\
\textbf{access}&&&&\textbf{access}&\\

\hline
&&&&& \\
   &  &  & Dihedral, Affine && HSP for general groups  \cite{EttingerHoyer} \\
 &  & Simon's & and Heisenberg HSP \cite{sen} & & \underline{Conjectured} exponential \underline{query} \\
 &  & algorithm$^*$ \cite{simon} & \underline{Conjectured} exponential \underline{query} & & advantage over binary\\
 &  &  & advantage over binary &&     \\
 \textbf{Classical} &  &  &  & \qquad \,\,\,? &   \\
 & & & & &                   Exponential advantage \\
                    &  &  &  && over (\textit{adaptive}) incoherent \\
                    &  &  &  && access from generalized Simon's \\
                    & & & && task~\cite{Chia1} and Welded Tree problem~\cite{Coudron1}  \\
             &&&&& \\       
                    
\hline
&&&&& \\
                   &  &  &  && State tomography: \\
                   &  &  &  && \underline{quadratic} \underline{query} advantage \\
                   &  &  &  && over single register \cite{haahsample,wrightsample} \\
                   &  &  &  & & \\ 
                   &  &  &   & & Exponential \underline{query} advantage in \\
                   &  &  &  && distinguishing coset states \\
   &  &  &  & & for the dihedral group \\
                   &  &  & & & over (\textit{non-adapative})  \\
                   \textbf{Quantum} &  & \qquad \,\, ?  & \quad \qquad \qquad \,\,\,\, ? & \qquad \,\, ? & single register \cite{baconHSP} \\                        &  &  & &  &   \\ & & & && Exponential \underline{query} advantage for \\
                   &  &  &  && tomographic protocols\cite{PreskillNew1} \\ & & & && \\ & & & &&
                    \textbf{This work}\text{*}: Exponential \\
                    &  &  &  && advantage over (\textit{adaptive}) \\
                    &  &  &  && incoherent access for \\
                    & & & && experimental tasks \\
&&&&& \\
\hline
\end{tabular}
$$$$
\caption{Comparison to known results with different access types to the lab oracles. The meaning of the rows and columns is explained in the text.
 }\label{table:compare} 
\end{table}
\end{center}
 
 \begin{itemize}
     \item {\bf Binary}. $\fontH{L}$ consists of $\ell$ qubits, and the only states we are allowed to prepare for the lab oracle to act on are standard computational basis states.  Once a lab oracle acts on a product state, we are only allowed to measure the result in a product basis.  Note that pre- and post-processing can entail full quantum computation. 
     \item {\bf Local-local}. $\fontH{L}$ consists of $\ell$ qubits, and the only states we are allowed to prepare for the lab oracle to act on are product states.  Once a lab oracle acts on a product state, we are only allowed to measure the result in a product basis.  Note that pre- and post-processing can entail full quantum computation. 
     \item {\bf Single register}. The QUALM allows generic gates to be applied to $\fontH{L}$ so that the qubits in $\fontH{L}$ can become mutually entangled, but $\fontH{L}$ and $\fontH{W}$ only interact via LOCC's and so they cannot become entangled with one another. 
     In addition, the applied LOCC's may not depend on previous measurement results; in other words, the QUALM is non-adaptive. In Table \ref{table:compare} we consider examples of single register access in which the $\fontH{N}$ register is empty. In this case, 
     one can think of the single register model as if the lab oracles queries are all applied in parallel, and the resulting measurement outcomes, namely the output classical information, is jointly processed thereafter.
     \item {\bf Incoherent access}. We have discussed this type in the previous subsection. Compared with the single register setting, the main difference is that adaptive strategy is allowed. There is still no quantum coherence between different applications of the lab oracle, but the LOCC applied can depend on all previous measurement results.
     \item {\bf Coherent access}. In this model, a general QUALM is allowed, and there are no restrictions. 
 \end{itemize}
 
A natural first example to compare our results to is that of 
 Simon's algorithm\cite{simon} which is a famous quantum algorithm in the oracle setting, that achieves an exponential advantage over any classical algorithm with access to the same oracle. Wouldn't such a result imply an exponential separation between coherent and incoherent QUALMs? Perhaps counter-intuitively however, 
 when stated in the QUALM language, Simon's algorithm is in fact an example of {\it local-local} access QUALM. Namely, it lies low in the ``coherence hierarchy'', and does not provide 
 an example for a QUALM complexity separation between coherent and incoherent QUALMs. 
 This also clarifies that the exponential advantage we prove is fundamentally different from that discussed in Simon's algorithm.

Another important algorithm to compare is the one which solves the Hidden subgroup problem \cite{EttingerHoyer}; this protocol is placed in the coherent column since it uses {\it coherent} access to the registers; it requires only polynomially many queries in contrast to the best classical algorithm which requires exponentially many queries. However, the exponential advantage in terms of query complexity is only conjectured here, since no known lower bounds on more restricted access types are known.
Moreover, and very importantly for the focus of this work, the quantum protocols suggested are not known to be efficient in terms of {\it gate} complexity and thus also not in terms of QUALM complexity; realizing them by efficient quantum circuits would imply an efficient quantum algorithm for these problems, resolving a quarter-century-old open problem. 

For context, we also added the related algorithm of \cite{sen}, which 
appears in the \textit{single register}
column. This result provides a protocol that does not require coherence between the registers, and hence lies lower in the coherence hierarchy. 
It holds only for a restricted set of groups. Like \cite{EttingerHoyer}, this protocol is not known to be gate efficient and it has not been proven to have an exponential query complexity advantage over classical algorithms.

Interestingly, two recent works~\cite{Chia1} and~\cite{Coudron1} from a seemingly unrelated area of quantum computation, and a different motivation, do provide exponential separations between coherent and incoherent, adaptive access in terms of both gate and query complexity. These results are presented in a very different language, as their original goal was to prove a separation in {\it depth} hierarchies of computational models. Some work is required to reformulate the results within the QUALM framework, and derive their implication of such an exponential separation between the two access types.  The problems at hand are respectively a recursive generalization of Simon's problem, and the welded tree problem; these problems (which are classical) are somewhat contrived, and are not natural from the point of view of applications in experimental physics.

The above examples all appear in the row called ``classical'', indicating that 
the input to the problem is a classical oracle.
Of more relevance to this work, is the row ``quantum'' in which the input is a truly 
quantum system, which brings us closer to the experimental setting. 

The first example in the quantum row is the quadratic sampling advantage of state tomography protocols when using entanglement \cite{haahsample,wrightsample}; we note that this improvement is merely quadratic in query complexity. Importantly, the advantage is not known to be achievable using a gate-efficient protocol (also, the lower bound proof does not hold for adaptive single register access; in fact, 
it seems that adaptive protocols 
do perform better than non-adaptive ones in this case \cite{haahsample}).

An example tightly related to our contribution is that of~\cite{baconHSP}. 
This work  provides a proven exponential advantage of the coherent versus the single register access, for a quantum state distinction problem emerging from the dihedral HSP. However, very importantly, the entangled protocol suggested has {\it exponential} gate complexity, and so the exponential advantage is restricted to query complexity and is not known to hold in QUALM complexity.  In addition, the advantage is known to hold only in the restricted case in which the  single register access is non-adaptive.

Of strong relevance to our work is the recent work of Huang, Kueng and Preskill~\cite{PreskillNew1} who independently addressed a similar question to ours, and compared coherent versus incoherent access to the physical system being studied, in the context of a quantum machine learning related task. The work of~ \cite{PreskillNew1} provides a tomography task in which there is an exponential query advantage of coherent versus (adaptive) incoherent access. However, the coherent protocol of~\cite{PreskillNew1} has exponential gate complexity, and thus its advantage is restricted to query complexity and does not extend to QUALM complexity (much like that of ~\cite{baconHSP}).

Finally, it is natural to ask whether it is possible to massage the well-known 
exponential advantages~\cite{raz,baryosef,gavinsky1,Regev1,Raz2,fingerprinting} of quantum communication complexity\cite{Yao93, deWolf1} into an exponential separation between coherent and incoherent access QUALMs, as the LOCC restriction on incoherent access QUALMs seems to capture exactly the difference between quantum and classical communication. Indeed, communication complexity can be embedded in the QUALM framework (see Subsection \ref{subsec:ExampleQUALMs}).  However,
communication complexity only counts the number of qubits sent from one party to another; by contrast, the full QUALM complexity additionally takes into account the complexity of pre- and post-processing those qubits, and performing computation on them (as needed to achieve the communication task).  More concretely, exponential advantages in the communication complexity setting occur when a classical communication task that requires $k$ communicated bits only requires $\mathcal{O}(\log(k))$ bits (or qubits) in the quantum setting.  In the QUALM framework, such a savings will have little effect on the QUALM complexity, since the gate complexity is already at least the size of the input $k$. Thus, to the best of our knowledge, the exponential quantum advantages in the communication complexity model do not provide examples of exponential separations in QUALM complexity between coherent and incoherent access.

Our work is the first to demonstrate an exponential advantage in QUALM complexity between coherent and incoherent access QUALMs for experimentally motivated tasks; the exponential advantage holds even when the incoherent QUALM is allowed to be adaptive. 
Moreover, this exponential advantage is achieved using an extremely simple coherent access QUALM, based solely on the SWAP test, and thus gives rise to the hope that it, 
or a variant thereof, can be implemented 
experimentally.

\subsection{Related models}\label{sec:generaltasks}
QUALMs generalize quantum circuits to the experimental setting; they are a hybrid between quantum circuits with
access to classical (or quantum) oracles \cite{nielsen}, and quantum interactive 
protocols \cite{KitaevWatrous}. 
The additional Nature register $\fontH{N}$ enables modeling certain essential characteristics of physical experiments: in particular, the
indirect access to the ``input'' physical system to be measured, their hidden ``memory'' between different applications, and the nature of the 
interaction between the experimental apparatus and the physical system. Formally, QUALMs are tightly related to quantum oracle algorithms, as well as to quantum interactive protocols (and variants thereoff such as quantum  strategies \cite{watrous} and quantum combs \cite{comb1}), and we are hopeful that results from these areas will be useful for the design and analysis of QUALMs. It is likely that methods and techniques from the area of quantum communication complexity will also be useful, due to the relations mentioned in 
Subsection \ref{subsec:comparison}. 

The notion of QUALMs we introduce here is related to another notion which has been recently studied fairly intensively.  In our framework, we are interested in modeling a {\it physical experiment} resulting in a {\it classical} outcome of a measurement. This can be viewed as a function, computed by the QUALM, whose input is the lab oracle and a bit string of classical data from the lab, and 
whose output is a
bit string (see Definition \ref{def:task}). 
In fact, as we describe in Subsection \ref{sec:qualms}, 
both input and output registers can be generalized to be quantum. 
Motivated by quantum algorithmic questions, 
recently researchers took interest in problems which can be cast in a tightly related model: a quantum circuit which
makes 
use of a quantum oracle by black box access, gets as an input a {\it quantum state}, and outputs another {\it quantum state}. Such a process generates a {\it quantum channel} which depends on (or is a function of) the quantum oracle being used as a black box (the work of  \cite{Gavorova20} formally defined this notion as a $\Task$). This is similar to the notion of task we use in our QUALM framework, except without the interactive ingredient; thus, the Nature register $\fontH{N}$ is missing in this setting.   Several works investigated in this model how oracle access to a unitary $U$ can be used to implement various channels which are functions of $U$; for instance, applying controlled-$U$ \cite{TGMV13,AFCB14, DNSM19,Gavorova20}, $U^*$ \cite{MSM19}, $U^T$ and $U^{-1}$ \cite{QDS+19a,QDS+19b}, raising $U$ to some fractional power \cite{SMM09, GSLW19}, and more. Ideas from this direction could be relevant for QUALM purposes.
 
 We mention that QUALMs were in part inspired by the work~\cite{CJQW1}, which developed quantum information-theoretic measures of spacetime correlations. (See also related work on quantum combs~\cite{comb1} and quantum process tensors~\cite{process1}.) For a given physical system, if we are only allowed to probe a spatial region $A$ at a given time, the results of all possible measurements are determined by the reduced density operator on $A$ at that time. If we are instead allowed to access the physical system at several different spacetime regions, which may be timelike or spacelike separated, the correlation functions cannot be characterized by an ordinary density operator. Instead,~\cite{CJQW1} considered ancilla systems which couple to the physical system at the specified spacetime regions. By an appropriate choice of the ancillas and their couplings to the physical system, all correlation functions for the specified spacetime regions can be determined by the density operator of the \textit{ancilla system}, which is called the ``superdensity operator'' of the original system for the specified spacetime regions.
 From the QUALM point of view, the superdensity operator can be viewed as the output of a particular type of QUALM. The physical system being studied corresponds to $\fontH{N}\otimes\fontH{L}$, with its dynamics specified by the lab oracle superoperator. The lab system $\fontH{L}$ contains all spacetime regions that are accessible to the experimentalist. The ancillas correspond to $\fontH{W}$, which is coupled to $\fontH{L}$.  Each coupling between $\fontH{L}$ and $\fontH{W}$ utilizes a new (i.e., previously unused) subsystem of $\fontH{W}$
 and otherwise there are no gates coupling different subsystems of $\fontH{W}$. The superdensity operator of the $\fontH{N}\otimes\fontH{L}$ system is simply the final density operator of $\fontH{W}$. The superdensity operator formalism can be usefully applied to analyze a variety of different tasks. For example, Ref.~\cite{Causal1} studied the causal influence between different space-time regions with superdensity operators, which can be viewed as a task achieved by this particular type of QUALM.

\subsection{Comments and open questions}\label{sec:discussion} 
 
Our motivation for this study is physics. We believe that QUALMs lay the groundwork for a theory of general quantum experiments and measurements, and provide a framework to study their complexity.  Further, an understanding of QUALMs enables the design of novel and more efficient experiments; this is of particular importance in the present era of quantum computation in which many new experimental tools are about to become available. Conceptually, we hope the tools we introduce in this paper aid in the understanding of various types of QUALMs, and in the development of more sophisticated QUALMs, to enable better precision, more savings in resources, and new insights into the design of physical quantum experiments. 

In this work, we provide evidence that entanglement 
and coherence can be truly {\it exponentially} advantageous when performing {\it experiments in the lab}. 
We believe this provides a second, very important motivation for the physical realization 
of quantum algorithmic components.

An important problem is of course to come up with further advantages of different QUALMs, for problems of physical interest.
We could for example ask a more sophisticated question like\footnote{We thank Subir Sachdev for suggesting this.}: could there exist an efficient QUALM which  detects the topological phase of a system?  For instance, is it possible to define an ensemble of ``random'' topological phases which are hard to distinguish by adaptive, sequential single register measurements, but which can be distinguished efficiently using a QUALM (presumably, once which is allowed coherent access)? 

It would be very interesting to come up with coherent access QUALMs which achieve exponential (or even polynomial) advantage, by leveraging coherence in more sophisticated ways than just the SWAP test and its variants. What other types of quantum computation tools are available in this context? Likely, inspiration from quantum algorithms will be helpful.

This work is of course motivated 
by experimental physics. In particular, the protocols which we present here for 
distinguishing a time-dependent random Hamiltonian from a fixed Hamiltonian, detecting non-unitarity, as well as the problem of detecting the symmetry of the physical system, are physically motivated questions; it is thus of course of great interest to see the coherent experimental setups suggested here demonstrated experimentally. 

A very important question in this context is that of noise 
resilience. 
It is not hard to see that our protocols, which rely on the SWAP test, lose their quantum advantage in the presence of local independent noise with any constant probability $p>0$ per qubit. One might suggest that a noise resilient example providing exponential advantage for coherent access QUALMs can be designed, using quantum fault tolerance~\cite{aharonovFT,KitaevFT, knillFT}. 
Indeed, if our entire protocol as well as the Haar random unitaries in the lab oracles are encoded fault tolerantly, and the error rate is below the fault tolerance threshold, then one arrives at an example in which an exponential advantage of coherent access QUALMs persists even in the presence of noise.  However, this leads to an extremely non-physical choice of the lab oracles; moreover, the QUALMs themselves cannot be assumed to be fault tolerant in the NISQ era \cite{PreskillNISQ}.  It is an important open question to clarify whether similar advantages can be exhibited, in a scalable, noise resilient way.

A related important question arises in the context of physically realizing or demonstrating such 
experiments. 
The examples of 
exponential advantages of coherent access provided in this work are achieved in a toy setting: 
the $U$'s considered are $n$-qubit Haar random, and thus highly non-local; for large $n$'s they are not realizable
(even if we consider the unitary design version as discussed at the end of Subsection \ref{sec:results}). 
Can our results be strengthened 
so that the lab oracles are closer to physical reality?

We note that the natural first attempt, namely characterizing $k$-local time-independent Hamiltonians for $k=\mathcal{O}(1)$ to some given accuracy $\epsilon$, is already fairly well-understood: given access to such a Hamiltonian as a lab oracle, the description of the Hamiltonian can be measured to within arbitrary accuracy very efficiently by local-local QUALMs (by applying Pauli measurements) using polynomially many samples and local measurements (see, e.g.~\cite{Li1}).  
In this case, there does not seem to be a significant advantage to using QUALMs of more coherent access types.  
A natural next step concerns the question of lab oracles which apply shallow (i.e.~constant depth) quantum circuits. 
At first sight, it might seem impossible to achieve efficiently by a QUALM (see the impossibility results of \cite{arunashalam}, as well as the sample-efficient entangled {\it and} adaptive protocols for a related task). 
Indeed, after the first version of this
paper was made public, a follow up work~\cite{ChenNew1} managed 
to resolve the problem and demonstrate 
the same exponential advantage, but 
when the lab oracle is defined using shallow quantum circuits and are thus efficiently implementable in the lab.  Furthermore,
initial demonstrations of variants of these protocols were recently performed (see the announcement~\cite{HuangTalk}).
Of course, further experimental demonstrations of these and other QUALM complexity advantages would be immensely 
interesting. 

We note that one might hope to massage 
our results in a different way, so that the unitaries defining the lab oracles in our examples become efficiently implementable. 
The replacement of the Haar random unitaries in our examples by unitary designs, as suggested in the remark in Subsection \ref{sec:results}, indeed makes the setting more realistic.  However, it has a drawback: once we fix the distribution over unitaries to be a unitary $k$-design for a particular $k$, we can always increase the number of queries $k'$ to be larger than $k$ (but still polynomial), in which case the approximation of the Haar measure breaks down and our impossibility proofs falls apart.
One could hope to instead use {\it pseudo-random unitaries}, defined in \cite{Ji}.  Pseudo-random unitaries are both practical to construct, and under certain computational assumptions are indistinguishable from Haar random unitaries by any polynomial circuit. 
Unfortunately, the constructions of efficient pseudo-random unitaries, suggested in \cite{Ji}, are not yet proven to be secure.
We note that any proofs in this setting would rely on the computational assumptions inherent to pseudo-randomness, and further would have to additionally impose the restriction that the gate complexity of the QUALM is polynomial (in our proofs, we only lower bound the {\it query} complexity). 


It is interesting to clarify the relations between the other types of access for QUALMs (i.e., comparing between other columns of Table~\ref{table:compare}). Is there a provable exponential (or even polynomial) advantage of single register access over local-local access? Of particular interest is whether there is an exponential advantage of incoherent access over single register access; this captures the question of the power of adaptiveness in incoherent experiments; In other words we ask: does adaptiveness help in the incoherent setting (see \cite{HarrowAdaptive1,Sev1})?
Relatedly, a problem which we leave open is: what is the power of the incoherent access model when we remove the requirement of complete measurements in between lab queries? Is the resulting, modified model still exponentially weaker than the coherent access model? 

We mention an interesting advance in the recent follow up work of \cite{ChenNew1}. 
In this work, the proof techniques in~\cite{PreskillNew1} have recently been synthesized with the proof techniques in the present paper, to arrive at a more general framework for proving exponential separations between coherent and incoherent QUALMs~\cite{ChenNew1}; this might be applicable for other QUALM separations.

Finally, in this work 
we did not specifically address the crucial question of translating the 
QUALM abstract description by quantum gates and oracle calls, to the experimental setting. As in the case of realizing quantum algorithms, this is of course a mounting challenge much beyond the scope of this work. In particular we highlight here one element in QUALMs which allow interacting between quantum physical systems to be measured, and qubits in a quantum computer; or alternatively, which allow applying quantum gates on the physical system which is being studied. 
Such components require quantum information to be transmitted between the degrees of freedom of the physical system being studied, into a form such that enables them to interact with the degrees of freedom of the ancilla qubits in the work space.  
Abstractly, the capabilities for doing 
this in the lab are 
modeled in the QUALM framework using  the admissable set of quantum gates; 
experimentally, realizing such a transmission could be an immensely  
challenging problem, related to the problems appearing in the context of {\it quantum networks}.

We hope that this work deepens our understanding of how quantum computing methods can enable more sophisticated and precise quantum experiments. 
There are many problems left open; much more work is needed in order to improve our understanding of the limitations and possibilities of 
quantum algorithmic measurements.

\subsection{Organization} 
The rest of the supplementary information is organized as follows. 
In Section~\ref{sec:def} we provide the definitions of the main players: tasks, lab oracles, QUALMs, and QUALM complexity. We also provide a few examples of how various familiar tasks can be cast in the QUALM framework. In Section~\ref{sec:mainresult} we introduce the fixed unitary problem and provide both a coherent QUALM for it (Theorem \ref{thm:mainswap}) as well as a proof that the QUALM complexity for any incoherent QUALM that solves it is exponential (Theorem \ref{thm:main}); Section~\ref{sec:corrs1} provides several corollaries and extensions of the previous section.  In Section~\ref{sec:symmtask} we introduce the
symmetry distinction problem, provide the coherent QUALM for it (Theorem \ref{thm:symmetryswap}) and prove an exponential lower bound for incoherent QUALMs for the task (Theorem \ref{thm:main2rough}). 
The Appendices contain some basic introduction to 
the diagrams we use throughout, as well as the necessary background on Haar integrals.

\section{QUALMs and Lab Oracles: Definitions}\label{sec:def}
\subsection{Notation} 
Throughout the paper, we denote subsystems (which will all consist of finitely many qubits) by $\fontH{N}, \fontH{L}, \fontH{W},...$, and we use the same notation to denote their associated Hilbert spaces.  The number of qubits in each subsystem will be denoted by lowercase letters $n,\ell,w$,...~and the set of states (density matrices) on each subsystem will be denoted by $\mathcal{D}(\fontH{N}), \mathcal{D}(\fontH{L}), \mathcal{D}(\fontH{W}),...$; we similarly denote classical probability distributions on $\{0,1\}^k$ by $\mathcal{D}(\{0,1\}^k)$.  We will often refer to the union of two or three non-intersecting registers; in this case we will denote the resulting Hilbert space (as well as the set of qubits) by concatenating the letters corresponding to the registers, e.g. $\fontH{LW}$ (although sometimees we will write $\fontH{L} \otimes \fontH{W}$).  

Operators, superoperators, as well as sets thereof, will be denoted using a calligraphic font: $\mathcal{E}, \mathcal{Q},\mathcal{O}$ etc. 
We will often abuse notation, and refer to an ordered sequence of superoperators 
 $\mathcal{Q}=(\mathcal{Q}_1,\mathcal{Q}_2,...,\mathcal{Q}_k)$
as equal to the operator which is the result of applying the superoperators in the sequence in the given order; that is, we will refer to $\mathcal{Q}$ as the superoperator $\mathcal{Q}=\mathcal{Q}_k \circ \cdots \circ \mathcal{Q}_2 \circ \mathcal{Q}_1$. 

As we will see below, lab oracles, which are not simply operators but a pair of a superoperator and a state, will be denoted by $\LO = (\mathcal{E}_{\fontH{NL}}, \rho_{\fontH{N}})$. 

\subsection{Quantum Algorithmic Measurements}
\label{sec:qualms} 
Throughout this subsection, we will consider three Hilbert spaces: \begin{itemize}
    \item The ``Nature'' Hilbert space $\fontH{N}$ of $n$ qubits, to which the experimentalist has no direct access; 
    \item The ``lab'' Hilbert space $\fontH{L}$ of $\ell$ qubits, which consists of the degrees of freedom of the physical system which the experimentalist {\it can} access, 
    \item A third Hilbert space $\fontH{W}$ of $w$ qubits can be thought of as the ``working space'' of the experimentalist -- it corresponds to the additional degrees of freedom which the experimentalist uses during the experiment. 
    \end{itemize} 
We imagine these three Hilbert spaces as being subsystems of a total Hilbert space $\fontH{H}$\,; that is, $\fontH{H} = \fontH{N} \otimes \fontH{L} \otimes \fontH{W}$.
    
Our first definition is that of the {\it lab oracle}; it is an attempt to capture the allowed or available interactions between the experimentalist and the
{\it physical system} on which the experiment is performed.
The setup is summarized in Fig.~\ref{fig:qualm}.

\begin{figure}
    \centering
    \includegraphics[width=3.5in]{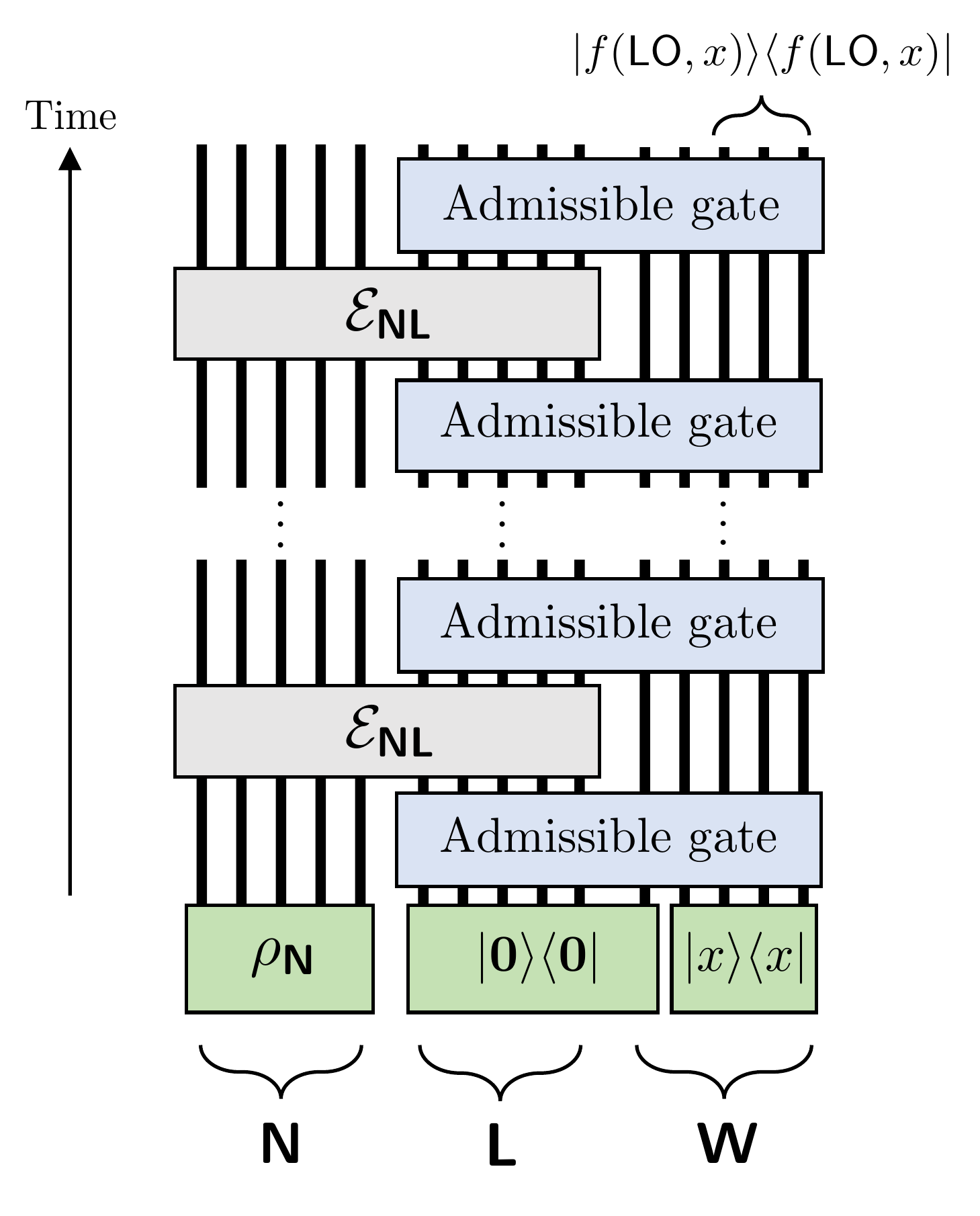}
\caption{A QUALM applied to a lab oracle, using diagrammatic notation explained in~\ref{App:diagrams}.  The green boxes represent the initial state $\rho_{\fontH{N}}\otimes \ket{\textbf{0}}\bra{\textbf{0}} \otimes |x\rangle \langle x|$ where $x$ is some bit string. The grey rectangles are the lab oracle superoperator $\mathcal{E}_{\fontH{NL}}$ as per Definition \ref{def:laboracle}. A QUALM (Definition~\ref{def:qualm}) consists of a sequence of applications of either admissible gates $\mathcal{G}_i\in\mathcal{G}$ (see Definition~\ref{def:admissiblegates}) or the lab oracle superoperator $\mathcal{E}_{\fontH{NL}}$.  The output of the QUALM is on a subsystem $\fontH{S}_{\text{out}}$\,, indicated by the bracket at the top right of the diagram. The QUALM (blue part of the diagram) and lab oracle superoperators (gray parts of the diagram) jointly comprise a quantum circuit mapping quantum and classical inputs (green parts of the diagram) to a final state on $\fontH{S}_{\text{out}}$\,. If this map produces a state $|f(\textsf{LO}, x)\rangle$ which encodes a specified function $f$ on lab oracles and initial bit strings that is targeted in a `task' (see Definition~\ref{def:task}), we say the QUALM achieves that task. }\label{fig:qualm}
\end{figure}

\begin{definition}[\bf Lab Oracle]\label{def:laboracle}
A lab oracle is specified by a pair $\LO = (\mathcal{E}_{\fontH{NL}},\rho_{\fontH{N}})$ where $\mathcal{E}_{\fontH{NL}}$ is a quantum superoperator (i.e.~a completely positive trace-preserving map) on $\fontH{N}\otimes\fontH{L}$ and $\rho_{\fontH{N}}$ is state on $\fontH{N}$.
The set of lab oracles is denoted by 
$\mathscr{L}\!\mathscr{O}(\fontH{N},\fontH{L})$. 
\end{definition}
\noindent In some circumstances it is suitable to generalize this definition so that the lab oracle state $\rho_{\fontH{N}}$ is instead an initial state jointly on $\fontH{N}$ and (a subset of) $\fontH{L}$. 

We regard the experimentalist as having command over $\fontH{L} \otimes \fontH{W}$.  The idea is that the nature subsystem $\fontH{N}$ is initialized in the state $\rho_{\fontH{N}}$, and that the only way to interrogate it is via $\mathcal{E}_{\fontH{NL}}$ which couples $\fontH{N}$ to $\fontH{L}$.  The experimentalist can then process the information accessed by this coupling by performing quantum computation on $\fontH{L} \otimes \fontH{W}$.

In the most general setting, we allow the experimentalist to perform arbitrary quantum operations on $\fontH{L}$ and $\fontH{W}$, i.e.~the experimentalist has access to a universal quantum gate set.  In other circumstances, it is useful to restrict the allowed operations to be classical, so as to compare with the fully quantum case.  There are other useful restrictions, depending on the context.  As such, it is useful to define the ``admissible gates'' allowed for the experimentalist to use on $\fontH{L} \otimes \fontH{W}$.  (As usual, we allow the gates to more generally be superoperators.)  We have the following definition:
\begin{definition}[\textbf{Admissible gates}]\label{def:admissiblegates} We denote by $\mathcal{G}$ a set of ``admissible gates'', namely a set of quantum superoperators acting on $\fontH{L}\otimes \fontH{W}$.  (Note that superoperators include measurements.)
\end{definition}

Now we define the notion of a \textit{task} which corresponds to the ``experimental problem'' to be solved: 
it describes what it is that the experimentalist wants to measure.  The experimentalist must achieve the task only by using the lab oracle superoperator, together with the operations at her disposal in her laboratory (i.e., the admissible gates on $\fontH{L}\otimes\fontH{W}$). 

\begin{definition}[\textbf{Task}] \label{def:task}
A `task' is a tuple $\Task = (\fontH{S}_{\text{in}}, \fontH{S}_{\text{out}}, f,\mathcal{G})$, associated with a given system $\fontH{N}\otimes\fontH{L}\otimes\fontH{W}$ (which is usually implicit). Here, $\fontH{S}_{\text{in}}$ is a $p$-qubit subsystem of $\fontH{W}$, $\fontH{S}_{\text{out}}$ is a $q$-qubit subsystem of $\fontH{W}$,
$f$ is 
a function
$$f : \{\LO_0, \LO_1, \LO_2,...\} \times \{0,1\}^p \longrightarrow \{0,1\}^q\,,$$
and $\mathcal{G}$ is a set of admissible gates on $\fontH{L}\otimes\fontH{W}$.  In the domain of $f$, $\{\LO_0, \LO_1, \LO_2,...\}$ is a set of lab oracles (here we denoted this set as discrete, but of course one can also consider a continuous set of lab oracles as input), i.e.~a subset of $\mathscr{L}\!\mathscr{O}(\fontH{N},\fontH{L})$
\end{definition}

This definition can be thought of as follows.  Given a set of lab oracles which represent possible arrangements of `Nature', the task is to compute the function $f$ which takes as input a lab oracle, some classical lab settings (i.e., bit strings in $\{0,1\}^p$), and outputs a classical `experimental result' (i.e., bit strings $\{0,1\}^q$).  The task is to be achieved by constructing a circuit from admissible gates in $\mathcal{G}$, in conjunction with interspersed calls to the lab oracle superoperator.

\begin{remark}[\textit{A} $\Task$ \textit{with a probabilistic output} (\textit{Generalizing Definition  \ref{def:task}})]
\label{remark:altdef1}In many situations it is more natural to talk about output 
{\it probability distributions} and define $f$ to be 
\begin{equation}\label{eq:proboutput1}
    f : \{\LO_0, \LO_1, \LO_2,...\} \times  \{0,1\}^p \longrightarrow \mathcal{D}(\{0,1\}^q)\,.
\end{equation}
Such is the case for classical sampling tasks, as well as in the context of the task of distinguishing between lab oracles, which is the main focus of this paper; Example $2$ in Subsection \ref{subsec:ExampleQUALMs} demonstrates the usefulness of this generalization when considering continuous sets of lab oracles. In Definition~\ref{def:distinguish2} and the surrounding text, we discuss the fact that this generalized notion of a task (i.e., with $f$'s of the form in Eqn.~\eqref{eq:proboutput1}) reduces to Definition~\ref{def:task} for the distinguishing tasks which are the main focus of the paper (Definitions \ref{def:distinguish2} and \ref{def:distinguish3}).
\end{remark} 

\begin{remark}[\textit{A} $\Task$ \textit{with quantum input and output states}]
We can also further generalize Definition \ref{def:task} by upgrading the classical
inputs and/or outputs of $f$ to be density matrices.  That is, $f$ becomes
\begin{equation}\label{eq:proboutput}f : \{\LO_0, \LO_1, \LO_2,...\} \times \mathcal{D}(\fontH{S}_{\text{in}}) \longrightarrow \mathcal{D}(\fontH{S}_{\text{out}})\;.
\end{equation}
This might be relevant 
when one is interested in sampling from Gibbs distributions of quantum Hamiltonians, or when considering QUALMs as {\it procedures}, or {\it subroutines} within other QUALMs. 
\end{remark}

Notice that if $\fontH{N} \simeq \mathbb{C}$ (i.e.~it is a trivial subsystem) and if the set of lab oracles is taken to be empty, then $f$ in Definition \ref{def:task} reduces to a map 
$$f : \{0,1\}^p \longrightarrow \{0,1\}^q\,.$$
Accordingly, the `task' is just that of computing the function $f$, using a quantum algorithm composed of gates from the gate set $\mathcal{G}$.  See Subsection \ref{subsec:ExampleQUALMs} for a more elaborate discussion of restricting QUALMs to standard quantum algorithms, as well as other examples.

\vskip.3cm
\indent We can now define the QUALM, which can be viewed as a specific choice of {\it protocol} for the execution of the desired experimental task.

\begin{definition}{\bf (QUALM)}\label{def:qualm} 
A QUALM  over the set of admissible gates $\mathcal{G}$ 
 acting on registers $\fontH{L},\fontH{W}$, is an ordered sequence of symbols $\mathcal{Q} = (\mathcal{Q}_1,\mathcal{Q}_2,...,\mathcal{Q}_{\text{final}})$ from the alphabet $\mathcal{G}\cup \{\square\}$, together with a specification of input and output subsystems $\fontH{S}_{\text{in}}, \fontH{S}_{\text{out}}
\subseteq \fontH{W}$. 
Here, we are treating $\mathcal{G}$ as a set of symbols (i.e., each gate labels a distinct symbol) and likewise $\square$ is a symbol.  Each QUALM has an associated map
\begin{equation}
\QUALM : \mathscr{L}\!\mathscr{O}(\fontH{N},\fontH{L}) \longrightarrow \textsf{QuantumCircuits}(\fontH{N}\otimes\fontH{L}\otimes\fontH{W})\,.
\end{equation}
This function takes in a lab oracle $\LO$, and outputs a quantum circuit on $\fontH{N}\otimes\fontH{L}\otimes\fontH{W}$.  Specifically, $\QUALM(\LO)$ `compiles' a quantum circuit $\mathcal{Q} = (\mathcal{Q}_1,\mathcal{Q}_2,...,\mathcal{Q}_{\text{final}})$ where each symbol in $\mathcal{G}$ is replaced by its corresponding gate, and each $\square$ is replaced by the superoperator $\mathcal{E}_{\fontH{NL}}$ corresponding to $\LO$. $\fontH{S}_{\text{in}}, \fontH{S}_{\text{out}}$ correspond to the input and output subsystems of the resulting circuit, respectively. 
\end{definition} 

\noindent In less formal terms, a QUALM is a quantum circuit built out of an admissible gate set, where the circuit has designated spots for a lab oracle superoperator to be inserted, and specified input and output qubits.  

Now we explain what it means for a QUALM to achieve a particular $\Task$. 
We first define the output density matrix of a QUALM for a given lab oracle $\LO$. The idea is to compile the quantum circuit $\QUALM(\LO)$ for the lab oracle $\LO$, and then to use it to evaluate $f(\LO, x)$.  To do so, we construct the initial state of the circuit to be (i) $\rho_{\fontH{N}}^{\LO}$ (i.e.~the state corresponding to the lab oracle $\LO$) on $\fontH{N}$, (ii) $|x\rangle \langle x|$ on $\fontH{S}_{\text{in}}$, and (iii) the zero state elsewhere.  The full initial state will be denoted as $\rho_{\fontH{N}}^{\LO}  \otimes |x\rangle \langle x|_{\fontH{S}_{\text{in}}} \otimes |\textbf{0}\rangle \langle \textbf{0}|$.  We will run the initial state through the circuit $\QUALM(\LO)$, and then trace out everything not in $\fontH{S}_{\text{out}}$. 

\begin{definition}[{\bf Output density matrix of a QUALM on a given lab oracle and classical input}] Given a QUALM on admissible gates $\mathcal{G}$, with specified input and output subsystems $\fontH{S}_{\text{in}}$ and 
$\fontH{S}_{\text{out}}$\,, its output density matrix on a given lab oracle $\LO$ and classical input $x$ is defined to be 
\begin{equation}
\label{eq:QUALMcriteria0}
\rho_{\fontH{S}_{\text{out}}}( \QUALM(\LO,x))
=   \text{tr}_{\overline{\fontH{S}}_{\text{out}}}\!\left\{\QUALM(\LO)\!\!\left[\rho_{\fontH{N}}^{\LO} \otimes |x\rangle \langle x|_{\fontH{S}_{\text{in}}} \otimes |\textbf{0}\rangle \langle \textbf{0}| \right]\right\}.
\end{equation}
\end{definition} 

To define the notion of a QUALM achieving a task, note that a $\Task$ specifies a function $f : \{\LO_0, \LO_1, \LO_2,...\} \times \{0,1\}^p \longrightarrow \{0,1\}^q$.
\begin{definition}[{\bf QUALM implementing a Task}] \label{def:implementing} A QUALM on admissible gates $\mathcal{G}$  with specified input and output subsystems $\fontH{S}_{\text{in}}$ and 
$\fontH{S}_{\text{out}}$ implements a $\Task = (\fontH{S}_{\text{in}}, \fontH{S}_{\text{out}}, f, \mathcal{G})$, if
\begin{equation}
\label{eq:QUALMcriteria}
\rho_{\fontH{S}_{\text{out}}}( \QUALM(\LO_i, x))
= |f(\LO_i, x)\rangle \langle f(\LO_i, x)|
\end{equation}
for all $\LO_i,x$ in the domain of $f$ (i.e.,~for all $\LO_i \in \{\LO_0, \LO_1, \LO_2,...\}$, and all $x \in \{0,1\}^p$).

We say that the QUALM implements the $\Task$ with error at most $\epsilon$ if for each input $\LO_i, x$, we have 
\begin{equation} 
\big\|\rho_{\fontH{S}_{\text{out}}}( \QUALM(\LO_i, x))-|f(\LO_i, x)\rangle \langle f(\LO_i, x)|\big\|_1\le \epsilon\,.
\end{equation} 
\end{definition}

It is worthwhile to again compare to the setting of ordinary quantum algorithms, in which $f : \{0,1\}^p \to \{0,1\}^q$.  Again, this is a restriction of $f$ defined above in Definition~\ref{def:task} by letting the set of lab oracles in the domain be the empty set.  If the QUALM superoperator is just a unitary circuit $U$, then Eqn.~\eqref{eq:QUALMcriteria} becomes
\begin{equation}
\label{E:QUALMnormalAlg1}
    \text{tr}_{\overline{\fontH{S}}_{\text{out}}}\!\left\{U \big(|x\rangle \langle x|_{\fontH{S}_{\text{in}}} \otimes |\textbf{0}\rangle \langle \textbf{0}|\big)U^\dagger\right\} = |f(x)\rangle \langle f(x)|
\end{equation}
which just says that if we input $x$ on $\fontH{S}_{\text{in}}$\,, then we output $f(x)$ on $\fontH{S}_{\text{out}}$\,.

{~}
\begin{remark}[{\it Subroutines, Error reduction by repetition, Approximation}]\label{remark:amp} 
We remark that standard manipulations from the theory of algorithms carry over to the QUALM setting in a natural way. 
In particular, a QUALM can be used as a subroutine by another QUALM, as long as they both act on the same lab register $\fontH{L}$ (like in standard subroutines, we can decide which qubit ``wires'' to glue from the original QUALM to the input qubits of the subroutine QUALM, and similarly for the outputs). The more natural definition for QUALMs in this case is the general one in which both the inputs and outputs 
are quantum registers, as in Eqn.~\eqref{eq:proboutput}, unless we have restrictions 
on the input and output wires which make them classical. 

Using such subroutine QUALMs, one can achieve error reduction (also known as amplification); this is a standard primitive in probabilistic algorithms \cite{arorabarak}, which is also needed in this work. This is done by a straightforward generalization of the way it is done for algorithms. For example, suppose we want to reduce the error probability of a given QUALM which achieves a certain deterministic task with error $1/3$; and further suppose that the image of the function $f$ has a single output bit, measured at the end of the QUALM in the computational basis. One can construct a new QUALM', which first copies the $p$-bit input string $x$ $m$ times (for some desired amplification parameter $m$). It then applies QUALM as a subroutine $m$ times, each time with a new set of qubits initialized to $0$ (which together with the $p$ qubits containing the appropriate copy of $x$, constituting the working register $\fontH{W}$ of the particular subroutine). QUALM' then applies a majority calculation on all outputs of the $m$ subroutines; this majority is the output bit of the QUALM'. 

In the same manner, we can consider amplifications for QUALMs which compute probabilistic functions as in Eqn.~\eqref{eq:proboutput}: if for two different lab oracles the output distributions are 
$\delta$ apart in total variation distance, one can amplify the distance by repetition and classical postprocessing.  

Finally, for QUALMs which output probability distributions, one can talk about QUALMs which achieve certain tasks approximately.
A metric on tasks is very naturally defined by considering the maximal total variation distance between the distributions which the two functions output for a given input, maximized over the lab oracles and classical inputs in the domain of the function. 

We omit the formal details of these facts and definitions, as the details are completely straightforward as generalizations of probabilistic algorithms.
\end{remark}

\subsection{QUALM complexity} 

Having defined tasks and QUALMs, we now turn to defining QUALM complexity.

\begin{definition}[\bf Gate complexity, query complexity, and QUALM complexity] \label{def:complexity0}
The gate complexity of a given QUALM over the admissible set of gates $\mathcal{G}$ is the length (i.e.~the number of symbols from $\mathcal{G}\cup \{\square\}$) of 
the sequence $\mathcal{Q}$, minus the number of $\square$ symbols.  We denote this by $\textsf{GateComplexity}$[QUALM], and call this the QUALM gate complexity.  Similarly, the query complexity is the number of $\square$'s appearing in $\mathcal{Q}$, and this is denoted by $\textsf{QueryComplexity}$[QUALM].  This is called the QUALM query complexity.  We call the sum
\begin{equation}
\textsf{GateComplexity}[\QUALM] + \textsf{QueryComplexity}[\QUALM] = |\mathcal{Q}|
\end{equation}
the QUALM complexity.
\end{definition} 

\begin{definition}[\bf Exact QUALM complexity of a task]\label{def:complexity}  Given a
$\Task = (\fontH{S}_{\text{in}}, \fontH{S}_{\text{out}},f,\mathcal{G})$,
its exact QUALM gate complexity is 
$$\min_{\{\QUALM \,: \, \QUALM\text{ implements }\Task\}} \textsf{GateComplexity}[\QUALM]\,,$$
its exact QUALM query complexity is
$$\min_{\{\QUALM \,: \, \QUALM\text{ implements }\Task\}} \textsf{QueryComplexity}[\QUALM]\,,$$
and its exact QUALM complexity is
$$\min_{\{\QUALM \,: \, \QUALM\text{ implements }\Task\}} \big(\textsf{GateComplexity}[\QUALM] + \textsf{QueryComplexity}[\QUALM]\big)\,.$$
Depending on the context, we can also specify that we take the minimum over all QUALMs which probabilistically implement the $\Task$.
\end{definition}
\noindent Note that it need not be the case that the exact QUALM gate complexity of a task plus the exact QUALM query complexity of a task be equal to the exact QUALM complexity of a task.  This is because there may be a trade off between the gate and query complexities that arises when minimizing over their sum.

\begin{remark}[{\it Modified definition of QUALM complexity}]
In some situations, it may be useful to modify the definition of the QUALM complexity to be
\begin{equation}
    \textsf{GateComplexity}[\QUALM] + \lambda\,\textsf{QueryComplexity}[\QUALM]
\end{equation}
for a penalty factor $\lambda$.  Accordingly, the exact QUALM complexity would be
\begin{equation}
    \min_{\{\QUALM \,: \, \QUALM\text{ implements }\Task\}} \big(\textsf{GateComplexity}[\QUALM] + \lambda\,\textsf{QueryComplexity}[\QUALM]\big)\,.
\end{equation}
For instance, $\lambda$ could be the number of elementary gates required to construct the lab oracle superoperator.
\end{remark}

As usual in computational complexity \cite{papa}, one is interested in {\it families} of tasks and QUALMs, where some parameter dictating the \textit{size} of the problem grows to infinity, and we ask how the complexity grows as a function of that parameter.   
\vskip.3cm

\begin{remark}[{\it Approximate complexity}]
As mentioned in Remark \ref{remark:amp}, one can talk about approximations and errors in QUALMs. The approximate QUALM gate complexity, approximate QUALM query complexity, and approximate QUALM complexity of a task, respectively, are then clearly defined using 
Definition \ref{def:complexity}, taking the minimum over all QUALMs which achieves the task to the desired approximation.
\end{remark} 

\subsection{Different types of access to a lab oracle} 
Towards clarifying the power of different QUALMs in terms of their computational abilities, we consider different types of {\it accesses} of QUALMs to a lab oracle.  The most natural (though presently least realistic) one is the access of a full quantum computer. By this we mean a general QUALM, without the restrictions to be specified shortly.  In particular, the set of admissible gates $\mathcal{G}$ can be a universal set, and there is no restriction on which gates can be applied at any given time. 

\begin{definition}[\textbf{Coherent access}]
We will refer to a general QUALM (i.e., with a universal gate set) defined in Definition~\ref{def:qualm} as a QUALM with {\it coherent access} to the lab oracle. 
\end{definition} 

In realistic physics experiments, access is often far more limited. The experimental setup may not be able to introduce quantum entanglement between the physical system $\fontH{N}\otimes\fontH{L}$ and the rest of the lab $\fontH{W}$. Measurements to the system may destroy the quantum coherence in the physical system or even completely destroy the quantum state. In the QUALM framework this is captured by restrictions on the admissible gate sets and the allowed sequences of gates. To make concrete progress we consider a model of many contemporary experiments which we call {\it incoherent access}.  To this end, we need to recall the definition of a one-round LOCC protocol between two parties (see, e.g., \cite{LOCCref1}). 
First, we recall that a map $\mathcal{N}^{\fontH{R}}$ on density matrices on the register $\fontH{R}$ is completely positive (CP for short) if it can be associated 
with a set of Kraus operators $\{\mathcal{A}_\alpha\}_\alpha$ such that for all $\rho$ on $\fontH{R}$ 
\[\mathcal{N}^{\fontH{R}}(\rho)=\sum_\alpha \mathcal{A}_{\alpha}\rho \mathcal{A}_{\alpha}^{\dagger}\,.\]
Furthermore, $\mathcal{N}^{\fontH{R}}$ is trace-preserving if
\[\sum_\alpha \mathcal{A}_\alpha^{\dagger} \mathcal{A}_\alpha=\mathds{1}.\] 
A completely positive trace-preserving (CPTP) map is often called, simply, a quantum channel.

Now we specify a quantum channel that implements a certain kind of communication between subsystems $\fontH{A}$ and $\fontH{B}$ called a ``one-round LOCC'' (see~\cite{LOCCref1}).

\begin{definition}[\bf One-round LOCC] \label{def:LOCC1}
A {\it one-round LOCC 
operator} from subsystems $\fontH{A}$ to $\fontH{B}$ 
is a 
quantum channel (i.e.~a CPTP map)  $\mathcal{E}^{\fontH{AB}}$ acting on
$\mathcal{D}(\fontH{A}\otimes\fontH{B})$, which is of the form 
\begin{equation}\label{eq:locc}
\mathcal{E}^{\fontH{AB}}(\cdot)=\sum_{\alpha} \mathcal{M}_{\alpha}^{\fontH{A}}(\cdot)\otimes \mathcal{N}_{\alpha}^{\fontH{B}}(\cdot) 
\end{equation}
with $\mathcal{M}_{\alpha}^{\fontH{A}}$ a completely positive (CP) map acting on the $\fontH{A}$ subsystem, and $\mathcal{N}_{\alpha}^{\fontH{B}}$ a completely positive and trace-preserving (CPTP) map acting on the $\fontH{B}$ subsystem.
\end{definition} 
\noindent It should be noted that this definition in fact requires $\mathcal{E}^{\fontH{A}}\equiv\sum_\alpha \mathcal{M}_{\alpha}^{\fontH{A}}$ to be a CPTP map (instead of merely a CP map), since $\mathcal{E}^{\fontH{A}}$ is the reduced channel obtained by tracing over $\fontH{B}$ in $\mathcal{E}^{\fontH{AB}}$. 

With the above definition in mind, we can now define incoherent access QUALMs: 

\begin{definition}[\bf Incoherent access]
\label{def:incoherent}
We say a QUALM uses {\it incoherent access} to the lab oracle if the following holds for the sequence of symbols $\mathcal{Q}$. 
Let $k$ be the number of times the symbol $\square$ appears in $\mathcal{Q} = (\mathcal{Q}_1, \mathcal{Q}_2,...,\mathcal{Q}_{\text{final}})$.  As usual, we associate $\mathcal{Q}$ with the channel $\mathcal{Q} = \mathcal{Q}_{\text{final}} \circ \cdots \circ \mathcal{Q}_2 \circ \mathcal{Q}_1$, and regroup terms as
\[\mathcal{Q}=
\mathcal{C}_{k} \circ \square \circ \mathcal{C}_{k-1}\circ \square \circ\cdots \circ\square \circ\mathcal{C}_{0}\] 
where $\mathcal{C}_i$ is a the sequence of gates applied after the $i$th call to the lab oracle, and before the $(i+1)$st one. We require that 
\begin{itemize} 
\item For each $i\in \{0,...,k\}$, the sequence of admissible gates  $\mathcal{C}_i$ can be written as an $r_i$-round LOCC protocol, for some finite number of rounds $r_i$. 
In other words, $\mathcal{C}_i$ can be written as a composition $\mathcal{C}_i= \mathcal{C}_{i,r_i} \circ \cdots \circ \mathcal{C}_{i,2} \circ \mathcal{C}_{i,1}$
such that for every $j$, $\mathcal{C}_{i,j}$ is a one-round LOCC channel. Without losing generality, we can assume $\mathcal{C}_{i,j}$ alternates between $\fontH{L}$-to-$\fontH{W}$ and $\fontH{W}$-to-$\fontH{L}$ one-round LOCC channels, since the composition of two $\fontH{L}$-to-$\fontH{W}$ one-round LOCC's is still a one-round LOCC.
\item Moreover we also require that for each $i$ there exists at least one index  $j\in \{1,...,r_{i}\}$ such that 
$\mathcal{C}_{i,j}$ is a complete measurement, which is a special one-round LOCC operator from $\fontH{L}$ to $\fontH{W}$, as follows: 
\begin{equation}\label{eq:speciallocc} 
\mathcal{C}_{i,j}(\cdot)=\sum_{\alpha\in\{0,1\}^\ell} \mathcal{M}_{\alpha}^{\fontH{L}}(\cdot)\otimes \mathcal{N}_{\alpha}^{\fontH{W}}(\cdot),
\end{equation}
such that $\mathcal{M}_{\alpha}^{\fontH{L}}$ is a rank one projection for each $\alpha$, i.e.,  $\mathcal{M}_{\alpha}^{\fontH{L}}(\rho)=\ket{\psi_{\alpha}}\bra{\psi_{\alpha}}\bra{\psi_{\alpha}}\rho\ket{\psi_{\alpha}}$ for some pure state $\ket{\psi_{\alpha}}$, 
and $\{\ket{\psi_\alpha}\}_{\alpha\in\{0,1\}^\ell}$ is an orthonormal basis for \fontH{L}.

\end{itemize} 

\end{definition}

The above definition roughly means that the interaction between the $\fontH{L}$ and $\fontH{W}$ registers is an LOCC throughout the QUALM protocol; moreover, between any two applications of the lab oracle, the lab register $\fontH{L}$ is measured using a complete measurement. No coherence can be generated between the state generated by a given single call to the lab oracle, and any other register used by the QUALM -- this is the source for the term ``incoherent access QUALM''.

We note that the above definition allows {\it adaptive} access to the lab oracle; namely, the state of the register $\fontH{L}$  before an application of the lab oracle may depend on previous measurement results both of $\fontH{L}$ and of $\fontH{W}$, and those in turn may depend on the lab oracle.

\subsection{Examples of Quantum Tasks}
\label{subsec:ExampleQUALMs}

Before turning to new and novel tasks and QUALMs, we give examples of how standard quantum algorithms, along with quantum experimental protocols, can be cast as tasks and QUALMs.  This is interesting mainly for quantum experimental protocols, which are not themselves quantum algorithms.  However, we begin by showing how quantum algorithms are a special case of quantum tasks and QUALMs:  \\ \\
\textbf{Example 1: Quantum algorithms.} Consider a $\Task = (\fontH{S}_{\text{in}}, \fontH{S}_{\text{out}}, f, \mathcal{G})$ where $\mathcal{G}$ is a universal quantum gate set.  We further let $\fontH{N},\fontH{L}\simeq \mathbb{C}$ be trivial subsystems, and choose $f$ to be
\begin{equation}
f : \{\mathcal{I}\}\times \{0,1\}^p \longrightarrow \{0,1\}^q
\end{equation}
with $|\fontH{S}_{\text{in}}|=p,  |\fontH{S}_{\text{out}}|=q$, and 
$\mathcal{I}$ is the trivial lab oracle.
Then a QUALM achieving this $\Task$ is a quantum algorithm computing the function $f$.
\\ \\
\indent Next, we turn to procedures like quantum tomography; we will focus on the state setting in particular.
\\ \\
\textbf{Example 2: Quantum state tomography.} Here we consider quantum state tomography on states with $\ell$ qubits (i.e., the size of the lab register $\fontH{L}$).  We further suppose that $|\fontH{N}| = n = 0$ and $|\fontH{L}| = \ell$.  The lab oracle $\LO_\rho$ is defined to be a channel which takes whatever state is on $\fontH{L}$ and turns it into $\rho$. Formally it is defined by  $\LO_\rho = (\{\,\},\,\mathcal{E}_{\fontH{L},\,\rho})$, where we do not specify a state on $\fontH{N}$ since the subsystem is empty, and $\mathcal{E}_{\fontH{L},\,\rho}$ only acts on $\fontH{L}$ where for any $\sigma \in \mathcal{D}(\fontH{L})$ we have
\begin{align}
\mathcal{E}_{\fontH{L},\,\rho}[\sigma] = \text{tr}(\sigma) \, \rho\,.
\end{align} Accordingly, each application of $\mathcal{E}_{\fontH{L},\,\rho}$ produces a copy of $\rho$ on $\fontH{L}$.

To model quantum tomography on a general $\ell$ qubit state up to error $\epsilon$ in the $1$-norm, we consider first 
a set of states $\rho_1,...,\rho_m$ with $m = \Omega(1/\epsilon^{(4^\ell)})$ which forms an epsilon net in $\mathcal{D}(\fontH{L})$; that is, for any $\ell$ qubit 
state $\rho$, there is a state  $\rho_i$ in the net such that $\|\rho-\rho_1\|_1\le \epsilon$.

Then consider $\Task = (\fontH{S}_{\text{in}}, \fontH{S}_{\text{out}}, f, \mathcal{G})$ where $|\fontH{S}_{\text{in}}| = 0$, $|\fontH{S}_{\text{out}}| = q = \lceil \log_2(m) \rceil$, and $\mathcal{G}$ is a universal gate set.  We further have
\begin{equation}
f : \{\LO_{\rho}\}_{\rho\in \mathcal{D}(\fontH{L})} \longrightarrow \{0,1\}^q
\end{equation}
such that for 
$
f(\LO_{\rho}) = k$ we have
$\|\rho-\rho_k\|_1\le \epsilon$ (where the output string  $k\in\{0,1\}^q$ is interpreted as a binary representation of an integer between $1$ and $m$). 

In fact, the task defined above is non-physical as it requires infinite precision. This is because the function $f$ {\it partitions} the set $\mathcal{D}(\fontH{L})$ according to the value $k$ of $f$. No matter how this partition is defined, there will always be  two arbitrarily close states such that the values of the function defined on them are different.

To circumvent this problem, one can restrict the input states to be just those states far from the boundaries of this partition. However, a more natural approach which is closer to physics is to define the state tomography task using a probabilistic function $f$; namely, $f$ outputs a \textit{distribution} over $q$-bit strings: 
\begin{equation}
f : \{\LO_{\rho}\}_{\rho\in \mathcal{D}(\fontH{L})} \longrightarrow \mathcal{D}\{\{0,1\}^q\}
\end{equation}
Further, we require that for any $\rho$ and any $k \in \{0,1\}^q$ we have 
\begin{equation}
\text{Pr}_{k \sim f(\LO_\rho)} \!\left[\|\rho-\rho_k\|_1\le \epsilon\right]\geq 1-\delta\,. 
\end{equation}
This defines the task of state tomography to within approximation $\epsilon$, with error $\delta$ (here $1-\delta$ can be referred to as the {\it confidence}). 

In a very similar way we can describe quantum channel tomography; we will not flesh it out here. 
\\ \\
\indent Now we examine more elaborate protocols, including quantum cooling and quantum communication.
\\ \\
\textbf{Example 3: Quantum cooling protocols.}
The following is an example of the more general version of a QUALM, where the output can be a quantum state. Suppose we have a set of Hamiltonians $\{H_1, H_2,...,H_m\}$ on $\ell$ qubits and corresponding Gibbs states at temperature $T$, namely $\{\rho_{1}(T), \rho_{2}(T),...,\rho_{m}(T)\}$.  We let $\fontH{N} \simeq \mathbb{C}$ (i.e., it is trivial), and consider a family of lab oracles $\LO_k = (\{\,\},\, \mathcal{S}_{\fontH{L},k})$ where
\begin{equation}
\mathcal{S}_{\fontH{L},k}(\sigma_{\fontH{L}}) = \rho_k(T)
\end{equation}
for any $\sigma_{\fontH{L}}$.  That is, the lab oracle superoperator prepares the Gibbs state $\rho_k(T)$ on $\fontH{L}$.

Then the cooling $\Task$ is as follows.  Suppose that $\fontH{W}$ contains a subsystem $\fontH{W}_1$ comprising a qudit with orthonormal basis $\{|1\rangle,|2\rangle,...,|r\rangle\}$, and another subsystem $\fontH{W}_2$ with the same number of qubits as $\fontH{L}$.  Then we let $\fontH{S}_{\text{in}} = \fontH{W}_1$, $\fontH{S}_{\text{out}} = \fontH{W}_2$, $\mathcal{G}$ be a universal gate set, and
\begin{equation}
    f : \{\LO_1, \LO_2,...,\LO_m\} \times \{1,2,...,r\} \longrightarrow \mathcal{D}(\fontH{W}_2)
\end{equation}
where
\begin{equation}
f(\LO_k, s) = \rho_k(T/s)\,.
\end{equation}
That is, the task is asking for us to output the Gibbs state cooled to $1/s$ of its initial temperature.

Notice that this task is defined \textit{without} access to the Hamiltonian inducing the Gibbs state. 
Interestingly,~\cite{QVC1} offers an efficient QUALM for a weaker version of this task, namely computing expectation values of local observables with respect to the reduced temperature Gibbs state. Related protocols in the context of quantum error correction are~\cite{Koczor1, Huggins1}.
\\ \\
\noindent \textbf{Example 4: Quantum communication protocols.} Communication complexity~\cite{Yao93} is a well-known model in theoretical computer science, in which one considers the following setting. There are two parties, Alice and Bob, and a function $F :\{0,1\}^p\times \{0,1\}^q\longmapsto \{0,1\}$, which is fixed in advance (e.g., if $p = q$ then $F$ can be the equality function between on bit strings). Each of the parties receives an input ($x$ and $y$ respectively) without the other party knowing its value. Alice and Bob now run a ``protocol'' in which they send bits (or qubits) back and forth between them. At the end of the protocol, one of them (say, Bob) is expected to hold the value 
$F(x,y)$. The communication complexity of the function is the minimal number of bits (in the quantum case, qubits) which they need to send to each other in order to achieve this goal. 

We describe a way to present the communication complexity model in the QUALM framework.  The communication complexity turns out to be equal to the QUALM query complexity.\footnote{
There is another, possibly more straightforward, way to embed the model in the QUALM framework. Let Alice be represented by $\fontH{L}$ and let Bob be represented by $\fontH{W}'$ (which is part of the working space $\fontH{W}$).  There is also a single qubit ``message'' register
$\fontH{M}$ included in $\fontH{W}$. Gates are not allowed to act jointly on qubits from both $\fontH{W}'$
and $\fontH{L}$. The communication protocol can be viewed as a computation, alternating between a computation on $\fontH{L} \otimes \fontH{M}$, and a computation on  $\fontH{W}' \otimes \fontH{M}$. 
The communication complexity of the protocol is then simply the number of such computation steps; however, in this presentation, the communication complexity is not related to the QUALM query complexity or the QUALM gate complexity.}

Let $\fontH{N}$ be an empty subsystem, let the lab decompose as $\fontH{L} \simeq \fontH{L}_A \otimes \fontH{L}_B$ where $\fontH{L}_A$ and $\fontH{L}_B$ are each single-qubit subsystems, and let $\fontH{W} \simeq \fontH{W}_A \otimes \fontH{W}_B$ where $|\fontH{W}_A| = |\fontH{W}_B| = w/2$ for $w$ even.  We can take $w$ to be very large.  Our task is denoted by $\Task_F$, and is labelled by a function $F : \{0,1\}^p \times \{0,1\}^q \to \{0,1\}$.  We take $\fontH{S}_{\text{in}}$ to be the union of the first $p$ qubits of $\fontH{W}_A$ with the first $q$ qubits of $\fontH{W}_B$, and take $\fontH{S}_{\text{out}}$ to be the first qubit of $\fontH{W}_B$.  Our admissible gate set $\mathcal{G}$ is \textit{not} universal; rather, it is the union of a universal gate set on $\fontH{L}_A \otimes \fontH{W}_A$ with a universal gate set on $\fontH{L}_B \otimes \fontH{W}_B$.  If we imagine $\fontH{L}_A \otimes \fontH{W}_A$ as belonging to Alice and $\fontH{L}_B \otimes \fontH{W}_B$ as belonging to Bob, then they can do arbitrary quantum computation on their own domains, but cannot directly share information with one another.  To allow Alice and Bob to communicate, we consider a lab oracle $\LO_{\text{SWAP}}$ on $\fontH{L}_A \otimes \fontH{L}_B$ which simply swaps the two registers.  Then the function $f$ corresponding to $\Task_F$ is
\begin{equation}
f : \{\LO_{\text{SWAP}}\} \times \{0,1\}^p \times \{0,1\}^q \longrightarrow \{0,1\}
\end{equation}
where
\begin{equation}
f(x,y) = F(x,y)
\end{equation}
for $x \in \{0,1\}^p$ and $y \in \{0,1\}^q$.  We have suppressed the dependence on $\LO_{\text{SWAP}}$ since there are no other lab oracles to choose from.

In less mathematical terms, the task is as follows.  Alice and Bob are each given bit strings $x$ and $y$, respectively, and desire to compute a function $F(x,y)$ which Bob is to output.  Alice and Bob can only communicate one qubit at a time via the lab oracle.  The QUALM query complexity of the task is the same as the quantum communication complexity of the task.\footnote{Technically, our particular construction here counts a bit swap between Alice and Bob as a single operation.  However, in standard quantum communication complexity, Alice sending a qubit to Bob is counted as one operation, and Bob sending a qubit to Alice is counted as one operation.  Of course, our factor of two mismatch does not affect the asymptotic scaling of the complexity.  If one desires, the factor of two can be remedied by making the task setup slightly more sophisticated.}
\\ \\
\indent In~\ref{App:verification}, we discuss the slightly more elaborate example of verifying the correct evolution of a quantum computer.
\\ \\
\indent To the best of our knowledge, all examples of quantum experiments (including all those mentioned in the introduction) can be cast in the QUALM framework. 
 Two additional interesting examples which we will mention here without going into explicit detail are control theory and imaging. Quantum tasks lend themselves well to quantum control theory problems (see e.g.~\cite{control1}).  The system to be controlled is regarded as part of Nature, and we can only interact with it or \textit{control} it via the lab oracle superoperator; this couples Nature to \textit{control parameters} in the lab.  The lab oracle superoperator may also produce some `readout' of the Nature system by `printing' bits on a lab subsystem. This way, each time we utilize the lab oracle superoperator, we can read out some property of the present state of Nature which may influence our subsequent choices of control parameters.  It is also natural to cast imaging tasks into the QUALM framework, along similar lines as Figure 1(a). 
 For imaging, the Nature register is the object to be measured, and the lab register contains the equipment that directly interacts with the object (for example, the cantilever and tip of an atomic force microscope).  The lab oracle specifies (i) the state of the object, and (ii) the superoperator that controls the dynamics of the object and its interaction with the lab register. Depending on the imaging technique, the interaction can occur by exchanging photons or electrons, or via a static electromagnetic field. The signal received by the lab register is then post-processed to obtain the output image in the working memory register. In most present-day experiments, the post-processing is classical, but in general it could be quantum. The task is a function that maps the (unknown) state and/or dynamics of the object to a desired output image that captures some designated features of the object.

\section{Exponential advantage for coherent access} \label{sec:mainresult} 

In this section, we construct a quantum task for which a coherent access QUALM has an exponential advantage over an incoherent access QUALM.  That is, we (i) specify a coherent access QUALM which achieves the task with a number of elementary gates linear in the number of qubits, and (ii) prove that incoherent access QUALMs require \textit{exponentially} many elementary gates in the number of qubits to achieve the task.  Note that (ii) is by far the most difficult part of our arguments, as we must prove the \textit{non-existence} of subexponential incoherent access QUALMs for our quantum task at hand.

The content of this section provides a proof of principle that coherent access QUALMs can have a significant advantage over incoherent access QUALMs, even if the latter are adaptive; accordingly some physical experiments may be made vastly more efficient by interfacing with a quantum computer -- possibly even with a very simple one.  This result represents a novel kind of \textit{exponential advantage} afforded by quantum computation to speed up physics experiments.  We discuss related examples to our proof of principle in Sections~\ref{sec:corrs1} and~\ref{sec:symmtask}.

\subsection{Problem definition and statement of results} 
Here we will study the problem of distinguishing between (i) many copies of a fixed, Haar random unitary, and (ii) many i.i.d.~Haar random unitaries.
While the task of distinguishing (i) and (ii) may seem slightly abstract, it is a toy version of a sensible experimental problem: given an experimental system, determine if its dynamics is time-translation invariant or stochastic. 
A system with discrete time translation symmetry is known as a Floquet system. If the time translation period is $T$, the time translation operator for time interval $[kT, (k+1)T]$ is a unitary $U$ independent of the integer $k$. By contrast, a stochastic system such as a random quantum circuit \cite{HarrowCircuit1} or a Brownian circuit (e.g.~\cite{Lashkari2013}) has a time-dependent evolution operator $U_k$ for the same time interval, which depends on some classical random variables that are uncorrelated for different time intervals. The toy model we consider corresponds to a highly chaotic case, where we assume the $U$ for the Floquet system is Haar random, and the $U_k$'s for stochastic system are i.i.d.~Haar random.

We begin by defining the task of distinguishing between two lab oracles, and then define our precise lab oracles of interest.

\begin{definition}[\textbf{The task of distinguishing between two lab oracles}]
\label{def:distinguish2}
Consider two lab oracles  $\LO_0=(\mathcal{E}_{\fontH{NL}},\rho_{\fontH{N}})$, 
$\LO_1=(\mathcal{E}_{\fontH{NL}}',\rho_{\fontH{N}}')$,
and a quantum task
$(\fontH{S}_{\text{in}},\fontH{S}_{\text{out}},f,\mathcal{G})$.  Here 
$\fontH{S}_{\text{in}}$ is empty, 
$\fontH{S}_{\text{out}}$ consists of a single qubit, $\mathcal{G}$ is 
any universal set of quantum gates, and 
$$f:\{\LO_0,\LO_1\}\to \{0,1\}\,,$$
where $f(\LO_0) = 0$ and $f(\LO_1) = 1$.
\end{definition}
\noindent By Definition \ref{def:implementing}, 
a QUALM achieves the above task with error at most $\epsilon$ if its output density matrix (in the present case, on a single qubit) is within $\epsilon$ in trace distance from the correct one. 

\begin{definition}[\bf A QUALM distinguishing between two lab oracles   with a bias $\delta$]\label{def:bias} 
Suppose we have a $\QUALM$ with a single qubit output, $|\fontH{S}_{out}|=1$, which  outputs the density matrix $\rho_0$ or $\rho_1$  when the lab oracle is $\textsf{LO}_0$ or $\textsf{LO}_1$, respectively. We say that the QUALM distinguishes between the two lab oracles with bias $\delta$ if 
the trace distance between its two output density matrices is at least $\delta$: 
\begin{equation}
    \|\rho_0 - \rho_1\|_1 \geq \delta\,.
\end{equation}
\end{definition}

\begin{remark}[\it From biases of single qubit density matrices to deterministic functions] \label{re:frombiastoerror} In the remainder of the paper, when we consider QUALMs performing tasks which distinguish between lab oracles, we work with the more general Definition \ref{def:bias} and thus incorporate biases $\delta$ into our analyses.  
Note that given a QUALM which distinguishes between the two lab oracles with bias $\delta=\Omega(1)$, we can 
by repetition and classical computation 
(as in Remark \ref{remark:amp}) construct a new QUALM' which achieves the task in Definition~\ref{def:distinguish2}, i.e.~which implements $f(\textsf{LO}_0) = 0$ and $f(\textsf{LO}_1) = 1$ with error less than $\epsilon$ (where $\epsilon$ is any inverse polynomial in $\ell$ of our choice). 
As long as $\delta$ is inverse polynomial in $\ell$, the QUALM query and gate complexities are just multiplied by some polynomial factor in $\ell$. 

Notice that the above is true since we require in Definition \ref{def:bias} that the output is a single qubit. 
This is done to avoid the following situation.  Suppose we have a bigger $\fontH{S}_{\text{out}}$ and the two output states $\rho_0,\rho_1$ have a bias $\delta'$ in the $1$-norm, namely $\|\rho_0 - \rho_1\|_1 \geq \delta'$.  Then there exists an operator on $\fontH{S}_{\text{out}}$ that, when measured, can distinguish the two output states with this bias.  However, for a large $\fontH{S}_{\text{out}}$ the complexity of finding such an operator and then carrying out the measurement can be large. As such, requiring the output to be only a single qubit ensures that the QUALM gate complexity in Definition's~\ref{def:complexity0} and~\ref{def:complexity} accounts for the complexity of the channels used for the readout. 
\end{remark}

For our purposes, we now consider two specific lab oracles:  $\textsf{LOP}_\ell$ and $\textsf{LOQ}_\ell$\,.  We will give precise definitions of these lab oracles shortly, but first we explain what they are in words.  Here $\textsf{LOP}_\ell$ is a lab oracle whose quantum channel applies a new Haar random unitary 
to $\ell$ bits in $\fontH{L}$ each time the oracle is called (i.e., for each application of the channel a new Haar random unitary is picked independently from scratch).  By contrast, $\textsf{LOQ}_\ell$ is a lab oracle whose quantum channel implements a single Haar random unitary $U$, i.e.~the same $U$ is applied each time the oracle is called.  We now make these definitions of the lab oracles more explicit:

\begin{definition}[\bf The fixed random unitary lab oracle $\textsf{LOQ}_\ell$]\label{def:LOQ}
For $\ell$ a positive integer, we define the lab oracle 
$\textsf{LOQ}_\ell=(\mathcal{E}_{\fontH{NL}},\rho_{\fontH{N}})$ as follows:
\begin{itemize} 
\item $\fontH{N}$ is an $e^{\mathcal{O}(\ell)}$ qubit Hilbert space,
\item $\fontH{L}$ is an $\ell$ qubit Hilbert space,
\item $\mathcal{E}_{\fontH{NL}}$ is a unitary transformation such that, for $U$ a unitary on $\ell$ qubits and $\ket{U}$ a classical description of $U$ on $e^{\mathcal{O}(\ell)}$ qubits (specified to within exponentially many bits of accuracy in each entry),
\[\mathcal{E}_{\fontH{NL}}(\ket{U}\bra{U}\otimes\ket{\alpha}\bra{\alpha})=\ket{U}\bra{U}\otimes U\ket{\alpha}\bra{\alpha}U^\dagger\,,\]
\item $\rho_{\fontH{N}}=\int_{\text{Haar}}  \ket{U}\bra{U}dU\,.$
\end{itemize} 
\end{definition} 
\noindent We likewise have:
\begin{definition}[\textbf{The newly chosen random unitary lab oracle} $\textsf{LOP}_\ell$]\label{def:LOP}
For $\ell$ a positive integer, define 
$\textsf{LOP}_\ell=(\mathcal{E}_{\fontH{NL}}',\rho_{\fontH{N}}')$ as follows.
\begin{itemize} 
\item $\fontH{N}$ is an $n \sim e^{\mathcal{O}(\ell)}$ qubit Hilbert space,
\item $\fontH{L}$ is an $\ell$ qubit Hilbert space,
\item $\mathcal{E}_{\fontH{NL}}' = \mathcal{E}_{\fontH{NL}} \circ \mathcal{C}_{\fontH{N}}$ where
$$\mathcal{C}_{\fontH{N}}[\rho] = \int_{Haar}  \ket{U}\bra{U}dU$$
for any $\rho$ (i.e., it prepares the same state regardless of the input) where $U$ is a unitary on $\ell$ qubits, and $\ket{U}$ a classical description of $U$ on $n$ qubits (specified to within exponentially many bits of accuracy in each entry), and $\mathcal{E}_{\fontH{NL}}$ is the same as in the previous definition, i.e.~
\[\mathcal{E}_{\fontH{NL}}(\ket{U}\bra{U}\otimes\ket{\alpha}\bra{\alpha})=\ket{U}\bra{U}\otimes U\ket{\alpha}\bra{\alpha}U^\dagger\,,\]
\item $\rho_{\fontH{N}}'=|0\rangle \langle 0|^{\otimes n}\,.$
\end{itemize} 
\end{definition} 

We note that the above definitions provide only approximations of the unitaries $U$ involved, due to finite precision following from the finite number of bits in the description of the unitary; since there is no limit on the number of bits in $\fontH{N}$, this introduces an arbitrarily small error, which does not affect any of our calculations. We ignore this in our discussion.

With the definitions of $\textsf{LOP}_\ell$ and $\textsf{LOQ}_\ell$ at hand, it is easy to establish the following theorem:
\begin{thm}[{\bf Coherent access QUALM that probabilistically distinguishes $\textsf{LOP}_\ell$ from $\textsf{LOQ}_\ell$}] 
\label{thm:mainswap}
There exists a family of coherent access QUALMs, $\{\mathrm{QUALM}_\ell\}_{\ell=1}^\infty$,
over universal admissible gate sets such that
QUALM$_\ell$ distinguishes between the lab oracles $\textsf{LOP}_\ell$ and $\textsf{LOQ}_\ell$ with error $<1/3$.  Furthermore, this is achieved with $\mathcal{O}(\ell)$ QUALM gate complexity and $\mathcal{O}(1)$ QUALM query complexity.
\end{thm}

\begin{proof}
The idea is to implement the familiar
SWAP test between two applications
of the lab oracle, using a control 
qubit. This gives a QUALM which always outputs $1$ in the $\textsf{LOQ}_\ell$ case, but its probability to output $1$ in the $\textsf{LOP}_\ell$ case is less than $2/3$.  
Thus, the QUALM achieves the task with bias $\ge \frac{1}{3}$. (This can then be amplified using standard repetition as in Remark \ref{remark:amp} to obtain a QUALM achieving the distinguishing $\Task$ of Definition \ref{def:distinguish2} with error $< \frac{1}{3}$.) We now describe the QUALM. 

The coherent access QUALM we construct has no input classical bits: namely, $\fontH{S}_{\text{in}}$ is empty. $\fontH{W}$ is the union of register $\fontH{W}_1$ consisting of $\ell$ qubits, and $\fontH{W}_2$ consisting of a single qubit called the ``control qubit''.  $\fontH{S}_{\text{out}}$ consists of the single qubit in $\fontH{W}_2$.

The series of gates $\mathcal{Q}$ is as follows. We (i) apply the lab oracle superoperator on $\fontH{N}\otimes \fontH{L}$, (ii) swap the additional register and $\fontH{L}$ and again apply the lab oracle, and (iii) finally 
apply $\text{SWAP}_\text{TEST}$ on the two $\ell$ qubit registers using the control qubit: 
\begin{equation}
\mathcal{Q}=\text{SWAP}_\text{TEST}(\fontH{L},\fontH{W}_1)\circ\square \circ \prod_{i=1}^\ell \text{SWAP}_{i,\ell+i}\circ \square\,.
\end{equation} 
The channel $\text{SWAP}_\text{TEST}(\fontH{L},\fontH{W}_1)$ consists of 
first applying the Hadamard gate on the control qubit to bring it to the state $\ket{+}$; then conditioned on the control qubit being $1$, the $i$th qubit in register $\fontH{L}$ is swapped with the $i$th qubit in the additional register, for all $i$  from $1$ to $\ell$. Then we apply another Hadamard on the control qubit $\fontH{W}_2$.  

Let us calculate the distribution on $\fontH{W}_2$ in each of the cases.  In the case the lab oracle is $\textsf{LOP}_\ell$, the state of the $\fontH{L}$ and $\fontH{W}$ registers (namely  the 
first $2\ell$ qubits) just before the application of the SWAP test is a mixture over  $U\ket{0^\ell}\otimes U'\ket{0^\ell}$ with $U,U'$ being independently chosen Haar random $\ell$ qubit unitaries.
In the case the lab oracle is $\textsf{LOQ}_\ell$,
the state will be 
a mixture over  $U\ket{0^\ell}\otimes U\ket{0^\ell}$ with
$U$ being a Haar random $\ell$ qubit unitary. 
By simple algebra, 
the probability for the output to be 
 $0$ after applying the SWAP test on the tensor product state  $\ket{\alpha}\otimes\ket{\beta}$ is  
\[\frac{1+|\langle\beta|\alpha\rangle|^2}{2}.\]
In the case the lab oracle is $\textsf{LOQ}_\ell$, this is always $1$ 
(for any state $U\ket{0^\ell}\otimes U\ket{0^\ell}$ in the mixture); 
in the case the lab oracle is $\textsf{LOP}_\ell$\,, the 
state of the $\fontH{L}$ and 
$\fontH{W}$ registers is a mixture over  $\ket{\alpha}\otimes \ket{\beta}$ 
with $\ket{\alpha},\ket{\beta}$ being two 
independently chosen Haar random $\ell$ qubit states. We recall a useful version of Levy's inequality: 
\begin{lemma}[Levy's inequality~\cite{Ledoux1, Low1}]
For any fixed unit vectors $\ket{\alpha}$ and $\ket{\beta}$ on $\ell$ qubits, we have 
\begin{equation}
    \text{\rm Prob}_{\text{\rm Haar}}\left( \left||\langle \alpha | U |\beta\rangle|^2 - \frac{1}{D}\right| \geq \epsilon  \right) \leq 4 \,e^{-D \epsilon^2/(18 \pi^2)}
\end{equation}
where $D = 2^\ell$.
 \label{lem:levy}
\end{lemma} 
Levy's inequality implies that (for sufficiently large $\ell$) the probability for the inner product between two Haar random states to be $<\frac{1}{100}$ is doubly exponentially close to $1$. Thus the probability to get $0$ in the SWAP test (averaged over the Haar random choices) is $>1/3$. 

Since the bias of this QUALM, which uses two queries to the lab oracle, is constant, only constantly many repetitions are needed. Thus, the QUALM query complexity of the amplified QUALM is indeed a constant. The QUALM gate complexity of the final amplified QUALM is a constant times that of the gate complexity of the above QUALM, which is $\mathcal{O}(\ell)$ due to the SWAP gates and Hadamards; to this we need to add the postprocessing of counting the constantly many output bits to decide on the outcome, which only requires constantly many 
gates. This gives a total gate complexity of  $\mathcal{O}(\ell)$.
\end{proof}

We now consider the case of incoherent access QUALMs. 
Our first main result is: 

\begin{thm}[{\bf Incoherent access QUALMs cannot efficiently distinguish $\textsf{LOP}_\ell$ from $\textsf{LOQ}_\ell$}]
\label{thm:main}
Consider an incoherent access QUALM for the task of distinguishing $\textsf{LOP}_\ell$ and $\textsf{LOQ}_\ell$, which makes $k$ queries to the lab oracle, such that  $k<\left(2^{\ell}/\sqrt{6}\right)^{4/7}$. Then \rm{(i)} \textit{the output single qubit density matrices on $\fontH{S}_{\text{out}}$ for the two cases are within trace distance $\mathcal{O} (k^3/2^\ell)$;} \rm{(ii)} \textit{the QUALM thus distinguishes the 
two lab oracles with bias at most $\delta=\mathcal{O}(k^3/2^\ell)$; and} \rm{(iii)}
\textit{consequently, the QUALM query complexity for distinguishing the two lab oracles with constant error (Definition \ref{def:distinguish2}) is lower bounded by 
the exponential $\Omega(2^{2\ell/7})$.}
\end{thm} 

The above two theorems together imply that coherent access to lab oracle experiments can be exponentially more efficient, even when compared to QUALMs with incoherent {\it adaptive} access to the lab oracle.

The proof of Theorem \ref{thm:main} is technically heavy. 
The overall idea is to first consider the output distribution over the measurement outcomes, for all measurements performed when applying the QUALM on 
the chosen lab oracle. For the two different 
lab oracles, we arrive at two different distributions; we will show that these are exponentially close in total variation distance. From this it will be easy to deduce that the output density matrices of 
the QUALM are also exponentially close. 
However, for the adaptive case, the proof is very complicated. 
For didactic reasons, we consider first the much simpler case in which the access is non-adaptive (and also further somewhat restricted), and prove it in 
Subsection \ref{sec:nonadaptive}. The general adaptive scenario will be proved in Section~\ref{subsec:proof of main theorem}. 

\subsection{Proof of a restricted case of Theorem \ref{thm:main}: the parallel repetition case} \label{sec:nonadaptive} 

To gain more intuition, in this subsection we will prove Theorem~\ref{thm:main} for the following very restricted case.  In particular, the QUALM is non-adaptive, and moreover we require that the different applications of the lab oracle are performed in exactly the same simple way: 
A state $V|0^\ell\rangle$ is prepared on $\fontH{L}$
for some fixed unitary $V$, 
the lab oracle is applied, a fixed unitary $Y$ is applied on $\fontH{L}$, and finally $\fontH{L}$ is measured in the computational basis. The outcome density matrix of the QUALM is some deterministic function of all these measurement outcomes. We first write the probability distribution over all the measurement outcomes, in the case the lab oracle is $\textsf{LOP}_\ell$ or $\textsf{LOQ}_\ell$, respectively, as follows. 
The distributions are over $k$
$\ell$-bit strings $x_1,...,x_k$
where each $x_i\in \{0,1\}^\ell$\,:
\begin{align}
    P_k(x_1,...,x_k; Y, V) &= \prod_{i=1}^k\int_{\text{Haar}}\!\!dU \, |\langle x_i | Y U V |0^\ell\rangle|^2 \label{eq:Pk0}
    \end{align} 
    \begin{align}
    Q_k(x_1,...,x_k; Y, V) &= \int_{\text{Haar}}\!\!dU \, \prod_{i=1}^k |\langle x_i | Y U V |0^\ell\rangle|^2\,. 
\end{align}
We let $D = 2^\ell$. We prove in Lemma 
\ref{lem:paralleltvd} a bound on the total variation distance between the two distributions, as a function of $k$. 

\begin{lemma}\label{lem:paralleltvd} For any choice of $\ell$-qubit unitaries $V$ and $Y$, 
the total variational distance between the above two distributions is $\mathcal{O}(k^22^{-\ell/4}).$  
\end{lemma} 
Note that the above lemma only implies a bound of $\Omega (2^{\ell/8})$ on the QUALM complexity. This is a slightly weaker bound than the one proven in the next subsection for the general case, Theorem \ref{thm:main}, but is still exponential.  
\begin{proof} We distinguish between the case in which the $\ell\cdot k$-bit string has a collision and the case in which it does not.  We say that an $\ell\cdot k$-bit string $(x_1,...,x_k)$ has a {\it collision}
if there exist two distinct indices $i,j\in\{1,...,k\}$, $i\ne j$ such that the associated $\ell$-bit strings $x_i$ and $x_j$ are equal.  
Otherwise, we say the $\ell \cdot k$-bit string has {\it no collision}. 

{~}

\noindent{\bf No collision case:}
We first claim that $P_k$ is {\it flat} on $\ell \cdot k$-bit strings with no collision, namely, $P_k(x_1,...,x_k)=
P_k(x'_1,...,x'_k)$ for any two no-collision $\ell \cdot k$-bit strings 
$(x_1,...,x_k)$ and $(x'_1,...,x'_k)$. 
This is easy to see since $P_k$ is simply the uniform distribution over {\it all} $\ell \cdot k$-bit strings, and gives all of them the same probability $ \frac{1}{2^{\ell k}}$. 
This follows since by definition, $P_k$ is a product distribution over the different $x_i$'s\,; to calculate the distribution over 
$x_i\in \{0,1\}^\ell$ for a fixed $i$,  recall that the Haar random measure on unitaries is invariant under rotation by a fixed unitary from the left or from the right, hence 
\begin{equation}\int_{\text{Haar}}\!\!dU \, |\langle x_i | Y U V |0^\ell\rangle|^2=\int_{\text{Haar}}\!\!dU \, |\langle x_i | U |0^\ell\rangle|^2\,.\end{equation} 
Moreover, if we let $A$ be a unitary 
which swaps $\ket{x_i}$ with $\ket{x_j}$ and 
applies the identity on all other computational basis states, then $AU$ is also Haar random, 
and hence also  \[\int_{\text{Haar}}\!\!dU \, |\langle x_i | AU |0^\ell\rangle|^2 = \int_{\text{Haar}}\!\!dU \, |\langle x_j | U |0^\ell\rangle|^2. \]
We get that all $x_i$'s are equally probable, and 
\begin{equation}
P_k(x_1,...,x_k; Y, V)=\frac{1}{D^k}\,. 
\end{equation}
We now claim that all $\ell \cdot k$-bit strings with no collision are equally probable also with respect to $Q_k$ (in the case of $Q_k$, unlike for $P_k$, we do not know this to be true when there are collisions).  
Consider two such no-collision $\ell\cdot k$-bit strings, $(x_1,...,x_k)$ and $(x'_1,...,x'_k)$, 
namely, $x_1,...,x_k$ are all distinct, and so 
are $x'_1,...,x'_k$.
One can easily convince oneself that 
this implies there exists a permutation on the $2^\ell$ $\ell$-bit strings which takes $x_i$ to $x'_i$ for all $i\in\{1,...,k\}$. 
We now define a unitary $A$, acting on 
the Hilbert space of $\ell$ qubits, which applies exactly the permutation on the computational basis induced by the above permutation on $\ell$-bit strings; in particular, it satisfies: 
$A\ket{x_i}=\ket{x'_i}$ for all $i$. By definition, and using the fact that $A$ is unitary, we have 
\begin{equation}\langle x_i|A^{-1}=\bra{x'_i}\end{equation} for all 
$i\in\{1,...,k\}$. 
We can write 
\begin{equation}
    Q_k(x_1,...,x_k; Y,V) = \int_{\text{Haar}}\!\!dU \, \prod_{i=1}^k |\langle x_i |A^{-1}A Y U V |0^\ell\rangle|^2 =  \int_{\text{Haar}}\!\!dU \, \prod_{i=1}^k |\langle x'_i |A Y U V |0^\ell\rangle|^2\,.
\end{equation}
Since $U$ is Haar random so is $AYUV$ for 
any fixed unitaries $A,Y,V$, and so we get that the above expression is also equal to 
\begin{equation}
    \int_{\text{Haar}}\!\!dU \, \prod_{i=1}^k |\langle x'_i |U |0^\ell\rangle|^2= Q_k(x'_1,...,x'_k; Y,V)\,.
\end{equation}
This completes the proof that all no-collision $\ell \cdot k$-bit strings are equiprobable with respect to $Q_k$.  

{~}

\noindent{\bf Collision case:}
We now claim that the total probability for a ``collision'' is exponentially small in $\ell$, for both $P_k$ and $Q_k$. For $P_k$ the probability for a collision is at most $\binom{k}{2}$ times the probability for a collision between $x_i$ and $x_j$ for some fixed choice of $i\ne j$; this latter probability is exactly $\frac{1}{2^\ell}$, 
and hence if we let $\epsilon_P$ to be the total probability mass on all $\ell \cdot k$-bit strings in $P_k$ with collisions, we have
\begin{equation}
\epsilon_P\le \binom{k}{2}\frac{1}{2^\ell}\,.
\end{equation}

 For $Q_k$, we rely on concentration of measure arguments, and in particular, on Levy's inequality (Lemma \ref{lem:levy}). Consider the event, which we denote by $B$, 
 that  
 \begin{equation}\label{eq:coordinatebound}|\bra{x}YUV\ket{0}|^2< 2^{-\ell/4}
 \end{equation} 
 for all $x\in\{0,1\}^\ell$ simultaneously.  
 We will make two claims. 
 First, we claim that the event $B$ happens with all but doubly exponentially small probability; second, we claim that conditioned on $B$, the probability for a collision is exponentially small. 
\begin{claim} \label{cl:bound44} 
If $U$ is distributed Haar randomly, then 
$Pr_{U\sim{\text Haar}}(\neg B)$ is doubly exponentially small in $\ell$. 
\end{claim} 
\begin{proof} 
Apply Lemma \ref{lem:levy} to the Haar random matrix $YUV$. The probability that a given coordinate 
of $YUV\ket{0}$ is more than $\epsilon(D)$ 
is $\le 4\exp(-\Omega(D \epsilon(D)^2))$.     
Pick $\epsilon(D)$ to be $D^{-1/4}$.  We get that the chances are $<4\exp(-\Omega(\sqrt{D}))$ that a particular coordinate would be larger than $D^{-1/4}$. The claim follows by applying a union bound on all $D = 2^\ell$ $\ell$-bit strings. 
\end{proof}

To bound $\epsilon_Q$, the total probability for a collision in $Q_k$, we 
next upper bound the total probability of all $\ell\cdot k$-bit strings {\it with} a collision, conditioned that the event $B$ occured, namely 
that $|\bra{x}YUV\ket{0}|^2< 2^{-\ell/4}$ for all $\ell$-bit strings $x$.
For each such $U$, the probability (now probability is with respect 
to the quantum measurement, not the choice of $U$!) for a collision is 
at most the sum over $i=2,...,k$, i.e.~over the probability that the next 
string will be equal to one of the previous $i-1$ chosen strings.  By the union bound, this is at most $k^2 2^{-\ell/4}$. 
Together with Claim~\ref{cl:bound44} we arrive at the following upper bound on $\epsilon_Q$, the probability for a collision 
in $Q_k$: 
\begin{equation}\epsilon_Q\le \text{Pr}\left[\neg B\right]+\text{Pr}\left[\text{collision}|B\right]\le 2k^22^{-\ell/4}.
\end{equation}
This is sufficient for our purposes.

{~}

\noindent{\bf Bounding the TVD:} $P_k$ is thus flat on no-collision 
$\ell \cdot k$-bit strings, and has at most $\epsilon_P\le \binom{k}{2}\frac{1}{2^\ell}$ 
probability for $\ell \cdot k$-bit strings with collision. 
$Q_k$ is also flat on no-collision $\ell \cdot k$-bit strings, and has 
at most $\epsilon_Q\le 2k^22^{-\ell/4}$ probability for a collision. The TVD between the two distributions is at most 
\begin{equation}\label{eq:tvd}
    \sum_{x\text{ with no collision}} 
    \left|\frac{1-\epsilon_P}{N}-\frac{1-\epsilon_Q}{N}\right|+ \sum_{x\text{ with collision}} 
    |P_k(x)-Q_k(x)|\le 2\epsilon_P+2\epsilon_Q\le 2\left( \binom{k}{2} 2^{-\ell}+k^22^{-\ell/4}\right),
\end{equation}
where $N$ is the number of $\ell \cdot k$ bit strings $x$ without collision. 
\\ \\
This completes the proof of Lemma \ref{lem:paralleltvd}. 
\end{proof} 

We now deduce item (i) in Theorem \ref{thm:main} (in the restricted case of this subsection) 
from the above lemma, using the following  
simple fact, whose proof is just the triangle inequality:  
\begin{fact}\label{fact:shrink prob}
Let $\rho: \{0,1\}^m \longmapsto \mathcal{D}(\mathbb{C}^{2})$ 
be a function from $m$-bit strings to density matrices of a single qubit. 
Let $\mu_1,\mu_2$ be two probability distributions over $m$ bit strings, with 
total variation distance $\|\mu_1-\mu_2\|_1$. 
Then if we define for $i\in\{1,2\}$ 
\[\rho_i =\sum_x \mu_i (x) \rho(x)\,,\]
we have 
\[\|\rho_1-\rho_2\|_1\le \|\mu_1-\mu_2\|_1.\] 
\end{fact}

Item (ii) in the theorem then follows from the well-known fact (e.g., \cite{AharonovKitaevNisan}) which states that the total variation distance between output distributions of measurements applied to two density matrices is at most the trace distance between these density matrices. Finally, item (iii) in the theorem follows from Remark~\ref{re:frombiastoerror}.

\begin{remark}[{\it Comment on our choice of lab oracles}]
The astute reader may wonder why we chose the lab oracle $\textsf{LOQ}_\ell$ in Theorem \ref{thm:main} to be averaged over Haar random $U$.  One might ask: given Eqn.~\eqref{eq:coordinatebound}, can we not satisfy it by a single $U$? We note that 
the fact that the output probability distribution gives every string a very small probability (and thus the probability of a collision is very small) is not sufficient to imply hardness of the distinguishing task. 
Consider for example the lab oracle $\textsf{LOQ}'_\ell$ which provides samples from a uniform distribution over some fixed chosen subset $S$ of $\{0,1\}^\ell$
containing {\it half} of the $\ell$-bit strings (and set $Y=V=\mathds{1}$ for this example). Note that Eqn.~\eqref{eq:coordinatebound} is satisfied.  But of course, if we knew  the set $S$, we in fact could solve the problem of distinguishing between the two lab oracles with only $\mathcal{O}(1)$ samples! By just sampling constantly many times, we are likely to get a string not in $S$, in the case the lab oracle is $\textsf{LOP}_\ell$. The proof of Lemma \ref{lem:paralleltvd} does not go through in this example \textit{despite} that fact that the probability of collision is negligible, because Eqn.~\eqref{eq:tvd} does not hold: the TVD is not small since $\textsf{LOQ}'_\ell$ does not allow sampling from all $\ell\cdot k$-bit strings with no collision, but only from a tiny fraction of them. Technically, our proof does go through for the lab oracle $\textsf{LOQ}_\ell$ since in this case we pick a Haar random $U$ and all $\ell \cdot k$-bit strings with no collision are allowed. More intuitively and generally, the reason that in our case distinguishing is indeed impossible is that the lab oracle $\textsf{LOQ}_\ell$ still has an enormous amount of randomness 
in the choice of $U$, and thus one cannot tailor the rule of distinguishing in a fixed QUALM in such a manner that fits all the possible unitaries at once.
\end{remark}

\subsection{Proof of Theorem \ref{thm:main}}\label{subsec:proof of main theorem}

In this subsection we prove our main result, Theorem \ref{thm:main}. 
The proof of this fact consists of two steps. The first step is to study a special case, which we call a `simple measurement' (SM) QUALM, and show that it cannot distinguish the two lab oracles $\textsf{LOP}_\ell$, $\textsf{LOQ}_\ell$\,. The second step is to show that the output of a general incoherent access QUALM can be related to a probabilistic average  of those of SM QUALMs, so that a general incoherent access QUALM cannot distinguish the two lab oracles if no SM QUALM can do that. 

\subsubsection{Simple measurement QUALM}

Before providing the rigorous definition of the SM QUALM, we would like to explain the physical intuition.  The SM QUALM describes a simple case of an incoherent access QUALM where there is only one measurement carried out after each application of lab oracle. The measurements are of a particular form: 
each is a POVM in which each element is of rank one. 
In the $i$th round, the measurement
output $s_i$ is recorded in a new tensor factor of $\fontH{W}$, denoted by $\fontH{W}_i$\,. The $i$th POVM  is thus described by 
\begin{equation} 
\{\lambda^i_{s_0s_1...s_i}\ket{y^i_{s_0s_1...s_i}}\bra{y^i_{s_0s_1...s_i}}\}_{s_i},\text{  where  }
\sum_{s_i}\lambda^i_{s_0s_1...s_i}\ket{y^i_{s_0s_1...s_i}}\bra{y^i_{s_0s_1...s_i}}=\mathds{1},\,\,~0<\lambda^i_{s_0s_1...s_i}\leq 1\,.
\end{equation} 
We note that both $\ket{y^i_{s_0s_1...s_i}}$ and 
$\lambda^i_{s_0s_1...s_i}$ depend not only on $s_i$ but also on  
all previous measurement results $s_0,...,s_{i-1}$ or in short, 
$s_{j<i}$\,. After the measurement\footnote{Typically in quantum computation, one considers complete projective measurements and deduces from this the general case. However as we will see below, for the purpose of proving Theorem \ref{thm:main} we need to consider measurements that are slightly more general, where all POVM elements are of rank $1$, but they do not correspond to projections onto orthogonal states.}, $\fontH{L}$ is prepared into a mixed state $\sigma_{s_0s_1...s_i}^i$, which again can depend on the previous measurement results 
$s_1,...,s_{i-1}$. After the last measurement, a readout channel $\mathcal{C}_{\text{out}}$ is applied to $\fontH{W}$, which maps the diagonal density operator of $\fontH{W}$ to a single qubit output state.  The more rigorous definition is given in the following.

\begin{figure}
    \centering
    \includegraphics[width=2.2in]{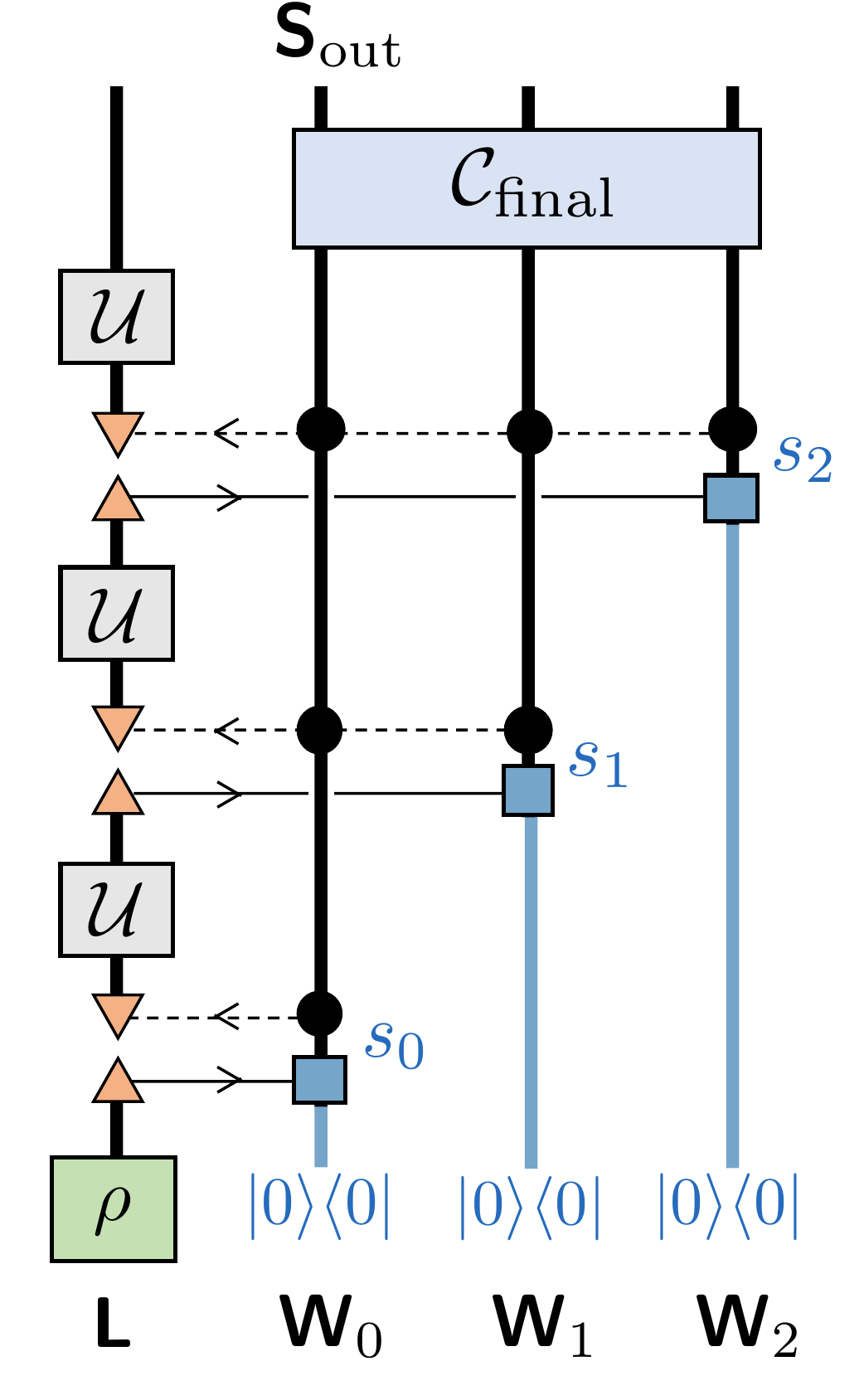}
    \caption{Illustration of the simple measurement QUALM defined in Definition \ref{def:SPM}. The upward pointing triangles indicate the weighted projection $\rho_{\fontH{L}}\rightarrow \bra{y_{s_0,...s_i}}\rho_{\fontH{L}}\ket{y_{s_0,...,s_i}}\lambda_{s_0,...,s_i}$. The horizontal solid lines indicate the recording of POVM measurement results in $\fontH{W}_i$. The horizontal dashed lines connected to downward pointing triangles indicate the preparation of initial state $\sigma_{s_0,...,s_i}$ controlled by previous measurement results $s_0,...,s_i$.}
    \label{fig:SPM}
\end{figure}

\begin{definition}\label{def:SPM} {\bf (Simple Measurement QUALM)}
The simple measurement QUALM is an incoherent access QUALM $(\fontH{L},\fontH{W},\mathcal{Q},\fontH{S}_{\text{out}})$ with
$\fontH{W}=\fontH{W}_0\otimes \fontH{W}_{1}\otimes \cdots \otimes \fontH{W}_k$ initialized at $\bigotimes_{i=0}^k\ket{0^{\fontH{W}_i}}$, and 
\[\mathcal{Q}=\mathcal{C}_{\text{final}}\circ 
\mathcal{Q}_{k} \circ \square \circ \mathcal{Q}_{k-1}\circ \square \circ\cdots \circ\square \circ\mathcal{Q}_{0}\,.\]   
Each $\mathcal{Q}_i$ applies to $\fontH{L}$ and $\fontH{W}_j,~j\leq i$. Since $\fontH{W}_i$ is initialized to $\ket{0^{\fontH{W}_i}}$, we only need to define the action of $\mathcal{Q}_i$ on states of the form $\rho_{\fontH{L}}\otimes\rho_{\fontH{W}_{j<i}}\otimes \ket{0^{\fontH{W}_i}}
\bra{0^{\fontH{W}_i}}$: 
\begin{align}
    \mathcal{Q}_{i}\left(\rho_{\fontH{L}}\otimes\rho_{\fontH{W}_{j<i}}\otimes\ket{0^{\fontH{W}_i}}\bra{0^{\fontH{W}_i}}\right)&=\sum_{s_{0},s_1,...,s_i}\lambda^i_{s_0s_1...s_i}\bra{y^i_{s_0s_1...s_i}}\rho_{\fontH{L}}\ket{y^i_{s_0s_1...s_i}}\sigma^i_{s_0s_1...s_i}\nonumber\\
    & \qquad ~\otimes K_{s_0s_1...s_{i-1}}\rho_{\fontH{W}_{j<i}}K_{s_{0}s_1...s_{i-1}}\otimes\ket{s^{\fontH{W}_i}_{i}}\bra{s^{\fontH{W}_i}_{i}}\label{eq:rank1POVM}
    \end{align}
    with $s_j=1,2,...,D_m$ a set of labels for each $j=0,1,...,k$;  $s_{j<i}\equiv \left\{s_0,s_1,...,s_{i-1}\right\}$ refers to the previous measurement results;
    \begin{equation} 
    0<\lambda^i_{s_0s_1...s_i}\leq 1\,;
    \end{equation}
    the vectors $\ket{y^i_{s_0s_1...s_i}}$ and 
    values $\lambda^i_{s_0s_1...s_i}$
    satisfy 
    \begin{equation}
\sum_{s_i}\lambda^i_{s_0s_1...s_i}\ket{y^i_{s_0s_1...s_i}}\bra{y^i_{s_0s_1...s_i}}=\mathds{1}\,;\label{eq:condlambda}
\end{equation}
$\sigma^i_{s_0s_1...s_i}$ are density matrices of $\fontH{L}$; $\fontH{W}_{j<i}\equiv \fontH{W}_0\otimes \fontH{W}_1\otimes\cdots\otimes \fontH{W}_{i-1}$; $K_{s_{j<i}}=\bigotimes_{j=0}^{i-1}\ket{s_j^{\fontH{W}_j}}\bra{s_j^{\fontH{W}_j}}$ is the projector to a particular record in $\fontH{W}_{j<i}$\,; and $\ket{s_j^{\fontH{W}_j}}$ is an orthornormal basis of $\fontH{W}_j$. Finally, $\fontH{S}_{\text{out}}$ is the first qubit in $\fontH{W}$, and $\mathcal{C}_{\text{final}}: \mathcal{D}(\fontH{W})\rightarrow \mathcal{D}(\fontH{W})$ is a general quantum channel. 
\end{definition}

The definition is illustrated in Fig.~\ref{fig:SPM}. Physically, $\mathcal{Q}_i$ is measuring the state of $\fontH{L}$ with the POVM defined by Eqn.~\eqref{eq:condlambda} and then based on the outcome $s_i$, prepares a new state $\sigma^i_{s_0s_1...s_i}$ which is in general a mixed state. The measurement basis and the state prepared depend also on the previous measurement outputs $s_0,s_1,...,s_{i-1}$, which we read from the state of $\fontH{W}_{j<i}$. It should be noted that the states $\ket{y_{s_0s_1...s_i}^i}$ are normalized but are not necessarily orthogonal to each other.
Due to the condition given in Eqn.~\eqref{eq:condlambda}, in general the range $D_m$ (where `$m$' stands for `measurement') of $s_i$ is greater than or equal to the $\fontH{L}$ Hilbert space dimension $D=2^\ell$.  Without loss of generality, we assume that $D_m$ is the same for each $i$.  If we trace over $\fontH{W}$, the reduced channel on $\fontH{L}$ is what is known as a measure-and-prepare channel.

For the SM QUALM, the output density operator is determined by applying the readout channel $\mathcal{C}_{\text{final}}$ to the diagonal density operator of $\fontH{W}$, which encodes the probability distribution of the measurement results $s_0,s_1,...,s_k$. For the fixed random unitary lab oracle $\textsf{LOQ}_\ell$ (Definition \ref{def:LOQ}), the probability distribution is
\begin{align}
    Q_k\left(s_0,s_1,...,s_k\right)= \text{Pr}(s_0)\cdot\left(\int_{\rm Haar} dU\prod_{i=1}^k\bra{y_{s_0s_1...s_i}^i}U\sigma_{s_0s_1...s_{i-1}}^{i-1}U^\dagger\ket{y_{s_0s_1...s_i}^i}\lambda_{s_0s_1...s_i}^i\right) \label{eq:Qk}
\end{align}
with $\text{Pr}(s_0)=|\bra{y_{s_0}^0}0^\fontH{L}\rangle|^2 \lambda^0_{s_0}$.  The probability distribution for the other lab oracle $\textsf{LOP}_\ell$ (Definition \ref{def:LOP}) is
\begin{align}
    P_k\left(s_0,s_1,...,s_k\right)=\text{Pr}(s_0)\left(\prod_{i=1}^k\int_{\rm Haar} dU_i\bra{y_{s_0s_1...s_i}^i}U\sigma_{s_0s_1...s_{i-1}}^{i-1}U^\dagger\ket{y_{s_0s_1...s_i}^i}\lambda^i_{s_0s_1...s_i}\right). \label{eq:Pk}
\end{align}
We note that this probability is 
equal to 
\begin{equation}
    P_k\left(s_0,s_1,...,s_k\right)=\text{Pr}(s_0)D^{-k}\prod_{i=1}^k\lambda_{s_i}^i\,.
\end{equation} 
In the following we will often denote
the ordered sequence $s_0,s_1,...,s_k$ by $s$ for simplicity.

Our conclusion is that these two probability distributions $Q_k$ and $P_k$ are difficult to distinguish, which is given in the following lemma:

\begin{lemma}\label{thm:tvd}
Let $Q_k$ and $P_k$ be defined by Eqn.'s~\eqref{eq:Qk} and~\eqref{eq:Pk} above, where $\ket{y_{s_0s_1...s_i}^i}$ are normalized pure states in $\fontH{L}$, and $\sigma_{s_0s_1...s_{i-1}}^{i-1}$ are normalized density operators in the same Hilbert space. We have $0\le \lambda_{s_0s_1...s_i}^i\le 1$ and $\sum_{s_i}\lambda_{s_0s_1...s_i}^i\ket{y_{s_0s_1...s_i}^i}\bra{y_{s_0s_1...s_i}^i}=\mathds{1}$ (condition (\ref{eq:condlambda})). $\fontH{L}$ consists of $\ell$ qubits, and $k<\left(2^{\ell}/\sqrt{6}\right)^{4/7}$.  $s_j=1,2,...,D_m$ is a set of labels for each $j=0,1,...,k$.  Then 
\begin{align}
    \|P_k-Q_k\|_1=\sum_{s}\left|P_k\left(s\right)-Q_k\left(s\right)\right|\leq \mathcal{O}\left(\frac{k^3}{2^\ell}\right)\,.
\end{align}
\end{lemma}

Lemma \ref{thm:tvd} will be proved in the next subsection. If we let $\rho(s)=\text{tr}_{\overline{\fontH{S}}_{\text{out}}}\!\left(\mathcal{C}_{\text{final}}\!\left[\bigotimes_{i=0}^k\ket{s_i^{\fontH{W}_i}}\bra{s_i^{\fontH{W}_i}}\right]\right)$, the output state for each of the lab oracles $\LOP_\ell$ and $\LOQ_\ell$ (namely the final state on $\fontH{S}_{\text{out}}$)  is
\begin{align}
    \rho_{\rm final}^P=\sum_{s}P_k(s)\rho(s),\qquad\rho_{\rm final}^Q=\sum_sQ_k(s)\rho(s)\,.
\end{align}
Based on Fact \ref{fact:shrink prob} in the previous section as well as Lemma \ref{thm:tvd}, we obtain $\|\rho_{\rm final}^P-\rho_{\rm final}^Q\|_1\leq \|P_k-Q_k\|_1\leq \mathcal{O}\!\left(\frac{k^3}{2^\ell}\right)$. This proves Theorem~\ref{thm:main} in the restricted case of simple measurement QUALMs. After presenting the proof in the next subsection, we will show how to deduce the full theorem from the SM case, in Section~\ref{subsubsection:general incoherent}.

\subsubsection{Proof of Lemma~\ref{thm:tvd}}
\label{sec:sketchofproof}

We write the probability distribution $Q_k$ in Eqn.~\eqref{eq:Qk} as follows: 
\begin{align}
\label{eq:defAB}
    &Q_k\left(s\right)= \text{Pr}(s_0) \int_{\rm Haar}dU \,\text{tr}\left(U^{\otimes k}A_sU^{\dagger \otimes k}B_s\right)
\end{align}
where  
\begin{align}
A_s=\bigotimes_{i=1}^k \sigma_{s_0s_1...s_{i-1}}^{i-1},\quad B_s=\bigotimes_{i=1}^k \ket{y_{s_0s_1...s_i}^i}\bra{y^i_{s_0s_1...s_i}}\lambda_{s_0s_1...s_i}^i\,.
\end{align}
Our goal is to prove that $Q_k(s)$ is close to $P_k\left(s\right)$ as in Eqn.~\eqref{eq:Pk}. A simple fact is that 
\begin{equation}P_k(s)=\text{Pr}(s_0)D^{-k}{\rm tr}(B_s).\end{equation} The first step is to carry the integration over $U$ using Weingarten functions:
\begin{align}
    \int_{\rm Haar}dU \left[U^{\otimes k}\right]_{IJ}\left[U^{*\otimes k}\right]_{KL}=\sum_{\sigma,\tau\in S^k}\tau_{KI}\sigma_{LJ}W(\tau \sigma^{-1},D)\,.\label{eq:weingartendecomposition}
\end{align}
Here $I,J,K,L$ label an orthogonal basis in the $k$-copied Hilbert space. In a tensor product basis (e.g.~the computational basis), we can take $I=\left\{x_1x_2...x_k\right\}$ with $x_j=1,2,...,D$ and $j=1,2,...,k$.  The action of the permutation group elements $\sigma$ and $\tau$ corresponds to the permutation of different Hilbert space copies. For example, for $k=2$ a pair permutation is defined by $\sigma (|\psi_1\rangle \otimes |\psi_2\rangle) = |\psi_2\rangle \otimes |\psi_1\rangle$. Many useful mathematical properties of the Weingarten functions are known, which will play an essential role in the proof. We will leave the mathematical details of Weingarten functions to~\ref{App:reviewHaar}.

Using Eqn.~\eqref{eq:weingartendecomposition}, $Q_k$ can be rewritten as (see Subsection \ref{sec:IntroWeingerten})
\begin{align}
    Q_k(s)&=\sum_{\sigma, \tau\in S^k}\text{tr}(A_s\sigma)\,\text{tr}(B_s\tau^{-1})W(\tau\sigma^{-1},D)\,.\label{eq:Qk Weingarten}
\end{align}
The biggest term in this sum corresponds to $\tau=\sigma=\mathds{1}$ since $\text{tr}(A_s\sigma)$ and $\text{tr}(B_s\tau^{-1})$ are both bounded by $1$, and since the Weingarten function is maximal at $\tau\sigma^{-1}=\mathds{1}$. The sum in Eqn.~\eqref{eq:Qk Weingarten} consists of three kinds of terms: (i) $\tau=\sigma=\mathds{1}$; (ii) $\tau=\mathds{1},~\sigma\neq \mathds{1}$; (iii) $\tau\neq \mathds{1}$. This leads to the following inequality:
\begin{align}
    \left|Q_k(s)-P_k(s)\right|&\leq \left|W(\mathds{1},D)-D^{-k}\right|{\rm tr}(B_s)+\sum_{\sigma\neq \mathds{1}} \left|W(\sigma^{-1},D)\right|\left|{\rm tr}(A_s\sigma)\right|{\rm tr}(B_s)\nonumber\\
    &\qquad \qquad \qquad +\sum_{\tau\neq\mathds{1}}\sum_\sigma \left|W(\tau\sigma^{-1},D)\right|\left|{\rm tr}(A_s\sigma)\right|\left|{\rm tr}(B_s\tau^{-1})\right|\,.
    \label{eq:threeterms1}
\end{align}
 Using 
\begin{align}
    \left|{\rm tr}(A_s\sigma)\right|\leq 1\label{eq:mathfact0}
\end{align} 
we can further bound the right-hand side by
\begin{align}
    \left|Q_k(s)-P_k(s)\right|&\leq \left[\left|W(\mathds{1},D)-D^{-k}\right|+\sum_{\nu\neq\mathds{1}}\left|W(\nu,D)\right|\right]{\rm tr}(B_s)+\sum_\nu\left|W(\nu,D)\right|\sum_{\tau\neq \mathds{1}}\left|{\rm tr}\left(B_s\tau^{-1}\right)\right|\,.
    \label{eq:threeterms2}
\end{align}
Carrying out the sum over $s$ and using the fact that $\sum_s B_s=\mathds{1}$, and $\text{tr}(\mathds{1})=D^k$ (where the identity is over the $\ell \cdot k$ qubits on which $B_s$ acts) we obtain the following bound of the $1$-norm distance between the two probability distributions:
\begin{align}
\label{E:maintobound1}
\delta\left(P_k,Q_k\right)&\equiv \sum_s\left|Q_k(s)-P_k(s)\right|\leq c_1+c_2T \\
c_1&=D^k\left|W(\mathds{1},D)-D^{-k}\right|+D^k\sum_{\nu\neq\mathds{1}}\left|W(\nu,D)\right|\\
c_2&=D^k\sum_\nu\left|W(\nu,D)\right|\\
T&=\frac{1}{D^k}\sum_{\tau\neq \mathds{1}}\sum_s\left|{\rm tr}\left(B_s\tau^{-1}\right)\right|\,.
\end{align}
In the regime  $k<\left(2^{\ell}/\sqrt{6}\right)^{4/7}$, we are going to prove
\begin{align}
    c_1&=\mathcal{O}\!\left(\frac{k^{7/2}}{D^2}\right)\label{eq:mathfact1}\\
    c_2&= 1+\mathcal{O}\!\left(\frac{k^2}D\right)\label{eq:mathfact2}\\
    T&\leq \frac{k^3}{D}+\frac{k^2}{D}+\mathcal{O}\!\left(\frac{k^5}{D^2}\right)\,.
    \label{eq:mainlemma}
\end{align}
These together will prove Lemma \ref{thm:tvd}, with the dominant term being $c_2 \, T$. 

Eqn.'s~\eqref{eq:mathfact1} and~\eqref{eq:mathfact2} are based on relatively simple mathematical facts about Weingarten functions (with the condition $k<\left(2^{\ell}/\sqrt{6}\right)^{4/7}$), which we derive in~\ref{App:reviewHaar}. In the main text we will focus on the more nontrivial proof of Eqn.~\eqref{eq:mainlemma}. The proof is based on the following inequality for each nontrivial $\tau$:
\begin{align}
\sum_{s}\left|{\rm tr}\left(B_s\tau^{-1}\right)\right|\leq D^{k-\left\lfloor\frac{L_\tau}2\right\rfloor}\,.
\label{eq:keylemma}
\end{align}
Here $L_\tau$ is the total length of nontrivial cycles in $\tau$, and $\left\lfloor\frac {L_\tau}2\right\rfloor$ is the greatest integer that is less than or equal to $L_\tau/2$. For example, $L_\tau=2$ corresponds to a single pair transmutation and more generally $L_\tau=k$ indicates that a cyclic permutation of length $k$ is the longest cycle.

We write the above term more explicitly by using the definition of $B_s$ in Eqn.~\eqref{eq:defAB}:
\begin{align}
    \left|{\rm tr}\left(B_s\tau^{-1}\right)\right|=\prod_{i=1}^k\left|\braket*{y_{s_0s_1...s_{\tau(i)}}^{\tau(i)}}{y_{s_0s_1...s_i}^i}\right|\lambda_{s_0s_1...s_i}^i
\end{align}
with $\tau(i)$ the index that $i$ is permuted to. (It should be noted that $\tau$ only acts on $i=1,2,...,k$ since they correspond to measurements following a lab oracle. The state $i=0$ does not appear in the permutation.)
This equation is illustrated in Fig.~\ref{fig:xtaux}. The subtlety is that the measurement basis choice can be adaptive, which means $\ket{y^i_{s_i}}$ generically depends on the measurement results $s_j$ for $j<i$ (note that $\ket{y^{\tau(i)}_{s_{\tau(i)}}}$ 
may depend on $s_j$ for $j<\tau(i)$). This dependency makes the sum over $s_i$ more difficult than the non-adaptive case. 

\begin{figure}
    \centering
    \includegraphics[width=3.5in]{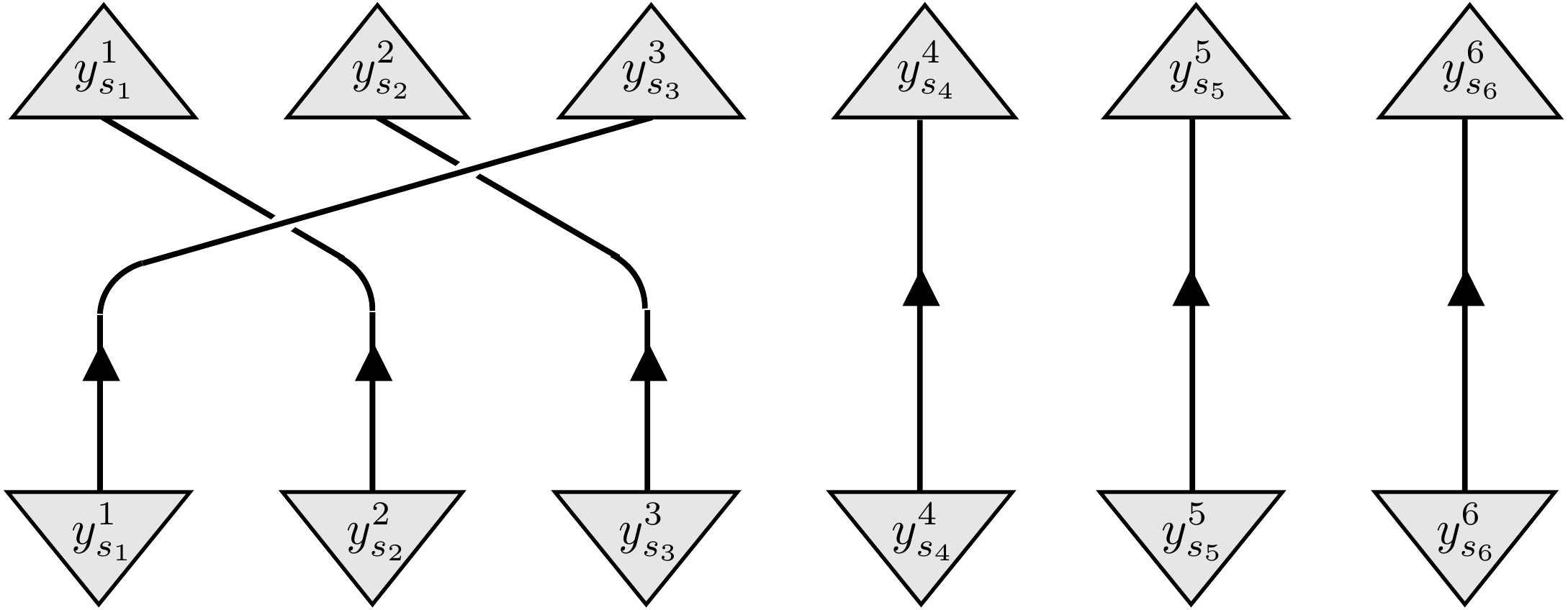}
    \caption{Illustration of the term $\prod_{i=1}^k\braket*{y_{s_0s_1...s_{\tau(i)}}^{\tau(i)}}{y_{s_0s_1...s_i}^i}$ for the example of a cyclic permutation $(123)$. For simplicity of notation we have denoted $y_{s_0s_1...s_i}^i$ by $y_{s_i}^i$ in the diagram.}
    \label{fig:xtaux}
\end{figure}

In the following we denote
\begin{align}
    M_{ji}=\braket*{y_{s_0s_1...s_j}^j}{y_{s_0s_1...s_i}^i}\label{eq:Mji}
\end{align}
where $M_{ji}$ depends on all the indices $s_a$ for $a\leq i$ or $a\leq j$. Each permutation $\tau$ can be decomposed into cyclic permutations. To illustrate the proof, let us consider an example case. For $n=8$ we can consider 
\begin{align}
    \tau=(175462)(3)(8)\label{eq:tauexample1}
\end{align}
which maps $175462$ cyclically to $754621$ and preserves $3,8$. For this $\tau$ we have 
\begin{align}
    \left|{\rm tr}\left(B_s\tau^{-1}\right)\right|=\left|M_{71}M_{57}M_{45}M_{64}M_{26}M_{12}\right|\prod_{i=1}^8\lambda_{s_0s_1...s_i}\,.
\end{align}
We illustrate a permutation by a curve in Fig.~\ref{fig:loops}. The $\tau$ we consider in Eqn.~\eqref{eq:tauexample1} corresponds to the subfigure (a). The $y$ values of the curve are the site labels $175462$, along with an additional $1$ returning to the starting point. Each link in this curve which goes from $y$-coordinates $i$ to $j$ corresponds to a term $M_{ji}$. When there is more than one loop, we draw a separate curve for each loop. For example, in Fig.~\ref{fig:loops}(b) we draw the curves for $\tau=(175462)(398)$. The total number of $M_{ji}$ (with $j\neq i$) in ${\rm tr}\left(B_s\tau^{-1}\right)$ is equal to the total length of the nontrivial loops, which we denote as $L_\tau$. For instance, $L_\tau=6$ for the example in Eqn.~\eqref{eq:tauexample1}.

\begin{figure}
    \centering
    \includegraphics[width=5.5in]{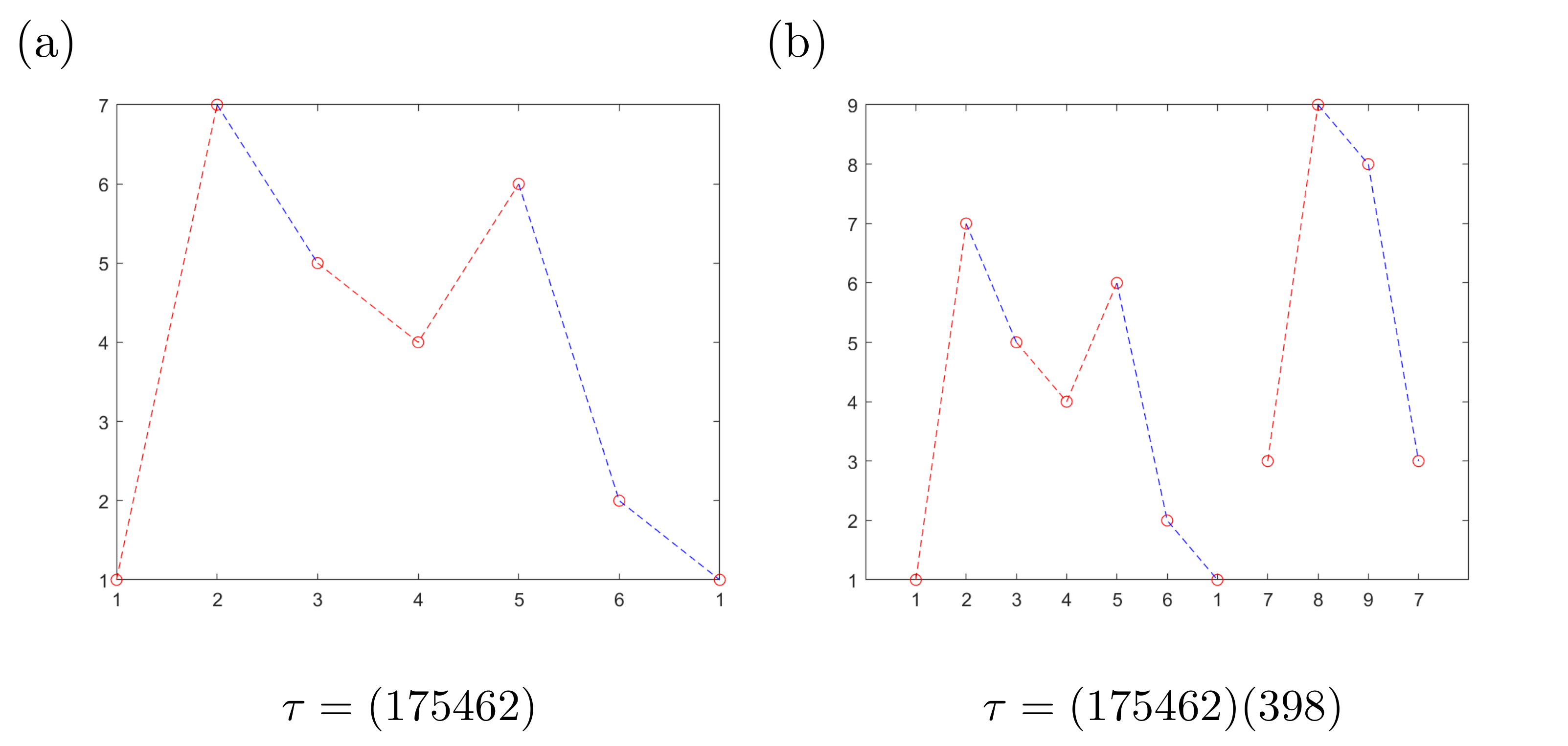}
    \caption{Illustration of a permutation. Each cycle in the permutation corresponds to a curve on the figure, starting and ending at the same lowest index. Each link in the figure corresponds to a term $M_{ji}$ in Eq. (\ref{eq:Mji}), with $j,i$ given by the $y$ value of the two end points. The links can be partitioned into two parts, colored by red and blue. The rules of the partition are that (1) the total length of red and blue links are only different at most by $1$; and (2) no segment contains a local maximum ($\Lambda$ shape). }
    \label{fig:loops}
\end{figure}

Now we define a partition of the $M_{ji}$ terms, which corresponds to the links in red and blue, respectively. The partition is required to satisfy the following two rules:
\begin{enumerate}
    \item The number of red links $l_r$ and the number of blue links $l_b$ are different by no more than $1$. 
    \item No segments contain a local maximum ($\Lambda$ shape).
\end{enumerate}
The second rule means that each local maximum (i.e.~cusp points, such as $7$ and $6$ in Fig.~\ref{fig:loops}(a)) has to border between red and blue links. By splitting at the cusp points, we obtain ``V'' shapes. Then we can split each ``V'' shape into two parts which are only different in length by $0$ or $\pm 1$. By choosing the location of the cut, we can make sure that the net length of the red and blue segments are equal (if the total loop length is even) or different by $1$ (if the total loop length is odd). Since $l_b+l_r=L_\tau$, without loss of generality we choose 
\begin{align}
    l_r=\left\lfloor \frac {L_\tau}2\right\rfloor\text{ and  }l_b=L_\tau-l_r\,.\label{eq:lblr}
\end{align}
Now we write
\begin{align}
    \left|{\rm tr}(B_s\tau^{-1})\right|&=\left|F_1\right|\left|F_2\right|\prod_{i=1}^8\lambda^i_{s_i}\\
    F_1&=\prod_{\langle ji\rangle=\text{red}}M_{ji}\\
    F_2&=\prod_{\langle ji\rangle=\text{blue}}M_{ji}
\end{align}
and use the inequality
\begin{align}
    \sum_s\left|{\rm tr}(B_s\tau^{-1})\right|&=\sum_s\left|F_1\right|\prod_{i=1}^8\lambda_{s_0s_1...s_i}\left|F_2\right|\leq \frac12\sum_s\prod_{i=1}^8\lambda_{s_0s_1...s_i}\left(|F_1|^2+|F_2|^2\right)\,.
    \label{eq:A1A2}
\end{align}
For instance, in the example $\tau=(175462)$ with the partition in Fig. \ref{fig:loops}(a), we have
\begin{align}
    F_1&=M_{71}M_{64}M_{45}\nonumber\\
    F_2&=M_{57}M_{26}M_{12}\,.
\end{align}

Now considering $\sum_s|F_1|^2$, we carry out the sum in reverse time order, starting from $s_8$. The sum over $s_8$ leads to a factor of $D$ since taking the trace of condition~\eqref{eq:condlambda} leads to $\sum_{s_i}\lambda_{s_0s_1...s_i}^i=D$. Then we can sum over $s_7$. Since the $M_{ji}$ for $i<7,j<7$ are independent of $s_7$, the sum can be carried out easily:
\begin{align}
    \sum_{s_7}\left|M_{71}\right|^2\lambda_{s_0s_1...s_7}^7\lambda_{s_0s_1}^1=\sum_{s_7}\braket*{y_{s_0s_1s_2...s_7}^7}{y_{s_0s_1}^1}\braket*{y_{s_0s_1}^1}{y_{s_0s_1s_2...s_7}^7}\lambda^7_{s_0s_1s_2...s_7}\lambda^1_{s_0s_1}=\lambda^1_{s_0s_1}\,.
\end{align}
Following the same reasoning, we can carry out the remainder of the sum using condition~\eqref{eq:condlambda}. We obtain
\begin{align}
    \sum_s|F_1|^2=D^5\,.
\end{align}

The same technique applies to a general permutation. In general, each red or blue segment contributes one factor of $D$ in the sum. Each segment with total length $l_\alpha$ (defined by the number of {\it links} in the segment) contains a sum over $l_\alpha+1$ sites that contribute $D$, which means it leads to a suppression by $D^{-l_\alpha}$ compared with the trivial permutation case $\sum_s\prod_{i=1}^k\lambda_{s_0s_1...s_i}^i=D^k$. Consequently,
\begin{align}
    \sum_s\left|F_1\right|^2=D^{k-l_r},\qquad \sum_s\left|F_2\right|^2=D^{k-l_b}
\end{align}
with $l_r$ and $l_b$ the total length of the red and blue segments, respectively. Interestingly, the result only depends on the net length of blue and red segments $l_b,l_r$ and is {\it independent} of the number of segments. Therefore, using Eqn.'s~\eqref{eq:A1A2} and~\eqref{eq:lblr} we obtain
\begin{align}
    \sum_s\left|{\rm tr}\left(B_s\tau^{-1}\right)\right|\leq \frac12D^k\left(D^{-l_b}+D^{-l_r}\right)\leq D^{k-\left\lfloor\frac {L_\tau}2\right\rfloor}\,.\label{eq:inequality}
\end{align}
Performing the sum over $\tau$ we find
\begin{align}
    T&\equiv \frac{1}{D^k}\sum_{\tau\neq \mathds{1}}\sum_s\left|{\rm tr}\left(B_s\tau^{-1}\right)\right|\leq \frac{1}{D^k}\sum_{\tau\neq \mathds{1}}D^{k-\left\lfloor\frac {L_\tau}2\right\rfloor}=\frac{1}{D^k}\sum_{L=2}^{k}N(k,L)D^{k-\left\lfloor\frac {L}2\right\rfloor}
\end{align}
with $N(k,L)$ the number of permutation elements with $L_\tau=L$. It is easy to see that 
\begin{align}
    N(k,L)\leq \left(\begin{array}{c}k\\L\end{array}\right)L!=\frac{k!}{(k-L)!}<k^L
\end{align}
and thus
\begin{align}
    T&< \sum_{L=2}^{\infty}k^LD^{-\left\lfloor\frac{L}2\right\rfloor}=\frac{(1+k)\frac{k^2}D}{1-\frac{k^2}D}=\frac{k^3}{D}+\frac{k^2}D+\mathcal{O}\left(\frac{k^5}{D^2}\right)\,.
\end{align}
This proves Eqn.~\eqref{eq:mainlemma}. (Note that in the regime we are considering, $\frac{k^3}{D}<O\left(D^{-1/7}\right)$ is exponentially small.)

This completes the proof of Lemma \ref{thm:tvd}.

\subsubsection{The general incoherent access QUALM}\label{subsubsection:general incoherent}

\begin{figure}
    \centering
    \includegraphics[width=2.5in]{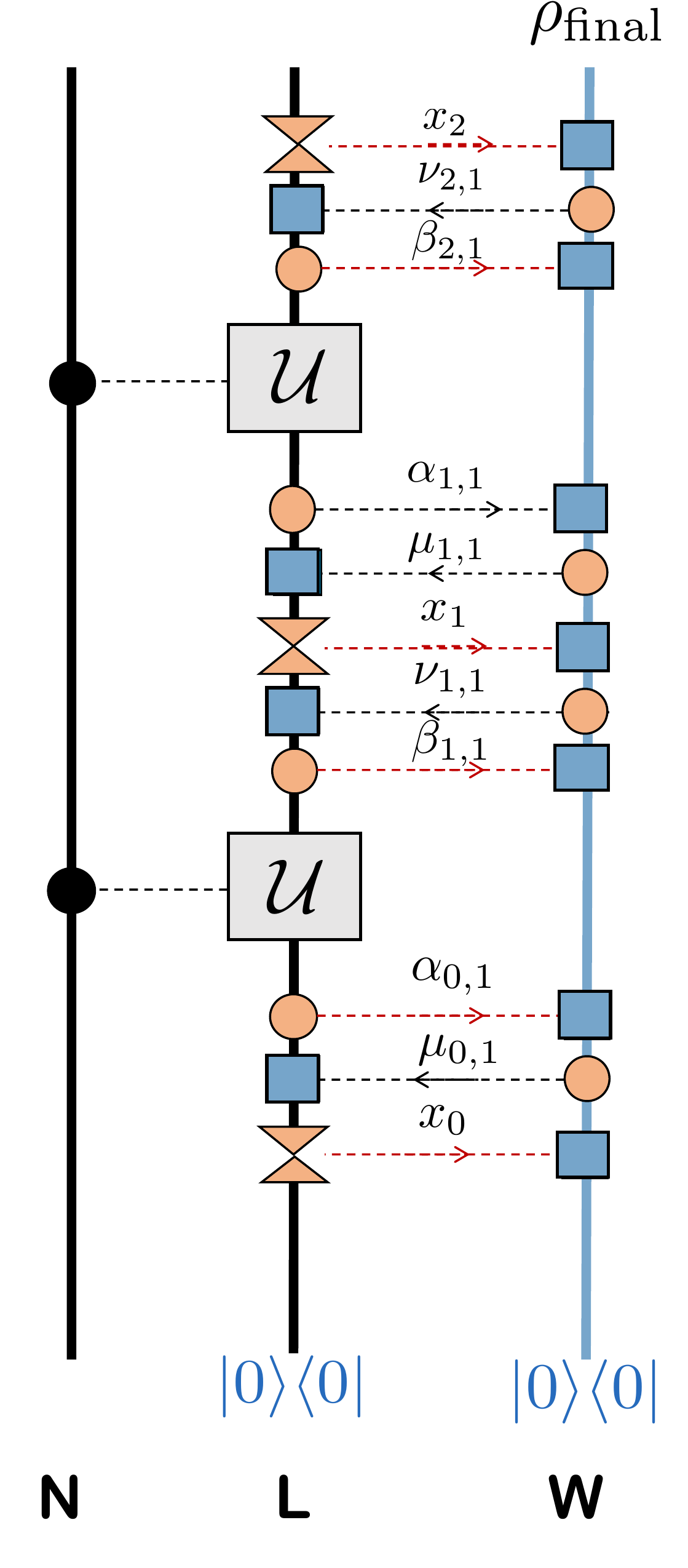}
    \caption{Illustration of the incoherent access QUALM defined in Definition~\ref{def:incoherent}. Each orange solid circle is a CP map and each blue box is a CPTP map. At least one of the CP maps between each application of the lab oracle is a complete measurement, indicated by the double triangle. The direction of the arrow in each horizontal dashed line indicates the direction of classical information flow. Only the LOCC's corresponding to $\beta_{i,j}$ and $x_i$ lead to conditional probabilities that depend on the lab oracle, which are indicated by red dashed lines.}
    \label{fig:incoherent}
\end{figure}

To complete the proof of Theorem \ref{thm:main},  we show that a general incoherent access QUALM cannot do better than a simple measurement QUALM. The key idea is the following. Since in an incoherent access QUALM $\fontH{L}$ and $\fontH{W}$ are only coupled by an LOCC, the output probability distribution of $\fontH{W}$ can be expressed as a sequence of conditional probabilities, each given by an LOCC. Then we can organize this chain of conditional probabilities into two groups: those that depend on the lab oracle, and those that do not. The latter can be expressed as a probabilistic average over a set of deterministic functions (i.e.~with conditional probability that is $1$ on a particular output and $0$ everywhere else), and for each deterministic function the QUALM reduces to a simple measurement QUALM.

More precisely, this is expressed in the following claim: 
\begin{claim}[\textit{Derandomization claim}] \label{cl:derandomization} 
For any incoherent access QUALM denoted by $\QUALM$, there is a set of 
simple measurement QUALMs, $\{\mathrm{QUALM}_{r}\}_{r}$, and a fixed probability distribution over the indices $r$ of $\Pi(r)$ (which crucially is independent of the lab oracle) 
such that 
\begin{itemize} 
\item The QUALM query complexity of each QUALM$_r$ is the same as that of QUALM;
\item For any lab oracle
$\LO$, the output of QUALM applied to $\LO$
is an average over the output of each QUALM$_{r}$ applied to 
$\LO$, where the average is taken with respect to the distribution 
$\Pi(r)$.
\end{itemize} 
\end{claim} 

From this claim, it follows that if QUALM achieves the task of distinguishing the two lab oracles $\textsf{LOP}_\ell$ and $\textsf{LOQ}_\ell$ with error $<\frac{1}{3}$ (which means that the output probability distributions for the two lab oracles are at least $\frac{1}{3}$ apart in total variation distance; see Definition \ref{def:distinguish2}) then 
there must be a simple measurement QUALM of 
the same query complexity 
which achieves the same total variation distance 
between the two output distributions. 
By Lemma \ref{thm:tvd}, this is impossible  
if the query complexity $k$ of QUALM is smaller than order $2^{2\ell/7}$, because in this case the two output distributions of the simple measurement QUALM, for the two lab oracles, are exponentially close in $\ell$. 

We would like to emphasize that Claim~\ref{cl:derandomization} applies to all incoherent access QUALMs with arbitrary lab oracles. To prove Theorem~\ref{thm:main} we only need to apply it to the specific lab oracles $\textsf{LOP}_\ell$ and $\textsf{LOQ}_\ell$, but later in Section~\ref{sec:symmtask} we will apply it to other cases.

It thus remains to prove Claim~\ref{cl:derandomization} 
to finish the proof of Theorem 
\ref{thm:main}. 

Towards this goal, we introduce some notation. 
An incoherent access QUALM is illustrated in Fig.~\ref{fig:incoherent}. We have labeled the complete measurement results in the $i$th round by bit strings $x_i$. All other one-round LOCC's are labeled by Greek letters. The $\fontH{L}\rightarrow\fontH{W}$ LOCC's after a complete measurement $x_i$ and before the lab oracle are labeled by $\alpha_{i,j}$'s, and those after the lab oracle and before the next complete measurement are labeled by $\beta_{i,j}$'s. Likewise, the $\fontH{W}\rightarrow \fontH{L}$ LOCC's after the complete measurement $x_i$ and before the lab oracle are labeled by $\mu_{i,j}$'s, and those after the lab oracle and before the next complete measurement are labeled by $\nu_{i,j}$'s. Note that $\alpha_{i,j}$ and $\beta_{i,j}$ are weak measurement results (i.e., results of POVM measurements) of $\fontH{L}$ that are sent to $\fontH{W}$, and $\mu_{i,j}$ and $\nu_{i,j}$ are weak measurement results of $\fontH{W}$ that are sent to $\fontH{L}$.

\begin{proof}[Proof of Claim~\ref{cl:derandomization}]
We first  observe that the output density matrix of 
the $\fontH{W}$ register, when the QUALM is applied on some lab oracle, is fully determined (i.e., specified {\it deterministically}) by the values that the random variables $\alpha_{i,j},\beta_{i,j},\mu_{i,j},\nu_{i,j}, x_i$ take. Namely, the only dependence on the lab oracle is through these values. 
We denote the final reduced density matrix on $\fontH{S}_{\text{out}}$ as a function of these variables by 
$\rho^{\fontH{S}_{\text{out}}}(\alpha,\beta,\mu, \nu, x)$. 

Let the probability distribution over these variables be $\Pi(\alpha,\beta,\mu,\nu, x)$.
The output density matrix is the average
\begin{equation}
\rho_{\rm final}=\sum_{\alpha,\beta,\mu,\nu, x}\Pi(\alpha,\beta,\mu,\nu,x)\rho^{\fontH{S}_{\text{out}}}(\alpha,\beta,\mu,\nu,x)\,.\end{equation} 
Crucially, the state $\rho^{\fontH{S}_{\text{out}}}(\alpha,\beta,\mu,\nu,x)$ does not depend on the lab oracle since $\fontH{W}$ can only probe $\fontH{L}$ using the measurement results $\alpha,\beta,\mu,\nu,x$, but the probability distribution $\Pi(\alpha,\beta,\mu,\nu,x)$ does depend on the lab oracle. 

For the two lab oracles we are interested in, we denote the corresponding probability distributions by $\Pi^P$ and $\Pi^Q$. We have
\begin{align}
    \left\lVert \rho_{\rm final}^{P}-\rho_{\rm final}^Q\right\rVert_1 &\leq \sum_{\alpha,\beta,\mu,\nu,x}\left|\Pi^P(\alpha,\beta,\mu,\nu,x)-\Pi^Q(\alpha,\beta,\mu,\nu,x)\right| \left\lVert \rho^{\fontH{S}_{\text{out}}}(\alpha,\beta,\mu,\nu,x)\right\rVert_1\nonumber\\
    &=\sum_{\alpha,\beta,\mu,\nu,x}\left|\Pi^P(\alpha,\beta,\mu,\nu,x)-\Pi^Q(\alpha,\beta,\mu,\nu,x)\right|\,.
    \label{eq:norm bias}
\end{align}
Therefore it is sufficient to show that the probability distributions $\Pi^{P}(\alpha,\beta,\mu,\nu,x), \Pi^{Q}(\alpha,\beta,\mu,\nu,x)$ are close to each other in total variation distance. In the following, we will omit the lab oracle label $P,Q$ since the discussion is independent of which of the two lab oracles is being used. We will return to the comparison of the particular lab oracles at the end.

In its current form, $\Pi(\alpha,\beta,\mu,\nu,x)$ is not the output of a simple measurement QUALM. There are two key differences between the general case and simple measurement case: (i) In a simple measurement QUALM there is only a single (generalized) projective measurement between two applications of lab oracles, while in the general case $\fontH{L}$ is measured by many weak measurements in addition to projective measurements. (ii) In between the measurements with outputs $\alpha$ and $\beta$, there are quantum channels applied to $\fontH{L}$ controlled by $\mu,\nu$, but $\mu,\nu$ are random variables determined by weak measurements on $\fontH{W}$, and therefore depend probabilistically on all previous parameters. By comparison, in a simple measurement
QUALM, the previous parameters $s_0s_1...s_{i-1}$ have a {\it deterministic} effect on the measurement basis vector $\ket{y_{s_0s_1...s_i}^i}$, the weights $\lambda_{s_0s_1...s_i}^i$, and the prepared state $\sigma_{s_0s_1...s_i}^i$. In the following we will carry some further transformations to relate the probability distribution $\Pi(\alpha,\beta,\mu,\nu,x)$ to simple measurement QUALMs. The proof contains three steps.
\\ \\
\noindent{\bf 1. Decomposition of the chain of conditional probabilities.} Each one-round LOCC channel in Definition \ref{def:LOCC1} has the form $\mathcal{E}_m=\sum_\eta \mathcal{E}_{m\eta}$\,, with $\eta$ labeling $\alpha,\beta,\mu,\nu$ or $x$ for different one-round LOCC's. The $\mathcal{E}_{m\eta}$ are CP maps of the form $\mathcal{M}_{m\eta}^{\fontH{L}}\otimes \mathcal{N}_{m\eta}^{\fontH{W}}$ (for an $\fontH{L}$-to-$\fontH{W}$ LOCC) or $\mathcal{N}_{m\eta}^{\fontH{L}}\otimes \mathcal{M}_{m\eta}^{\fontH{W}}$ (for a $\fontH{W}$-to-$\fontH{L}$ LOCC). We have $m=1,2,...,k\cdot r\equiv M$ if there are $k$ rounds of application of lab oracles, and $r$ rounds of LOCC's in each step. By a slight abuse of notation, here we denote by $\mathcal{E}_{M\eta_M}$ the last channel applied in the QUALM {\it followed by} a channel which traces out all
qubits in $\fontH{W}$ which are not in $\fontH{S}_{\text{out}}$\,. This means that the final density matrix is a state of $\fontH{S}_{\text{out}}$.
We use the general notation $\mathcal{E}_{m\eta}$ for convenience in the following discussion about conditional probabilities. We will return to the more specific labels $\alpha,\beta,\mu,\nu,x$ after this discussion.

The effect of applying $\mathcal{E}_m,~m=1,2,...,M$ sequentially to an initial state state $\rho_{\rm in}$ is
\begin{align}
    \rho_{\rm final}&=\sum_{\eta_1\eta_2...\eta_{M}}\mathcal{E}_{M\eta_M}\circ\mathcal{E}_{M-1,\eta_{M-1}}...\circ\mathcal{E}_{1\eta_1}\!\left[\rho_{\rm in}\right]\nonumber\\
    &=\sum_{\eta_1\eta_2...\eta_M}\Pi_M\left(\eta_1,\eta_2,...,\eta_M\right)\rho_M\left(\eta_1,\eta_2,...,\eta_M\right)
\end{align}
with
\begin{align}
    \Pi_M\left(\eta_1,\eta_2,...,\eta_M\right)&={\rm tr}\left(\mathcal{E}_{M\eta_M}\circ\mathcal{E}_{M-1,\eta_{M-1}}...\circ\mathcal{E}_{1\eta_1}\!\left[\rho_{\rm in}\right]\right)\label{eq:PiMdef}\\
    \rho_M\left(\eta_1,\eta_2,...,\eta_M\right)&=\frac{\mathcal{E}_{M\eta_M}\circ\mathcal{E}_{M-1,\eta_{M-1}}...\circ\mathcal{E}_{1\eta_1}\!\left[\rho_{\rm in}\right]}{\Pi_M(\eta_1,\eta_2,...,\eta_m)}\,.
\end{align}
 Normalization of the channel $\sum_{\eta_M}\mathcal{E}_{M\eta_M}$ implies that 
\begin{align}
    \sum_{\eta_M}\Pi_M\left(\eta_1,\eta_2,...,\eta_M\right)=\Pi_{M-1}\left(\eta_1,\eta_2,...,\eta_{M-1}\right)
\end{align}
where the right-hand side is the probability obtained by only applying the first $M-1$ CP maps. Thus the probability $\Pi_M$ can be written as a chain of conditional probabilities:
\begin{align}
    \Pi_M\left(\eta_1,\eta_2,...,\eta_M\right)&=\prod_{m=2}^M\Pi_m\left(\eta_m|\eta_1,\eta_2,...,\eta_{m-1}\right)\Pi_1(\eta_1)\\
    \Pi_m\left(\eta_m|\eta_1,\eta_2,...,\eta_{m-1}\right)&=\frac{\Pi_m\left(\eta_1,\eta_2,...,\eta_m\right)}{\Pi_{m-1}\left(\eta_1,\eta_2,...,\eta_{m-1}\right)}={\rm tr}\left\{\mathcal{E}_{m\eta_m}\!\left[\rho_{m-1}\left(\eta_1,\eta_2,...,\eta_{m-1}\right)\right]\right\}\,.
    \label{eq:generalcondprob}
\end{align}
The conditional probability is determined by a weak measurement on the state $\rho_{m-1}$ which depends on the previous measurement outputs $\eta_1,...,\eta_{m-1}$.

Now we apply this general decomposition to an incoherent access QUALM. 
We decompose the probability distribution $\Pi(\alpha,\beta,\mu,\nu,x)$ into a product of two terms
\begin{equation}
    \Pi(\alpha,\beta,\mu,\nu,x)=\Pi_1\left(\beta,x|\alpha,\mu,\nu\right)\Pi_2\left(\alpha,\mu,\nu|\beta,x\right)\label{eq:decomposition} 
    \end{equation} 
where we denote
    \begin{align} 
    \Pi_1(\beta,x|\alpha,\mu,\nu)&\equiv\prod_{i,j}\Pi\left(\beta_{i,j}|...\right)\Pi\left(x_i|...\right)\label{eq:Pi1}\\
    \Pi_2\left(\alpha,\mu,\nu|\beta,x\right)&\equiv\prod_{i,j}\Pi\left(\mu_{i,j}|...\right)\Pi\left(\nu_{i,j}|...\right)\Pi\left(\alpha_{i,j}|...\right)\,.
    \label{eq:Pi2}
\end{align} 
$\Pi_1$ consists of the product of terms in $\Pi$
in which $\beta$ or $x$ are random variables conditioned on others, 
whereas $\Pi_2$ is the product of those terms in which the $\alpha, \mu,\nu$ are random. 
Note that the notation used for the terms in the right-hand side means that the conditioning is done on all {\it earlier parameters}.
Here and below, in conditional probabilities we will denote all previous measurement outputs $\alpha,\beta,\mu,\nu,x$ by ellipses, i.e.~``$...$''. Since all labels have a unique time ordering, this abbreviation does not cause any ambiguity. 
We now claim: 
\begin{claim} For all values of $\beta, x$, 
the conditional probability distributions $\Pi_2\left(\alpha,\mu,\nu|\beta,x\right)$ do not depend on the lab oracle.  
\end{claim} 
\begin{proof} 
In this case, for fixed labels $\alpha,\beta,\mu,\nu,x$, the conditional state (corresponding to $\rho_m(\eta_1,\eta_2,...,\eta_{m}))$ in the general discussion is a tensor product state of $\fontH{L}$ and $\fontH{W}$. We do not see a way to avoid having tedious notation, so to avoid confusion let us first consider the example case in Fig. \ref{fig:incoherent}(a). The state after applying the LOCC labeled by $\nu_{1,1}$ is
\begin{align}
    \rho\left(\nu_{1,1},\beta_{1,1},\alpha_{0,1},\mu_{0,1},x_0\right)&= \rho^{\fontH{L}}\left(\nu_{1,1},...\right)\otimes  \rho^{\fontH{W}}\left(\nu_{1,1},...\right)\\
    \rho^{\fontH{L}}\left(\nu_{1,1},...\right)&=\Omega_{\fontH{L}}^{-1}\mathcal{N}_{\nu_{1,1}}^{\fontH{L}}\circ\mathcal{M}_{\beta_{1,1}}^{\fontH{L}}\circ\mathcal{U}\circ \mathcal{M}_{\alpha_{0,1}}^{\fontH{L}}\circ\mathcal{N}_{\mu_{0,1}}^{\fontH{L}}\!\left[\ket{x_0}\bra{x_0}\right]\label{eq:rhoLexample}\\
    \rho^{\fontH{W}}\left(\nu_{1,1},...\right)&=\Omega_{\fontH{H}}^{-1}\mathcal{M}_{\nu_{1,1}}^{\fontH{W}}\circ\mathcal{N}_{\beta_{1,1}}^{\fontH{W}}\circ \mathcal{N}_{\alpha_{0,1}}^{\fontH{W}}\circ\mathcal{M}_{\mu_{0,1}}^{\fontH{W}}\circ \mathcal{N}_{x_0}^{\fontH{W}}\!\left[\ket{0_{\fontH{W}}}\bra{0_{\fontH{W}}}\right]\,.
    \label{eq:rhoWexample}
\end{align}
As a reminder, the $\mathcal{M}$'s are CP maps and the $\mathcal{N}$'s are CPTP maps; see Definition~\ref{def:LOCC1}. $\mathcal{U}$ denotes the superoperator associated with applying the unitary $U$ from the lab oracle, and $\Omega_{\fontH{L}}^{-1}, \Omega_{\fontH{H}}^{-1}$ are normalization factors. Because $\mathcal{N}_{\nu_{1,1}}^{\fontH{L}}$ is trace-preserving, the conditional probability $\Pi\left(\nu_{1,1}|\beta_{1,1},...,x_0\right)$ defined in Eqn.~\eqref{eq:generalcondprob} only depends on $\rho^{\fontH{W}}$:
\begin{align}
    \Pi\left(\nu_{1,1}|\beta_{1,1},...,x_0\right)&={\rm tr}\left\{\mathcal{M}_{\nu_{1,1}}^{\fontH{W}}\!\left(\rho^{\fontH{W}}(\beta_{1,1},\alpha_{0,1},\mu_{0,1},x_0)\right)\right\}
    \end{align}
where $\rho^{\fontH{W}}\left(\beta_{1,1},\alpha_{0,1},\mu_{0,1},x_0\right)$ is the state before applying the CP map $\mathcal{M}_{\nu_{1,1}}^{\fontH{W}}$. Explicitly, we have $\rho^{\fontH{W}}\left(\beta_{1,1},\alpha_{0,1},\mu_{0,1},x_0\right)=\tilde{\Omega}^{-1}_{\fontH{W}}\mathcal{N}_{\beta_{1,1}}^{\fontH{W}}\circ \mathcal{N}_{\alpha_{0,1}}^{\fontH{W}}\circ\mathcal{M}_{\mu_{0,1}}^{\fontH{W}}\circ \mathcal{N}_{x_0}^{\fontH{W}}\!\left[\ket{0_{\fontH{W}}}\bra{0_{\fontH{W}}}\right]$. 

Importantly, $\rho^{\fontH{W}}\left(\beta_{1,1},...,x_0\right)$ only depends on the labels and does not depend on the lab oracle. Therefore the same is true for the conditional probability $\Pi\left(\nu_{1,1}|\beta_{1,1},...,x_0\right)$. This argument applies to all conditional probabilities $\Pi(\nu_{i,j}|...),~\Pi(\mu_{i,j}|...)$ corresponding to $\fontH{W}$-to-$\fontH{L}$ LOCC's.

Following a similar argument, one can see that the conditional probabilities corresponding to $\fontH{L}$-to-$\fontH{W}$ LOCC's are determined by the conditional state of $\fontH{L}$, which does depend on the lab oracle, denoted by $\mathcal{U}$ in Eqn.~\eqref{eq:rhoLexample}. For example, 
\begin{align}
    \Pi\left(\beta_{1,1}|\alpha_{0,1},\mu_{0,1},x_0\right)={\rm tr}\left\{\mathcal{M}^{\fontH{L}}_{\beta_{1,1}}\!\left[\rho^L\left(\alpha_{0,1},\mu_{0,1},x_0\right)\right]\right\}
\end{align}

Up to this point, the discussion applies as long as $\fontH{L}$ and $\fontH{W}$ are only coupled by LOCC's, and we have not used the condition that there is a complete projective measurement between any two subsequent applications of the lab oracle. The projective measurement leads to an important simplification because the state after the measurement $\ket{x_i}$ is completely determined, such that the application of the lab oracle before this step has no effect on the conditional states $\rho^{\fontH{L}}$ occuring afterwards. Consequently, the conditional probability $\Pi(\beta_{i,j}|...)$ only depends on the $i$th application of the lab oracle, and $\Pi(\alpha_{i,j}|...)$ does not depend on the lab oracle at all. 
\end{proof} 

\noindent{\bf 2. Derandomization of the lab oracle independent part.} Since $\Pi_2$ in Eqn.~\eqref{eq:Pi2} is independent of the lab oracle, we can consider it as probabilistic instructions given by $\fontH{W}$ to $\fontH{L}$. If the probability distribution $\Pi_2$ did not depend on $\beta,x$, Eqn.~\eqref{eq:decomposition} would mean that the general QUALM output is a probabilistic average for fixed $\alpha,\mu,\nu$, and therefore we can consider the QUALM as an average over QUALMs with fixed $\alpha,\mu,\nu$, treating these variables as constants, and we can focus on $\Pi_1$. However, since the conditional probability distribution in each term may depend on the previous measurement results $\beta,x$, which in turn depend on the lab oracle, we are not ready to drop $\Pi_2$. Instead, we introduce an additional step to relate these probabilistic instructions to deterministic ones.

For this purpose, we consider the following simple mathematical fact:

\begin{fact}\label{thm:deltadecomposition}
For an arbitrary probability distribution $P(\alpha)$ of a discrete index $\alpha=1,2,...,N$ (normalized as $\sum_{\alpha=1}^NP(\alpha)=1$), there exists a function $f(r): [0,1]\rightarrow \left\{1,2,...,N\right\}$ which takes a variable $r\in[0,1]$ and maps it to one of the integers $\alpha=1,2,...,N$, such that when $r$ is a random variable with uniform probability distribution, the probability that $f(r)$ takes value $\alpha$ is $P(\alpha)$:
\begin{align}
    {\rm Prob}\left(f(r)=\alpha\right)=\int_0^1 dr\,\delta_{\alpha,f(r)}=P(\alpha)\,,
    \label{eq:deltadecomposition}
\end{align}
where $\delta_{i,j}$ is the Kronecker delta.
\end{fact}
Intuitively, one just needs to divide the interval $[0,1]$ into $N$ subintervals, and make sure that the $\alpha$th interval has width $P_\alpha$ so that the probability with which $r$ falls in that interval is $\alpha$. An explicit formula for $f(r)$ is
\begin{align}
    f(r)=\sum_{\alpha=1}^N\alpha\, \theta\left(r-\sum_{\beta=1}^{\alpha-1}P(\beta)\right)\theta\left(\sum_{\beta=1}^{\alpha}P(\beta)-r\right)\label{eq:functiondef}
\end{align}
where $\theta(x)$ is the step function that is $1$ for $x>0$ and $0$ otherwise. This function has $N$ steps, taking value $\alpha$ when $\sum_{\beta=1}^{\alpha-1}P(\beta)<r<\sum_{\beta=1}^{\alpha}P(\beta)$. An example is illustrated in Fig.~\ref{fig:functiondef}.

\begin{figure}
    \centering
    \includegraphics[width=4in]{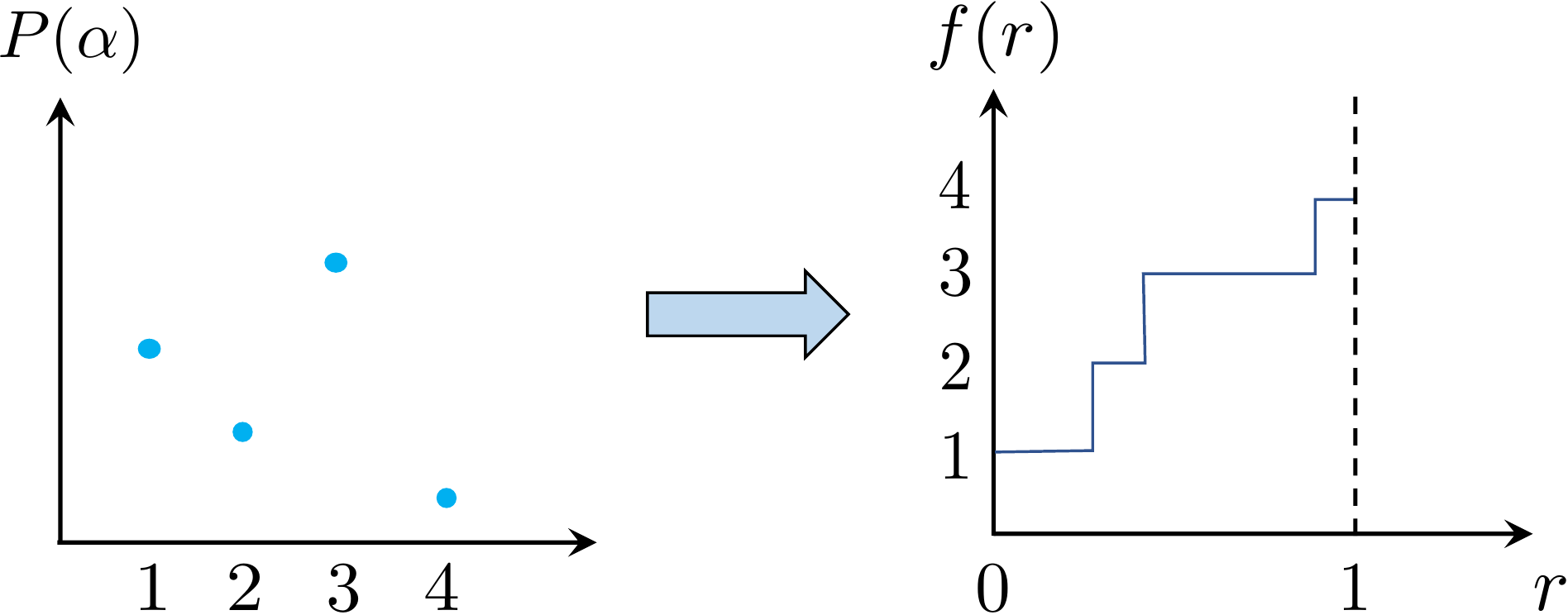}
    \caption{Illustration of the function $f(r)$ defined in Eqn.~\eqref{eq:functiondef}, such that the probability of $f(r)=\alpha$ is $P(\alpha)$. In the example, $\alpha=1,2,3,4$.}
    \label{fig:functiondef}
\end{figure}

Now for each conditional probability in $\Pi_2$ in Eqn.~\eqref{eq:Pi2} we introduce an independent random variable. For example, $\Pi(\mu_{i,j}|...)$ corresponds to a function $f_\mu\left(r_{ij}^\mu;...\right)$. For each fixed value $r_{ij}^\mu$, this function outputs a particular index $\mu_{i,j}$ which depends on the previous measurement results. The same decomposition is incorporated in the $\nu_{i,j}$ and $\alpha_{i,j}$ terms with independent random variables $r_{ij}^\nu$ and $r_{ij}^\alpha$. Applying the decomposition in Eqn.~\eqref{eq:deltadecomposition} to all these variables we obtain
\begin{align}
    \Pi_2(\alpha,\mu,\nu|\beta,x)&=\int_0^1 \prod_{ij}dr_{ij}^\mu \,dr_{ij}^\nu \,dr_{ij}^\alpha\, \delta_r\left(\alpha,\mu,\nu|\beta,x\right)\\
    \delta_r\left(\alpha,\mu,\nu|\beta,x\right)&=\prod_{ij}\delta_{\mu_{i,j},f^\mu_{ij}(r^\mu_{ij};...)}\delta_{\nu_{i,j},f^\nu_{ij}(r^\nu_{ij};...)}\delta_{\alpha_{i,j},f^\alpha_{ij}(r^\alpha_{ij};...)}\,.
\end{align}
This expression means that we can decompose the probability $\Pi_2$ into an average of $\delta$-functions $\delta_r\left(\alpha,\mu,\nu|\beta,x\right)$ over uniform, independent random variables $r_{ij}^{\alpha,\mu,\nu}$. In more physical terms, if we replace $\Pi_2$ by $\delta_r$ in the QUALM, it corresponds to replacing the probabilistic instructions by deterministic ones: at each step, $\fontH{W}$ performs a computation and obtains a particular value of $\mu,\nu$ or $\alpha$ based on the previous measurement results. Applying this decomposition in Eqn.~\eqref{eq:decomposition} we obtain
\begin{align}
    \Pi(\alpha,\beta,\mu,\nu,x)=\int dr\,\Pi_1\left(\beta,x|\alpha,\mu,\nu\right)\delta_r\left(\alpha,\mu,\nu|\beta,x\right)
\end{align}
where we have used $\int dr$ as an abbreviation for the integration over all variables $r_{ij}^{\mu,\nu,\alpha}$. 
Since $\Pi_2$ and therefore the functions $\delta_r$ are independent of the lab oracle, for two different lab oracles the $1$-norm difference defined in the right-hand side of Eqn.~\eqref{eq:norm bias} is bounded by
\begin{align}
    \sum_{\alpha,\beta,\mu,\nu,x}&\left|\Pi^P\left(\alpha,\beta,\mu,\nu,x\right)-\Pi^Q\left(\alpha,\beta,\mu,\nu,x\right)\right|\leq \int dr\,\Delta_r^{PQ}\label{eq:bias decomposition}
\end{align}
where
\begin{align}
    \Delta_r^{PQ}&=\sum_{\alpha,\beta,\mu,\nu,x}\left|\Pi_1^P\left(\beta,x|\alpha,\mu,\nu\right)-\Pi_1^Q\left(\beta,x|\alpha,\mu,\nu\right)\right|\delta_r\left(\alpha,\mu,\nu|\beta,x\right)\nonumber\\
    &\equiv\sum_{\beta,x}\left|\Pi_r^P(\beta,x)-\Pi_r^Q(\beta,x)\right|\label{eq:bias for fixed r}
\end{align}
and
\begin{align}
    \Pi_r(\beta,x)&=\sum_{\alpha,\mu,\nu}\Pi_1(\beta,x|\alpha,\mu,\nu)\delta_r(\alpha,\mu,\nu|\beta,x)\,.
    \label{eq:Pi_r}
\end{align}
The reason for the last equality in Eqn.~\eqref{eq:bias for fixed r} is that the sum in 
Eqn.~\eqref{eq:Pi_r} contains only a single non-zero term, due to the Kronecker delta function.

Observe that $\Pi_r(\beta,x)$ is the output probability distribution of a different QUALM, where all indices $\alpha,\mu,\nu$ are  deterministic functions specified by  $\delta_r(\alpha,\mu,\nu|\beta,x)$
as functions of previous 
measurement results. This QUALM is derived from the original QUALM by fixing the input random bits to be $r$. Our remaining task is to show the relation of this QUALM to a simple measurement QUALM. 
In particular, we will prove that $\Delta_r^{PQ}$ is upper bounded by the total variation distance between the two lab oracles, achieved by a simple measurement QUALM.
\\ \\
\noindent{\bf 3. Relating the lab oracle dependent part to a simple measurement QUALM.}
We can decompose the conditional probability $\Pi_1$ in Eqn.~\eqref{eq:Pi1} into the terms after the $i$th application of the lab oracle:
\begin{align}
    \Pi_1\left(\beta,x|\alpha,\mu,\nu\right)&=\prod_{i=1}^k\Pi_i\left(\beta_{i,j},x_{i}|...\right)
\end{align}
where
\begin{align}
    \Pi_i\left(\beta_{i,j},x_{i}|...\right)&=\prod_j\Pi\left(\beta_{i,j}|...\right)\Pi\left(x_i|...\right)\nonumber\\
    &=\bra{x_i}\mathcal{N}^{\fontH{L}}_{\nu_{i,j_{\rm max}}}\circ\mathcal{M}^{\fontH{L}}_{\beta_{i,j_{\rm max}}}\circ\cdots \circ \mathcal{M}_{\beta_{i,1}}^{\fontH{L}}\circ \mathcal{U}[\rho^{\fontH{L}}_{i-1}]\ket{x_i}\label{eq:measurement after U}
\end{align}
and $\rho^{\fontH{L}}_{i-1}$ the state of $\fontH{L}$ right before the $i$th application of the lab oracle, which depends on the previous measurement results.
Here $\ket{x_i}$ denotes the vectors of the complete $i$th measurement.\footnote{Though the notation is slightly confusing, note that these vectors form a general orthonormal basis which is not necessarily the computational basis.} Note that each $\nu_{i,j}$ already takes a deterministic value that is determined by previous measurement results. The difference with the simple measurement QUALM case is that here we have a sequence of weak measurements defined by the CP maps 
\begin{align}
    \mathcal{E}^{i}_{\beta_{i,1}\beta_{i,2}...\beta_{i,j_{\rm max}}}=\mathcal{N}^{\fontH{L}}_{\nu_{i,j_{\rm max}}}\circ\mathcal{M}^{\fontH{L}}_{\beta_{i,j_{\rm max}}}\circ \cdots \circ \mathcal{M}_{\beta_{i,1}}^{\fontH{L}}\label{eq:CPafterU}
\end{align}
before the projective measurement on states $\ket{x_i}$. To show that this is equivalent to the generalized projective measurement in a simple measurement QUALM, we first use the following fact:
\begin{fact}
Consider a family of CP maps $\mathcal{E}_m,~m=1,2,...,M$ from states in a Hilbert space $\fontH{L}$ to operators in the same Hilbert space, such that $\sum_{m=1}^M\mathcal{E}_m$ is a CPTP map. In addition, introduce an auxiliary system $\fontH{A}$ with Hilbert space dimension $M$, and an orthonormal basis $\ket{m_{\fontH{A}}},~m=1,2,...,M$. Then the map $\mathcal{E}$ defined in the following is a CPTP map from $\fontH{L}$ to $\fontH{L}\otimes \fontH{A}$:
\begin{align}
    \mathcal{E}[\rho]=\sum_{m=1}^M\mathcal{E}_m[\rho]\otimes\ket{m_{\fontH{A}}}\bra{m_{\fontH{A}}}\,.
\end{align}
\end{fact}
\noindent This result can be verified by checking directly that $\mathcal{E}$ is completely positive and trace-preserving.  Now one can see that applying $\mathcal{E}_m$ followed by a projection on $\ket{x_i}$ is equivalent to applying the CPTP map $\mathcal{E}$ followed by a projection on $\ket{x_i}\otimes\ket{m_{\fontH{A}}}$:
\begin{align}
    \bra{x_i}\mathcal{E}_m[\rho]\ket{x_i}=\left(\bra{x_i}\otimes\bra{m_{\fontH{A}}}\right)\mathcal{E}[\rho] \left(\ket{x_i}\otimes\ket{m_{\fontH{A}}}\right)\,.
\end{align}
To relate to the projective measurements, we can further decompose $\mathcal{E}$ into Kraus operators:
\begin{align}
    \mathcal{E}[\rho]=\sum_b A_b\rho A_b^\dagger
\end{align}
with $A_b: \fontH{L}\rightarrow \fontH{L}\otimes \fontH{A}$ being linear operators. Using this decomposition we obtain
\begin{align}
    \bra{x_i}\mathcal{E}_m\left(\rho\right)\ket{x_i}&=\sum_b\bra{y_{x_imb}}\rho\ket{y_{x_imb}}\lambda_{x_imb}\label{eq:sum over b}\\
    \lambda_{x_imb}&=(\bra{x_i}\otimes\bra{m_{\fontH{A}}})A_bA_b^\dagger(\ket{x_i}\otimes\ket{m_{\fontH{A}}})\\
    \ket{y_{x_imb}}&=\lambda_{x_imb}^{-1/2}A_b^\dagger \ket{x_i}\otimes\ket{m_{\fontH{A}}}\,.
\end{align}
One can check that
\begin{align}
    \sum_{x_imb}\lambda_{x_imb}\,\ket{y_{x_imb}}\bra{y_{x_imb}}=\sum_bA_b^\dagger A_b=\mathds{1}\,.
\end{align}
Therefore, the basis $\{\ket{y_{x_imb}}\}$ and coeffcients $\lambda_{x_imb}$ satisfy the requirements of a simple measurement QUALM. 

\begin{figure}
    \centering
    \includegraphics[width=3.5in]{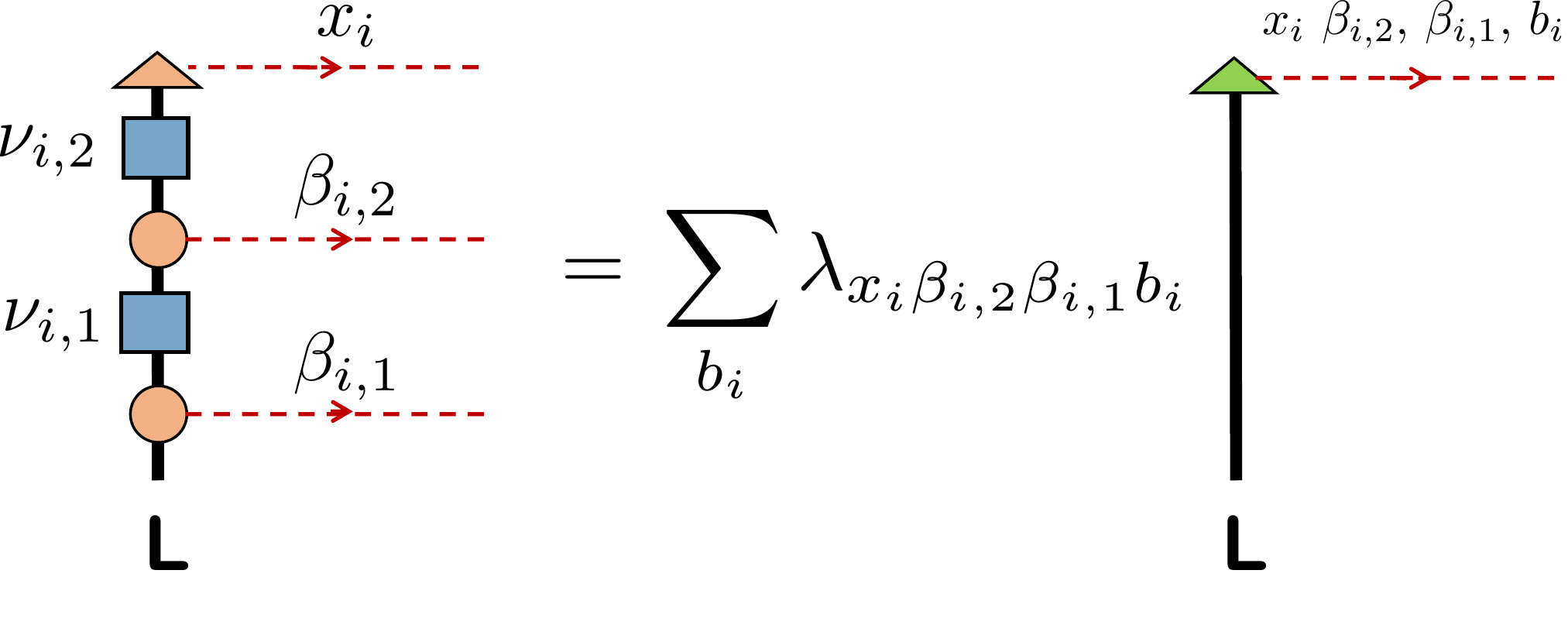}
    \caption{Illustration of how the CP maps after the $i$th application of the lab oracle (see Fig.~\ref{fig:incoherent}) followed by a projective measurement are equivalent to a single POVM with operators $\lambda_{x_i\beta_{ij}b_i}\ket{y_{x_i\beta_{ij}b_i}}\bra{y_{x_i\beta_{ij}b_i}}$, with a sum over the extra index $b_i$\,. This is an application of the general result in Eqn.~\eqref{eq:sum over b}.}
    \label{fig:consolidation}
\end{figure}

Now we can apply the procedure above for $\mathcal{E}^i_{\beta_{i,1}...\beta_{i,j_{\rm max}}}$ defined in Eqn.~\eqref{eq:CPafterU} for each round $i$ ($m$ now corresponds to all $\beta_{i,j}$ with fixed $i$). This defines states  $\ket{y_{x_i,\beta_{i,j},b_i}}$ and weights $\lambda_{x_i,\beta_{i,j},b_i}$, as is illustrated in Fig.~\ref{fig:consolidation}. We can then define a simple measurement QUALM which carries POVM operators $\lambda_{x_i,\beta_{i,j},b_i}\ket{y_{x_i,\beta_{i,j},b_i}}\bra{y_{x_i,\beta_{i,j},b_i}}$, and prepares states $\rho_i^{\fontH{L}}$ (defined in Eqn.~\eqref{eq:measurement after U}). The output of this simple measurement QUALM is not identical to $\Pi_r(\beta,x)$ in Eqn.~\eqref{eq:Pi_r} since it also depends on the new variables $b_i$\,. However, according to Eqn.~\eqref{eq:sum over b} they are related by simply performing the sum over $b_i$\,:
\begin{align}
    \Pi_r(\beta,x)&=\sum_{b_1b_2...b_k}\widetilde{\Pi}_{r}\left(\beta,x,b\right)\,,
\end{align}
where $\widetilde{\Pi}_r$ on the right-hand side is the probability distribution of the simple measurement QUALM defined above.

In other words, the probability distribution $\Pi_r$ is obtained by applying the simple measurement QUALM and forgetting about some of the variables $b$, which will only reduce the capability of distinguishing two lab oracles. More explicitly, $\Delta_r^{PQ}$ in Eqn.~\eqref{eq:bias for fixed r} is bounded by
\begin{align}
    \Delta_r^{PQ}\leq \sum_{\beta,x,b}\left|\widetilde{\Pi}_r^{P}\left(\beta,x,b\right)-\widetilde{\Pi}_{r}^{Q}\left(\beta,x,b\right)\right|\,.
\end{align}
The right-hand side of this equation is the bias for the simple measurement QUALM, so that we can apply the upper bound given by Lemma \ref{thm:tvd}. Using this bound and Eqn.'s~\eqref{eq:norm bias} and~\eqref{eq:bias decomposition}, we conclude that a general incoherent access QUALM with $k$ rounds of applications of the lab oracle ($k<\left(2^\ell/\sqrt{6}\right)^{4/7}$) and $\ell$ qubits in $\fontH{L}$ has a bias upper bound
\begin{align}
    \left\| \rho_{\rm final}^P-\rho_{\rm final}^Q\right\|_1\leq \mathcal{O}\left(\frac{k^3}{2^\ell}\right)\,.
\end{align}
This proves Theorem \ref{thm:main}.
\end{proof}

\section{Corollaries}
\label{sec:corrs1}
Here we discuss three interesting corollaries of Theorem \ref{thm:main}, corresponding to different kinds of tomography-themed quantum tasks and QUALMs.  These are summarized informally below:
\begin{itemize}
    \item In the first corollary, we extend Theorem \ref{thm:main} to the task of distinguishing between lab oracles encoding more elaborate ensembles of unitaries.
    \item We provide a similar result for the task of distinguishing lab oracles that produce certain ensembles of states.
    \item Finally, we provide a related bound for distinguishing a random unitary channel from a completely depolarizing channel.
\end{itemize}
We now proceed with a more detailed discussion.

\subsection{Correlated random unitaries}

\begin{definition}[\textbf{A unitary ensemble distinction problem}] Consider a lab oracle $\textsf{LO}_\ell^\mu$ such that its quantum channel applies a unitary $U_k$ to $\fontH{L}$ upon being queried $k$ times.  That is, the lab oracle superoperator applies $U_1$ the first time it is queried, $U_2$ the second time it is queried, and so on.  Suppose that the lab oracle superoperator determines each of these unitaries by sampling from a probability distribution $\mu\left(U_1,U_2,...,U_k,...\right)$, and that the probability distribution is left-invariant:
\begin{align}
    \mu\left(VU_1,VU_2,...,VU_k,...\right)=\mu\left(U_1,U_2,...,U_k,...\right)
\end{align}
for an arbitrary unitary transformation $V$. The task we consider is to distinguish $\textsf{LO}_\ell^{\,\mu}$ and $\textsf{LO}_\ell^{\,\mu'}$ for different distributions $\mu$ and $\mu'$.
\end{definition}
\noindent Note that $\textsf{LOP}_\ell$ and $\textsf{LOQ}_\ell$ discussed above are examples of lab oracles of the form $\textsf{LO}_\ell^\mu$.  In particular, $\textsf{LOP}_\ell$ corresponds to the $U_i$ chosen i.i.d., and $\textsf{LOQ}_\ell$ corresponds to all $U_i$'s perfectly correlated with each other. Note also that to implement the lab oracle described above in our framework, we need a register within 
$\fontH{N}$ which is incremented by $1$ every time the lab oracle is called. We will not formally write down the lab oracle here since this is straightfoward, and its exact form is not important for our discussion.

We have the following corollary of Theorem \ref{thm:main}:

\begin{corollary}\label{cor:Qktilde}
For an arbitrary left-invariant probability distribution $\mu\left(U_1,U_2,...,U_k,...\right)$ and an arbitrary incoherent access QUALM which queries the lab oracle at most $k<\left(2^{\ell}/\sqrt{6}\right)^{4/7}$ times, the output density matrix $\rho_{\fontH{S}_{\text{out}}}^\mu$ of the QUALM when the lab oracle is $\textsf{LO}_\ell^{\,\mu}$ and the output density matrix $\rho_{\fontH{S}_{\text{out}}}^P$ of the QUALM when the lab oracle is $\textsf{LOP}_\ell$ satisfy the following bound:
\begin{align}
\left\lVert \rho_{\fontH{S}_{\text{out}}}^\mu-\rho_{\fontH{S}_{\text{out}}}^P\right\rVert_1\leq \mathcal{O}\left(\frac{k^3}{2^\ell}\right)\,.
\end{align}
\end{corollary}

\begin{proof}
To prove this corollary, we define $U_i=U_1R_i$ with $R_1=\mathds{1}$. Due to left-invariance, $\mu\left(U_1,U_2,...,U_k,...\right)=\mu\left(\mathds{1},R_2,R_3,...,R_k,...\right)$. For a given incoherent access QUALM, the output density operator can be expressed as an average
\begin{align}
    \rho_{\fontH{S}_{\text{out}}}^\mu=\int\prod_{i\geq 2}^k dR_i \,\mu\left(\mathds{1},R_2,R_3,...,R_k\right)\int dU_1\, \rho_{\fontH{S}_{\text{out}}}\!\left(U_1,R_2,R_3,...,R_k\right)\label{eq:rhoWmu}
\end{align}
where $\rho_{\fontH{S}_{\text{out}}}\!\left(U_1,R_2,R_3,...,R_k\right)$ is the output state of $\fontH{S}_{\text{out}}$ for fixed $U_1,R_2,R_3,...,R_k$. Now if we carry out the integration over $U_1$, we see that
\begin{align}
    \tilde{\rho}_{\fontH{S}_{\text{out}}}(R_2,R_3,...,R_k)\equiv \int dU_1\,\rho_{\fontH{S}_{\text{out}}}(U_1,R_2,...,R_k)
\end{align}
can be interpreted as the output state of a modified QUALM with fixed unitary lab oracle $\textsf{LOQ}_\ell$.  In this modified QUALM, there is an additional gate $R_i$ before each lab oracle. The new QUALM is still incoherent, so that Theorem \ref{thm:main} implies
\begin{align}
    \left\lVert \tilde{\rho}_{\fontH{S}_{\text{out}}}(R_2,R_3,...,R_k)-\rho_{\fontH{S}_{\text{out}}}^P\right\rVert_1\leq  \mathcal{O}\left(\frac{k^3}{2^\ell}\right)\,.
\end{align}
Integration over $R_i$ in Eqn.~\eqref{eq:rhoWmu} gives us
\begin{align}
    \left\lVert \rho_{\fontH{S}_{\text{out}}}^\mu-\rho_{\fontH{S}_{\text{out}}}^P\right\rVert_1\leq\int \prod_{i=2}^kdR_i\mu\left(\mathds{1},R_2,R_3,...,R_k\right)\left\lVert \tilde{\rho}_{\fontH{S}_{\text{out}}}(R_2,R_3,...,R_k)-\rho_{\fontH{S}_{\text{out}}}^P\right\rVert_1\leq \mathcal{O}\left(\frac{k^3}{2^\ell}\right)\,.
\end{align}
\end{proof}

This result guarantees (following the same argument as in the proof of Theorem~\ref{thm:main}; see the end of Subsection~\ref{sec:nonadaptive}) that the lab oracles $\textsf{LO}_\ell^\mu$, $\textsf{LO}_\ell^{\mu'}$ corresponding to two different left-invariant probability distributions $\mu$ and $\mu'$ are difficult to distinguish since they are both ``close'' to $\textsf{LOP}_\ell$.

\subsection{State ensemble distinction}

\begin{definition}[\textbf{A state ensemble distinction problem}]
Suppose we have a lab oracle $\textsf{LOS}_\ell^{\,\nu}$ where its quantum channel prepares the state $|\psi_k\rangle$ on $\fontH{L}$ upon being queried $k$ times.  Specifically, the lab oracle superoperator prepares $|\psi_1\rangle$ the first time it is queried, prepares $|\psi_2\rangle$ the second time it is queried, etc.  We suppose that the lab oracle superoperator determines these states by sampling from a probability distribution $\nu\left(|\psi_1\rangle,|\psi_2\rangle,...,|\psi_k\rangle,...\right)$, and require that the probability distribution is unitarily invariant:
\begin{align}
    \nu\left(U|\psi_1\rangle,U|\psi_2\rangle,...,U|\psi_k\rangle,...\right)=\nu\left(|\psi_1\rangle, |\psi_2\rangle,...,|\psi_k\rangle,...\right)
    \label{eq:unitary invariance}
\end{align}
for any unitary $U$. We consider the task of distinguishing $\textsf{LOS}_\ell^{\,\nu}$ and $\textsf{LOS}_\ell^{\,\nu'}$ where $\nu$ and $\nu'$ are distinct distributions.
\end{definition}

An example is distinguishing between two lab oracles where one always generates the same Haar random state, and the other generates a new Haar random state each time. This is analogous to distinguishing $\textsf{LOQ}_\ell$ and $\textsf{LOP}_\ell$ in the setting of unitary ensembles. Another example to consider is that one lab oracle always generates the same Haar random state while the other one randomly generates one of the two independently chosen Haar random states $\ket{\psi_1},\ket{\psi_2}$.

In similar spirit to the previous subsection, we have the following corollary to Theorem \ref{thm:main}:

\begin{corollary}\label{cor:state distinction}
An incoherent access QUALM calling the oracle at most $k$ times, with $k<\left(2^{\ell}/\sqrt{6}\right)^{4/7}$, can only distinguish two state ensemble probability distributions $\nu_1$ and $\nu_2$ with bias $\leq \mathcal{O}(k^3/2^\ell)$ if both distributions are unitarily invariant, satisfying Eqn.~\eqref{eq:unitary invariance}.
\end{corollary} 

\begin{proof} 
This corollary can be proven by contradiction. We choose a reference state $\ket{0}$. For each of the given state ensembles $\nu_1\left(\ket{\psi_1},\ket{\psi_2},...,\ket{\psi_k},...\right)$ and $\nu_2\left(\ket{\psi_1},\ket{\psi_2},...,\ket{\psi_k},...\right)$, we can define a probability distribution of unitaries $U_i$
\begin{align}
    \widetilde{\nu}_{1}\left(U_1,U_2,...,U_k,...\right)&=\Omega^k\nu_{1}\left(U_1\ket{0},U_2\ket{0},...,U_k\ket{0},...\right) \\
    \widetilde{\nu}_{2}\left(U_1,U_2,...,U_k,...\right)&=\Omega^k\nu_{2}\left(U_1\ket{0},U_2\ket{0},...,U_k\ket{0},...\right)\,.
\end{align}
Here $\Omega$ is a constant normalization factor that takes care of the difference in integration measures of state integration and unitary operator integration. The integration over states is over the manifold $\mathbb{C}P\left(2^\ell-1\right)=SU(2^\ell)/SU(2^\ell-1)$, so 
\begin{align}
    \int d\psi \, f\left(\ket{\psi}\right)=\frac1{\Omega}\int_{\rm Haar} f\left(U\ket{0}\right)dU\,.
\end{align}
If an incoherent access QUALM can distinguish between these two lab oracles with a larger bias than $\mathcal{O}\left(k^3/2^n\right)$, then the same bias can be achieved for a different task of distinguishing the two ensembles of unitaries $\widetilde{\nu}_1$ and $\widetilde{\nu}_2$, since we could simply prepare input states $\ket{0}$ before each application of the lab oracle and then apply the QUALM that distinguishes the ensemble of states obtained this way. This is a contradiction with Corollary \ref{cor:Qktilde}.
\end{proof}

\subsection{Detecting non-unitarity}

Another simple corollary that follows directly from Theorem \ref{thm:main} is the following:

\begin{corollary}
Consider a lab oracle $\textsf{LOD}_\ell$ which applies a completely depolarizing channel every time it is called, i.e.~ $\mathcal{E}_{\fontH{L}}[\rho_{\fontH{L}}]=\frac{\mathds{1}}{2^\ell}$. For any incoherent access QUALM applied $k$ times with with $k<\left(2^{\ell}/\sqrt{6}\right)^{4/7}$, $\textsf{LOD}_\ell$ and the fixed unitary lab oracle $\textsf{LOQ}_\ell$ cannot be distinguished with a bias bigger than $\mathcal{O}\left(k^3/2^\ell\right)$.
\end{corollary}

\begin{proof}
The proof of this corollary is almost trivial, since for any $\rho$, 
\begin{align}
    \int_{\rm Haar} dU \, U\rho U^\dagger =\frac{\mathds{1}}{2^\ell}\,.
\end{align}
In words, averaging a state over random unitary evolution is equivalent to applying the completely depolarizing channel. The output of any incoherent access QUALM for $\textsf{LOD}_\ell$ is therefore identical to that of $\textsf{LOP}_\ell$ which uses independent random unitaries.  Thus, we have reduced our problem to one previously solved.
\end{proof}

\section{QUALM for symmetry of time evolution operator}
\label{sec:symmtask}

\subsection{Outline of problem}

In this section, we will be inspired by a new question: given an experimental system evolved by a time-independent Hamiltonian $H$, how do we determine the symmetries of $H$?  There are different kinds of symmetries we might be interested in, such as global and local symmetries.  Clearly this is an experimentally relevant problem, as symmetries are key to understanding the physics of natural systems.  We will address a toy version of the problem which itself may be useful in contemporary experiments.

We consider the three most commonly studied symmetry classes, known as the unitary, orthogonal and symplectic classes. These comprise a classification of different forms of time-reversal symmetry (or a lack thereof). The time evolution operator of a system without time-reversal symmetry belongs to the unitary group $U(D)$ when the Hilbert space dimension is $D$. The time-reversal transformation is an anti-unitary operator which satisfies $T\left(a\ket{\psi_1}+b\ket{\psi_2}\right)=a^*T\ket{\psi_1}+b^*T\ket{\psi_2}$, and preserves the norm of the state. Time-reversal symmetry requires the time evolution operator $U$ to satisfy\footnote{There is a difference between time-reversal symmetry of $U$ and that of a Hamiltonian. We will not expand on this point and will focus on the requirement in Eqn.~\eqref{eq:time-reversal def}.} $UT = TU$, or equivalently
\begin{align}
TUT=T^2U\,.\label{eq:time-reversal def}
\end{align}
Physically, every state should return to itself if we apply time-reversal twice, up to a phase, and so $T^2$ must be proportional to identity. It can be proven that the phase has to be $\pm 1$ (simply because $T^2$ commutes with $T$), so that there are two classes of time-reversal symmetric systems. 

If $T^2=\mathds{1}$, the canonical choice is $T=K$ where $K$ is the complex conjugation operator defined with respect to the standard basis
\begin{equation}
K \left( \sum_{x } c_x |x\rangle \right) = \sum_{x} c_x^* |x\rangle\,.
\end{equation}
In this case, $TUT=U^*$ and the time-reversal symmetry condition becomes $U^*=U$, so that $U$ is a real matrix satisfying $U^TU=\mathds{1}$, which belongs to the orthogonal group $O(D)$.

If $T^2=-\mathds{1}$, then $T$ can be written as $T=\textbf{J}K$ with $\textbf{J}$ a real anti-symmetric matrix satisfying $\textbf{J}^2=-\mathds{1}$. A standard choice is 
\begin{equation}\label{eq:J}
\textbf{J} := \begin{bmatrix} \textbf{0}_{D/2} & \mathds{1}_{D/2} \\ - \mathds{1}_{D/2} & \textbf{0}_{D/2} \end{bmatrix}
\end{equation} where $\textbf{0}_{D/2}$ is the $D/2 \times D/2$ matrix of all zeros (note that $D$ has to be even for this symmetry class). In this case, the time-reversal requirement becomes $-\textbf{J}U^*\textbf{J}=U$, or equivalently \begin{equation}\label{eq:timereveralJ}
-\textbf{J}U^T\textbf{J}=U^{-1}.
\end{equation} 
The unitaries satisfying this requirement form the compact symplectic group ${\rm Sp}(D/2)$. 

Some well-known examples of orthogonal and symmetric matrices are realized for the unitary dynamics of a single $SU(2)$ spin. For a spin in a static magnetic field, the Hamiltonian is $H=\vec{h}\cdot\vec{S}$, with $\vec{S}$ the spin angular momentum operator. For integer spin, the time evolution operator $e^{-iHt}$ is an orthogonal matrix, and for half odd integer spin (i.e., the value of the spin is in $\mathbb{Z}_{\geq 0} + \frac{1}{2}$) it is a symplectic matrix.

On each of the three Lie groups $U(D), O(D)$ and ${\rm Sp}(D/2)$, there is a unique Haar measure that is invariant under both left-translation and right-translation.  Here, then, is the problem we will address: suppose we are given oracle access to a unitary which is either (i) a fixed Haar random unitary, (ii) a fixed Haar random orthogonal matrix, and (iii) a fixed Haar random symplectic matrix.  Can we efficiently distinguish whether the oracle is of type (i), (ii), or (iii)?  This is tantamount to determing the type of time-reversal symmetry (or lack thereof) of the unitary encoded in the oracle.

The answer is: for coherent access, we can efficiently distinguish between (i), (ii), (iii) using $\mathcal{O}(1)$ oracle queries and $\mathcal{O}(\ell)$ gates.  For incoherent access, (i), (ii), and (iii) can only be reliably distinguished using $e^{\Omega(\ell)}$  oracle queries and as many gates.  As such we have an exponential gap in the query complexity, and thus in the QUALM complexity, between the coherent and incoherent settings, when it comes to determining the type of time-reversal symmetry.  The proofs of these statements generalize the proofs developed in previous sections.  The main additional complication is that Weingarten functions for the orthogonal and symplectic Haar ensembles are more sophisticated than the ordinary Weingarten functions, although the basic architecture of the proofs is the same as those in Section~\ref{sec:mainresult}.  We will draw extensively from the Weingarten function results in~\cite{Matsumoto1, CollinsMatsumoto1, Gu1}, and a review can be found in~\ref{App:reviewHaar}.

\subsection{Problem definition: distinguishing between symmetries of unitaries}

The problem stated above involves distinguishing three lab oracles.  Here is a precise statement of this task, which generalizes Definition~\ref{def:distinguish2}:
\begin{definition}[\bf The task of distinguishing between three lab oracles]
\label{def:distinguish3}
Given three lab oracles  $\LO_0=(\mathcal{E}_0,\rho_0)$, 
$\LO_1=(\mathcal{E}_1,\rho_1)$, $\LO_2=(\mathcal{E}_2,\rho_2)$,
consider a task of the form: 
$(\fontH{S}_{\text{in}},\fontH{S}_{\text{out}},f,\mathcal{G})$ for some $\mathcal{G}$ 
where $\fontH{S}_{\text{in}}$ is empty, $\fontH{S}_{\text{out}}$ consists of two qubits, and
\[f:\mathscr{L}\!\mathscr{O}(\fontH{N},\fontH{L})\to \{0,1,2\}\,,\]
where $f(\LO_i) = i$.
\end{definition}

As in Definition \ref{def:distinguish2}, we may have an error $\epsilon$ in the computation of this function; and just like in Definition \ref{def:bias}, we can talk about a QUALM which outputs for each of the lab oracles a probability distribution over the possible outcomes  
$\{0,1,2\}$.  If the minimal total variation distance between any pair of these three distributions
is $\delta$, we say that the QUALM achieves the task of distinguishing between these three lab oracles 
with bias $\delta$. Yet again, this bias can be amplified by standard repetition and a simple follow 
up classical calculation. 

Next we define the relevant lab oracles for the ``symmetry distinction'' problem:

\begin{definition}[\textbf{The fixed random unitary lab oracles} $\LO_\ell^U, \LO_\ell^O, \LO_\ell^{\text{Sp}}$]
For $\ell$ a positive integer, we define the lab oracles
$\LO_\ell^U=(\fontH{N},\fontH{L},\mathcal{E},\rho^U)$, $\LO_\ell^O=(\fontH{N},\fontH{L},\mathcal{E},\rho^O)$, $\LO_\ell^{\text{Sp}}=(\fontH{N},\fontH{L},\mathcal{E},\rho^{\text{Sp}})$ as follows:
\begin{itemize} 
\item $\fontH{N}$ is an $e^{\Omega(\ell)}$ qubit Hilbert space,
\item $\fontH{L}$ is an $\ell$ qubit Hilbert space with dimension $D = 2^\ell$,
\item $\mathcal{E}$ is a unitary transformation which for $U$ a unitary on $\ell$ qubits, and $\ket{U}$ a classical description of $U$ on $e^{\mathcal{O}(\ell)}$ qubits (specified to within exponentially many bits of accuracy in each entry), does the following: 
\[\mathcal{E}(\ket{U}\otimes\ket{\alpha})=\ket{U}\otimes U\ket{\alpha}\]
\item The states $\rho^U, \rho^O, \rho^\text{Sp}$ are given by:
\subitem $\bullet$ \, $\rho^U=\left(\int_{U(D)}  \ket{U}\bra{U}dU\right)\otimes \ket{0^\ell}\bra{0^\ell}$
\subitem $\bullet$ \, $\rho^O=\left(\int_{O(D)}  \ket{O}\bra{O}dO\right)\otimes \ket{0^\ell}\bra{0^\ell}$
\subitem $\bullet$ \, $\rho^\text{Sp}=\left(\int_{\text{Sp}(D/2)}  \ket{S}\bra{S}dS\right)\otimes \ket{0^\ell}\bra{0^\ell}$ 
\end{itemize} 
\end{definition} 
\noindent Then the problem statement is as follows:
\\ \\
\noindent \textbf{Problem statement (Symmetry distinction):} \textit{Given a set of admissible gates, find a QUALM which implements the task of distinguishing between the three lab oracles} $\LO_\ell^U, \LO_\ell^O, \LO_\ell^{\text{Sp}}$ with error at most $1/3$. 
\\ \\
Having stated the problem precisely, we now give our main results (the latter makes precise Theorem~\ref{thm:main2rough} stated in the introduction):
\begin{thm}\label{thm:symmetryswap}
For all $\ell \geq 1$, there exists a $\QUALM_\ell$ in the coherent access model which solve the symmetry distinction problem between $\LO_\ell^U, \LO_\ell^O, \LO_\ell^{\text{Sp}}$ and has $\mathcal{O}(\ell)$ QUALM complexity. 
\end{thm}

\begin{thm}\label{thm:main2}
Any incoherent access QUALM which makes $k<\left(2^\ell/12\right)^{2/7}$ queries to the lab oracle can at most distinguish between $\LO_\ell^U, \LO_\ell^O, \LO_\ell^{\text{Sp}}$ with a bias $\delta$ of order $ O\left(k^3/2^{\ell}\right)$. Consequently, the QUALM query complexity for distinguishing these three lab oracles with constant error (Definition \ref{def:distinguish3}) is lower bounded by the exponential $\Omega(2^{2\ell/7})$.
\end{thm} 
\noindent These two theorems are clearly analogs of Theorems~\ref{thm:mainswap} and~\ref{thm:main} above.
Once again, the proof of Theorem \ref{thm:symmetryswap} is easy (though slightly more complicated than that of Theorem \ref{thm:mainswap}) whereas the proof of the lower bound, 
Theorem~\ref{thm:main2}, is far from trivial. 
In the following two subsections, we prove Theorems~\ref{thm:symmetryswap} and~\ref{thm:main2} in turn.

\subsection{Proof of Theorem~\ref{thm:symmetryswap}}

Here we give an efficient QUALM which efficiently distinguishes $\LO_\ell^U, \LO_\ell^O, \LO_\ell^{\text{Sp}}$.  The method is essentially a variation of the SWAP test.  This amounts to a proof of Theorem~\ref{thm:symmetryswap}, and goes as follows.

\begin{figure}
    \centering
    \includegraphics[width=5.5in]{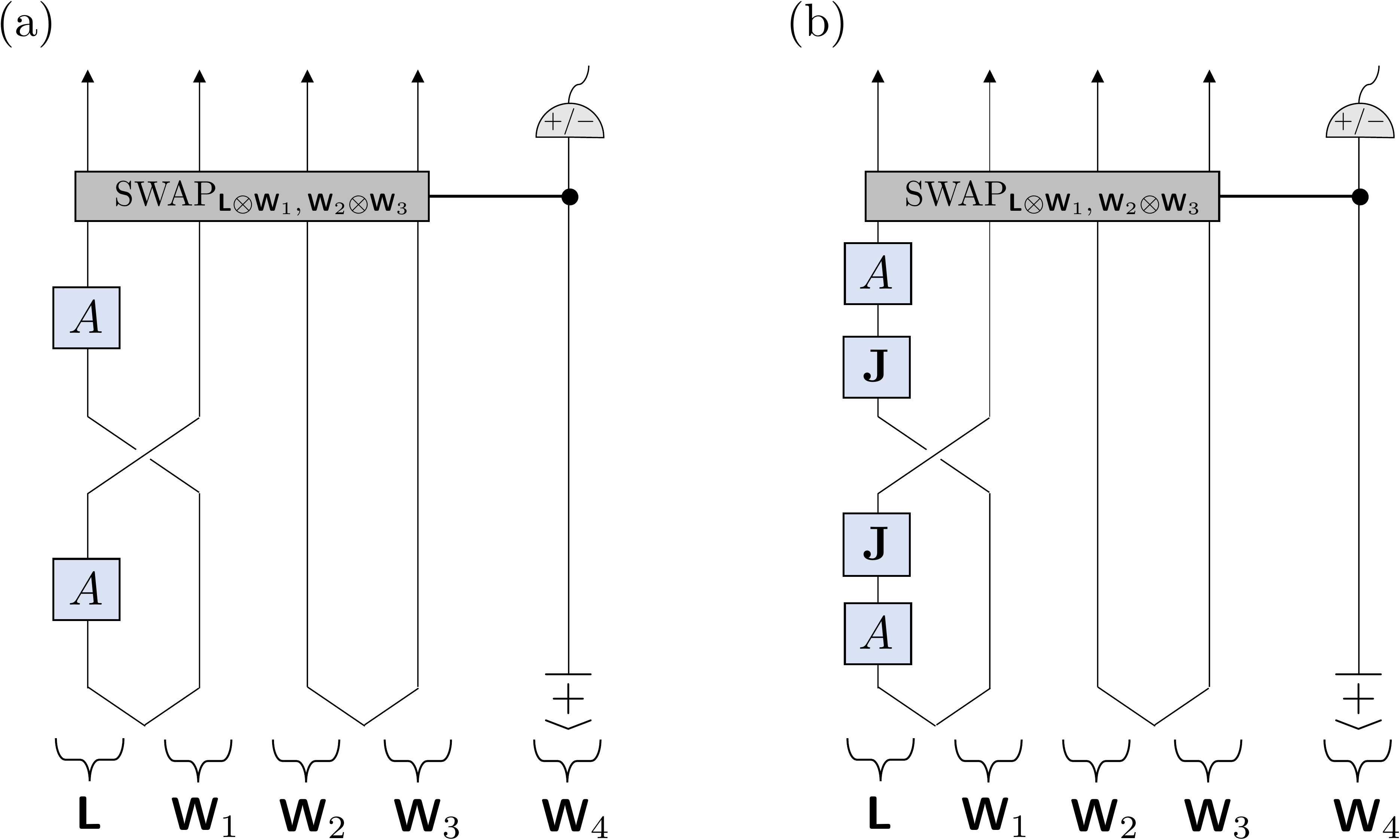}
    \caption{(a) Illustration of the coherent access QUALM that distinguishes between $\LO_\ell^O$ from the other two lab oracles $\LO_\ell^U$ and $\LO_\ell^{\text{Sp}}$. (b) The coherent access QUALM that distinguishes between  $\LO_\ell^U$ and $\LO_\ell^{\text{Sp}}$. }
    \label{fig:coherent symmetry distinction}
\end{figure}
\begin{proof}
Recall that $\fontH{L}$ has size $\ell$, and we let $\fontH{W}$ have size $3\ell+1$.  We write $\fontH{W} \simeq \fontH{W}_1 \otimes \fontH{W}_2 \otimes \fontH{W}_3 \otimes \fontH{W}_4$ where $\fontH{W}_1, \fontH{W}_2, \fontH{W}_3$ each contain $\ell$ qubits, and $\fontH{W}_4$ contains a single qubit.  Here we provide a description of the QUALM.  Letting $\{|x\rangle\}_{x \in \{0,1\}^\ell}$ be the standard basis on $\ell$ qubits, and letting $|+\rangle = \frac{1}{\sqrt{2}}(|0\rangle + |1\rangle)$ be a single qubit state, we prepare the state
\begin{equation}
|\Phi\rangle_{\fontH{LW}} := \left(\frac{1}{\sqrt{D}} \sum_{x \in \{0,1\}^\ell} |x\rangle_{\fontH{L}} \otimes |x\rangle_{\fontH{W}_1}\right) \otimes \left(\frac{1}{\sqrt{D}} \sum_{y \in \{0,1\}^\ell} |y\rangle_{\fontH{W}_2} \otimes |y\rangle_{\fontH{W}_3}\right) \otimes  |+\rangle_{\fontH{W}_4}\,.
\end{equation}
This state can be prepared from the all zero state using $2\ell+1$ Hadamard gates and $2\ell$ \text{CNOT} gates.  Next we query the lab oracle once and have it apply a unitary to $\fontH{L}$.  For now, we label the applied unitary by $A$, to be agnostic as to whether it is a regular unitary, an orthogonal matrix, or a symplectic matrix.  Then we swap the $\fontH{L}$ qubits with the $\fontH{W}_1$ qubits (which again requires $\mathcal{O}(\ell)$ $2$-qubit gates), yielding
\begin{align}
&\text{SWAP}_{\fontH{L}, \fontH{W}_1} \left(\frac{1}{\sqrt{D}} \sum_{x \in \{0,1\}^\ell} A|x\rangle_{\fontH{L}} \otimes |x\rangle_{\fontH{W}_1}\right) \otimes \left(\frac{1}{\sqrt{D}} \sum_{y \in \{0,1\}^\ell} |y\rangle_{\fontH{W}_2} \otimes |y\rangle_{\fontH{W}_3}\right) \otimes  |+\rangle_{\fontH{W}_4} \nonumber \\
=\,&\left(\frac{1}{\sqrt{D}} \sum_{x \in \{0,1\}^\ell} |x\rangle_{\fontH{L}} \otimes A|x\rangle_{\fontH{W}_1}\right) \otimes \left(\frac{1}{\sqrt{D}} \sum_{y \in \{0,1\}^\ell} |y\rangle_{\fontH{W}_2} \otimes |y\rangle_{\fontH{W}_3}\right) \otimes  |+\rangle_{\fontH{W}_4} \nonumber \\
=\,&\left(\frac{1}{\sqrt{D}} \sum_{x \in \{0,1\}^\ell} A^T|x\rangle_{\fontH{L}} \otimes |x\rangle_{\fontH{W}_1}\right) \otimes \left(\frac{1}{\sqrt{D}} \sum_{y \in \{0,1\}^\ell} |y\rangle_{\fontH{W}_2} \otimes |y\rangle_{\fontH{W}_3}\right) \otimes  |+\rangle_{\fontH{W}_4}
\end{align}
where we have leveraged the identity\footnote{To see this, write $A = \sum_{x_1, x_2} A_{x_1 x_2}|x_1\rangle \langle x_2|$ and observe that $\sum_x |x\rangle \otimes A |x\rangle = \sum_{x,x_1,x_2} |x\rangle \otimes A_{x_1 x_2} |x_1 \rangle \langle x_2|x\rangle = \sum_{x_1, x_2} A_{x_1 x_2} |x_2\rangle \otimes |x_1\rangle = \sum_{x, x_1, x_2} A_{x_1 x_2} |x_2\rangle \langle x_1|x\rangle \otimes |x\rangle = \sum_x A^T |x\rangle \otimes |x\rangle$\,.}
\begin{equation}
\label{E:maxentidentity1}
\sum_{x \in \{0,1\}^\ell} |x\rangle \otimes A|x\rangle = \sum_{x \in \{0,1\}^\ell} A^T |x\rangle \otimes |x\rangle 
\end{equation}
for any $A$, where $A^T$ is defined as the matrix transpose in the $\ket{x}$ basis, i.e.~$\bra{x_1}A^T\ket{x_2}=\bra{x_2}A\ket{x_1}$. We query the oracle again, giving us
\begin{equation}
    \left(\frac{1}{\sqrt{D}} \sum_{x \in \{0,1\}^\ell} A A^T|x\rangle_{\fontH{L}} \otimes |x\rangle_{\fontH{W}_1}\right) \otimes \left(\frac{1}{\sqrt{D}} \sum_{y \in \{0,1\}^\ell} |y\rangle_{\fontH{W}_2} \otimes |y\rangle_{\fontH{W}_3}\right) \otimes  |+\rangle_{\fontH{W}_4}\,.
\end{equation}
Finally, we perform a controlled swap between $\fontH{L}\otimes \fontH{W}_1$ and $\fontH{W}_2 \otimes \fontH{W}_3$, conditioned on the $|+\rangle$ qubit in $\fontH{W}_4$.  Subsequently measuring the $|+\rangle$ qubit in the $\{|+\rangle, |-\rangle\}$ basis, we will measure
\begin{equation}
\label{eq:symmoutcome1}
\text{Pr}\big[\text{measure }+\big] =  \frac{1}{2} \left(1 + \left|\frac{1}{D}\text{tr}\left\{A A^T\right\}\right|^2 \right)\,, \qquad \text{Pr}\big[\text{measure }-\big] = \frac{1}{2}\left(1 -\left|\frac{1}{D} \text{tr}\left\{A A^T\right\}\right|^2 \right)\,.
\end{equation}
If the lab oracle is $\LO_\ell^O$, then we measure $|+\rangle$ with probability one.  However, if the lab oracle is $\LO_\ell^U$ or $\LO_\ell^\text{Sp}$, we will measure $|+\rangle$ and $|-\rangle$ with probabilities that are exponentially close to $1/2$ (i.e., exponentially suppressed in $\ell$). To see this, recall Markov's inequality for a random variable $X$:
\begin{equation}
\text{Pr}\left[X \geq C\right] \leq \frac{\mathbb{E}[X]}{C}\,.
\end{equation}
Using this, we find in the unitary setting
\begin{align}
\text{Pr}_{A \sim U(D)}\left[\left|\frac{1}{D} \text{tr}\left\{A A^T \right\}\right|^2 \geq \frac{1}{D}\right] &\leq D \int_{U(D)} dU \, \left|\frac{1}{D} \text{tr}\left\{U U^T\right\}\right|^2 \nonumber \\
&= \frac{1}{D} \sum_{a,b,c,d} \int_{U(D)} dU \, U_{ab} U_{ab} U_{cd}^\dagger U_{cd}^\dagger \nonumber \\
&= 
\frac{2}{D+1} 
\end{align}
where we have used Eqn.~\eqref{E:Haarmoment2} to go from the second line to the third line. Similarly in the symplectic setting,
\begin{align}
\text{Pr}_{A \sim \text{Sp}(D/2)}\left[\left|\frac{1}{D} \text{tr}\left\{A A^T \right\}\right|^2 \geq \frac{1}{D}\right] &\leq D \int_{\text{Sp}(D/2)} dS \, \left|\frac{1}{D} \text{tr}\left\{S S^T\right\}\right|^2 \nonumber \\
&= \frac{1}{D} \sum_{a,b,c,d} \int_{\text{Sp}(D/2)} dS \, S_{ab} S_{ab} S_{cd}^\dagger S_{cd}^\dagger \nonumber \\
&= \frac{2}{D+1}
\end{align}
where Eqn.~\eqref{E:sympmoment2} has been used to go from the second line to the third line.
Comparing with~\eqref{eq:symmoutcome1}, we see that in the unitary and symplectic cases $\text{Pr}\big[\text{measure }+\big]$ and $\text{Pr}\big[\text{measure }-\big]$ are exponentially close to $1/2$ with probability exponentially close to one.

In summary, with $\mathcal{O}(1)$ measurements, we can distinguish $\LO_\ell^O$ from both of $\LO_\ell^U$ and $\LO_\ell^\text{Sp}$ with constant bias. 
If we measure $|-\rangle$ in our set of $\mathcal{O}(1)$ measurements, the lab oracle could be either $\LO_\ell^U$ or $\LO_\ell^\text{Sp}$.  Then we repeat the procedure above, but with a minor alteration.  After querying the oracle the first time, we apply $\textbf{J}$ to $\fontH{L}$, giving us
\begin{equation}
    \left(\frac{1}{\sqrt{D}} \sum_{x \in \{0,1\}^\ell} \textbf{J} A |x\rangle_{\fontH{L}} \otimes |x\rangle_{\fontH{W}_1}\right) \otimes \left(\frac{1}{\sqrt{D}} \sum_{y \in \{0,1\}^\ell} |y\rangle_{\fontH{W}_2} \otimes |y\rangle_{\fontH{W}_3}\right) \otimes  |+\rangle_{\fontH{W}_4}\,.
\end{equation}
Note that the application of $\textbf{J}$ is easy, since $\textbf{J} = i \sigma^y \otimes \textbf{1}_2^{\otimes (\ell-1)}$.  
Next we swap $\fontH{L}$ with $\fontH{W}_1$ as before.  But before the second oracle query, we apply $\textbf{J}$ to $\fontH{L}$, at this point giving us
\begin{align}
    &\left(\frac{1}{\sqrt{D}} \sum_{x \in \{0,1\}^\ell} \textbf{J} |x\rangle_{\fontH{L}} \otimes \textbf{J} A |x\rangle_{\fontH{W}_1}\right) \otimes \left(\frac{1}{\sqrt{D}} \sum_{y \in \{0,1\}^\ell} |y\rangle_{\fontH{W}_2} \otimes |y\rangle_{\fontH{W}_3}\right) \otimes  |+\rangle_{\fontH{W}_4} \nonumber \\
        =\,&\left(\frac{1}{\sqrt{D}} \sum_{x \in \{0,1\}^\ell} \textbf{J} A^T \textbf{J}^T |x\rangle_{\fontH{L}} \otimes |x\rangle_{\fontH{W}_1}\right) \otimes \left(\frac{1}{\sqrt{D}} \sum_{y \in \{0,1\}^\ell} |y\rangle_{\fontH{W}_2} \otimes |y\rangle_{\fontH{W}_3}\right) \otimes  |+\rangle_{\fontH{W}_4} \nonumber \\
    =\,&\left(\frac{1}{\sqrt{D}} \sum_{x \in \{0,1\}^\ell} (-\textbf{J} A^T \textbf{J}) |x\rangle_{\fontH{L}} \otimes |x\rangle_{\fontH{W}_1}\right) \otimes \left(\frac{1}{\sqrt{D}} \sum_{y \in \{0,1\}^\ell} |y\rangle_{\fontH{W}_2} \otimes |y\rangle_{\fontH{W}_3}\right) \otimes  |+\rangle_{\fontH{W}_4} 
\end{align}
where we have used Eqn.~\eqref{E:maxentidentity1}
as well as $\textbf{J}^T = - \textbf{J}$. 

Then querying the oracle again, performing a controlled swap between $\fontH{L} \otimes \fontH{W}_1$ and $\fontH{W}_2 \otimes \fontH{W}_3$ with $|+\rangle$ on $\fontH{W}_4$ being the control qubit, and subsequently measuring the control qubit in the $\{|+\rangle, |-\rangle\}$ basis, we find
\begin{equation}
\label{Eq:symmoutcome2}
\text{Pr}\big[\text{measure }+\big] =  \frac{1}{2} \left(1 + \left|\frac{1}{D}\text{tr}\left\{A (-\textbf{J}A^T \textbf{J})\right\}\right|^2 \right)\,, \quad \text{Pr}\big[\text{measure }-\big] = \frac{1}{2}\left(1 -\left|\frac{1}{D} \text{tr}\left\{A (-\textbf{J}A^T \textbf{J})\right\}\right|^2 \right)\,.
\end{equation}
Then if the lab oracle is $\LO_\ell^{\text{Sp}}$ by Eqn.~\eqref{eq:timereveralJ} we will measure $|+\rangle$ with probability one, whereas if the lab oracle is $\LO_\ell^{U}$ we will measure $|+\rangle$ and $|-\rangle$ with probabilities exponentially close to $1/2$.  To see this, we again use the Markov inequality as
\begin{align}
\text{Pr}_{A \sim U(D)}\left[\left|\frac{1}{D} \text{tr}\left\{A (-\textbf{J}A^T\textbf{J}) \right\}\right|^2 \geq \frac{1}{D}\right] &\leq D \int_{U(D)} dU \, \left|\frac{1}{D} \text{tr}\left\{U (-\textbf{J} U^T\textbf{J})\right\}\right|^2 \nonumber \\
&= \frac{1}{D} \sum_{a,b,c,d,e,f,g,h} \int_{U(D)} dU \, U_{ab} \textbf{J}_{bc} U_{dc} \textbf{J}_{da} U_{ef}^\dagger \textbf{J}_{fg} U_{hg}^\dagger \textbf{J}_{he}\nonumber \\
&= 
\frac{2}{D-1}
\end{align}
where we have used Eqn.~\eqref{E:Haarmoment2} to go from the second line to the third line.
Then indeed, for Eqn.~\eqref{Eq:symmoutcome2} in the unitary case, $\text{Pr}\big[\text{measure }+\big]$ and $\text{Pr}\big[\text{measure }-\big]$ are exponentially close to $1/2$ with probability exponentially close to unity.
 As such, we can distinguish between $\LO_\ell^{\text{Sp}}$ and $\LO_\ell^{U}$ with constant bias. 

In totality, we can efficiently distinguish between $\LO_\ell^U, \LO_\ell^O, \LO_\ell^{\text{Sp}}$ with constant bias, using $\mathcal{O}(1)$ queries of 
the oracle and $\mathcal{O}(\ell)$ gates.
\end{proof}

\subsection{Proof of Theorem~\ref{thm:main2}}

We use a similar proof strategy for Theorem~\ref{thm:main2} as we did for Theorem~\ref{thm:main}.  The strategy in the latter case was to first prove Lemma~\ref{thm:tvd} which established Theorem~\ref{thm:main} in the simple measurement QUALM setting, and then to leverage Claim~\ref{cl:derandomization} to generalize to the full incoherent access QUALM setting.  For Theorem~\ref{thm:main2}, we first prove a Lemma which establishes the result in the simple measurement QUALM setting.  Then again leveraging Claim~\ref{cl:derandomization}, we obtain the full incoherent access QUALM setting.

So let us focus on establishing Theorem~\ref{thm:main2} for the simple measurement QUALM setting.  Indeed, it suffices to prove that the following three probability distributions are exponentially close in total variational distance.  Consider the probability distributions over $k$ $\ell$-bit strings, using the notation of Definition~\ref{def:SPM}:
\begin{align}
    Q_k^U(s_0, s_1, ..., s_k) &= \text{Pr}(s_0)\cdot\left(\int_{U(D)}\!\!dU \, \prod_{i=1}^k\langle y_{s_0 s_1 ... s_i}^i| U \sigma_{s_0 s_1...s_{i-1}}^{i-1} U^\dagger |y_{s_0 s_1...s_i}^i\rangle \lambda_{s_0 s_1...s_i}^i\right) \label{eq:QkU}\\
      Q_k^O(s_0, s_1, ..., s_k) &= \text{Pr}(s_0)\cdot\left(\int_{O(D)}\!\!dO \, \prod_{i=1}^k\langle y_{s_0 s_1 ... s_i}^i| O \sigma_{s_0 s_1...s_{i-1}}^{i-1} O^\dagger |y_{s_0 s_1...s_i}^i\rangle \lambda_{s_0 s_1...s_i}^i\right)  \label{eq:QkO}\\
      Q_k^{\text{Sp}}(s_0, s_1, ..., s_k) &= \text{Pr}(s_0)\cdot\left(\int_{\text{Sp}(D/2)}\!\!dS \, \prod_{i=1}^k\langle y_{s_0 s_1 ... s_i}^i| S \sigma_{s_0 s_1...s_{i-1}}^{i-1} S^\dagger |y_{s_0 s_1...s_i}^i\rangle \lambda_{s_0 s_1...s_i}^i\right)  \label{eq:QkSp}
\end{align}
where $\text{Pr}(s_0) = |\langle y_{s_0}^0|0^{\fontH{L}}\rangle|^2 \lambda_{s_0}^0$\,.  In order to show that $Q_k^U$, $Q_k^O$, $Q_k^{\text{Sp}}$ are all exponentially close in total variational distance, it suffices that they are all exponentially close to a single probability distribution
\begin{equation}
    P_k(s_0,s_1,...,s_k) = \text{Pr}(s_0) D^{-k} \prod_{i=1}^k \lambda_{s_i}^i
\end{equation}
since we can then use the triangle inequality to the bound the pairwise distances we seek.

We have the following Lemma, which generalizes Lemma~\ref{thm:tvd}:
\begin{lemma}
\label{lemm:symm}
Let $Q_k^U, Q_k^O, Q_k^{\text{Sp}}$ and $P_k$ be defined as above, where $\ket{y_{s_0s_1...s_i}^i}$ are normalized pure states in $\fontH{L}$ and $\sigma_{s_0s_1...s_{i-1}}^{i-1}$ are normalized density operators in the same Hilbert space. We have $0\le \lambda_{s_0s_1...s_i}^i\le 1$ and $\sum_{s_i}\lambda_{s_0s_1...s_i}^i\ket{y_{s_0s_1...s_i}^i}\bra{y_{s_0s_1...s_i}^i}=\mathds{1}$. The lab subsystem $\fontH{L}$ consists of $\ell$ qubits, and the number of queries to the lab oracle $k$ satisfies $k<\left(2^\ell/12\right)^{2/7}$. Furthermore, $s_j=1,2,...,D_m$ is a set of labels for each $j=0,1,...,k$.  Then 
\begin{align}
\label{eq:QkUcase1}
    \|P_k-Q_k^U\|_1 &=\sum_{s}\left|P_k\left(s\right)-Q_k^U\left(s\right)\right|\leq \mathcal{O}(k^3/2^\ell)  \\
    \|P_k-Q_k^O\|_1 &=\sum_{s}\left|P_k\left(s\right)-Q_k^O\left(s\right)\right|\leq \mathcal{O}(k^3/2^\ell) \\
    \|P_k-Q_k^{\text{Sp}}\|_1  &=\sum_{s}\left|P_k\left(s\right)-Q_k^{\text{Sp}}\left(s\right)\right|\leq \mathcal{O}(k^3/2^\ell)\,.
\end{align}
\end{lemma}
\noindent As remarked above, we see that $Q_k^U$, $Q_k^O$ and $Q_k^{\text{Sp}}$ are pairwise close to one another with distance $\mathcal{O}(1/2^\ell)$ by the triangle inequality.

The proof will require several facts about Weingarten functions of the orthogonal and symplectic ensembles.  These are reviewed in a fair amount of detail in~\ref{App:reviewHaar}, so presently we will just provide some basic notations and definitions.  Let $S_k$ denote the set of permutations on $k$ elements, and let $P_2(2k)$ denote the set of pair permutations on $2k$ elements.

Since $P_2(2k)$ is less familiar than the set of ordinary permutations, let us delve into it further.  Recall that a pair permutation is a partition of $2k$ elements into $k$ pairs.  An element $\mathfrak{m} \in P_2(2k)$ is written as
\begin{equation}
\mathfrak{m} = \{\mathfrak{m}(1), \mathfrak{m}(2)\}\{\mathfrak{m}(3),\mathfrak{m}(4)\} \cdots \{\mathfrak{m}(2k-1),\mathfrak{m}(2k)\}
\end{equation}
with the ordering convention $\mathfrak{m}(2i-1) < \mathfrak{m}(2i)$ for $1 \leq i \leq k$ and $\mathfrak{m}(1) <\mathfrak{m}(3) < \cdots < \mathfrak{m}(2k-1)$.  We let $\mathfrak{e}$ denote the identity pairing $\mathfrak{e} = \{1,2\}\{3,4\}\cdots\{2k-1,2k\}$.  There is a canonical map from each $\mathfrak{m}$ to a unique permutation $\sigma_{\mathfrak{m}}$ in $S_{2k}$ satisfying $\sigma_{\mathfrak{m}}(i) = \mathfrak{m}(i)$.  We will often use $\mathfrak{m}$ and its counterpart in $S_{2k}$ interchangeably.

For the proof, it will be convenient to define the following $D^k\times D^k$ matrices $\Delta_{\mathfrak{m}}$ and $\Delta_{\mathfrak{m}}'$ (see~\cite{Matsumoto1, CollinsMatsumoto1}).  We let  
\begin{equation}
\left[\Delta_{\mathfrak{m}}\right]_{i_1i_3...i_{2k-1}}^{i_2i_4...i_{2k}} = \prod_{s=1}^k \delta_{i_{\mathfrak{m}(2s-1)}, i_{\mathfrak{m}(2s)}}\label{eq:Deltan}
\end{equation}
and
\begin{align}
\left[\Delta_{\mathfrak{m}}'\right]_{i_1i_3...i_{2k-1}}^{i_2i_4...i_{2k}} = \prod_{s=1}^k \textbf{J}_{i_{\mathfrak{m}(2s-1)}, i_{\mathfrak{m}(2s)}}\,,\label{eq:Deltanprime}
\end{align}
where $\textbf{J}$ is the $D \times D$ canonical symplectic form in Eqn.~\eqref{eq:J}. We note that for $\mathfrak{m}=\mathfrak{e}$, we have 
\begin{equation}\label{eq:idorthogonal} 
\Delta_{\mathfrak{e}}=
\mathds{1}^{\otimes k}
\end{equation}
and 
\begin{equation}\label{eq:idsymplextic} 
\Delta_{\mathfrak{e}}'=\textbf{J}^{\otimes k}.
\end{equation} 

We will further use the multi-index notations $II' = (i_1,i_1'...,i_{k},i_k')$ and similarly and $A_{MN} = A_{m_1, n_1} A_{m_2, n_2} \cdots A_{m_k, n_k}$ for any $D \times D$ matrix $A$.

\begin{figure}
    \centering
    \includegraphics[width=5.5in]{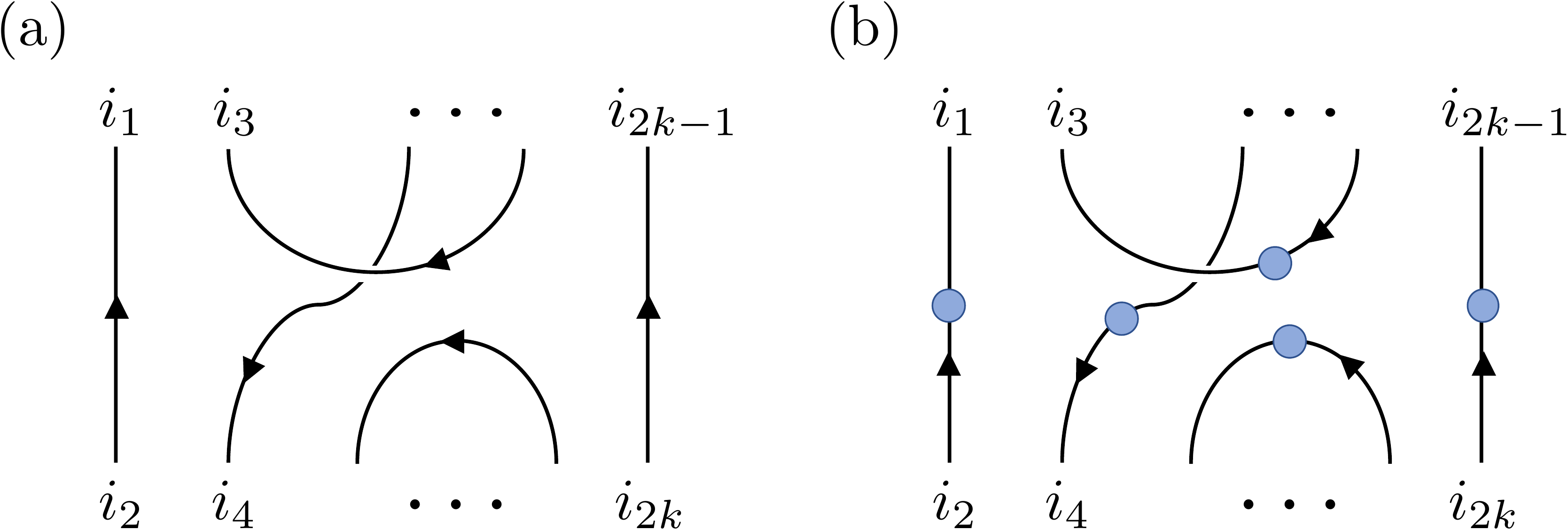}
    \caption{Illustration of the matrix elements of the operators $\Delta_{\mathfrak{m}}$ and $\Delta_{\mathfrak{m}}'$ in Eqn.'s~\eqref{eq:Deltan} and~\eqref{eq:Deltanprime}, in (a) and (b) respectively.  Fixed values of odd indices $i_1i_3...i_{2k-1}$ specify a row of the matrix, and fixed values of the even indices $i_2i_4...i_{2k}$ specify a column of the matrix. As usual, we use the tensor diagram notation explained in~\ref{App:diagrams}. Each line in (a) represents a Kronecker $\delta$, such as $\delta_{i_1i_2}$. Each line in (b) with a blue solid circle represent a $\textbf{J}$ matrix element such as $\textbf{J}_{i_1i_2}$.}   
    \label{fig:permutationsOSp}
\end{figure}

With these notations in mind, we commence with the proof of Lemma~\ref{lemm:symm}.

\begin{proof}
We have previously proven the $Q_k^U$ case in Eqn.~\eqref{eq:QkUcase1} -- this is the same as Lemma~\ref{thm:tvd}.  But we still need to prove the $Q_k^O$ and $Q_k^{\text{Sp}}$ cases.  In fact, the proof of the $Q_k^U$ case in Lemma~\ref{thm:tvd} is somewhat modular: it contains structure that is common to all three cases, with certain ingredients that need to be modified on a case-by-case basis.

We begin by defining
\begin{align}
A_s=\bigotimes_{i=1}^k \sigma_{s_0 s_1...s_{i-1}}^{i-1}\,,\qquad B_s =\bigotimes_{i=1}^k |y_{s_0 s_1 ... s_i}^i\rangle \langle y_{s_0 s_1...s_i}^i| \, \lambda_{s_0 s_1...s_i}^i
\end{align}
so that
\begin{align}
Q_k^U(s) &= \sum_{I,J,I',J'} \int_{U(D)}dU \, U_{IJ} A_{s\,JI'} U_{I'J'}^\dagger B_{s\,J'I} \\
Q_k^O(s) &= \sum_{I,J,I',J'} \int_{O(D)}dO \,O_{IJ} A_{s\,JI'} O_{I'J'}^T B_{s\,J'I} \\
Q_k^{\text{Sp}}(s) &= \sum_{I,J,K,L,M,N} \int_{Sp(D/2)}dS \,S_{IJ} A_{s\,JK} \textbf{J}_{KL} S_{LM}^T \textbf{J}_{MN}^T B_{s\,NI}\,.
\end{align}
In~\ref{App:reviewHaar} we use the Weingarten calculus to arrive at the following equalities (see Eqn.'s~\eqref{E:mainhaar1},~\eqref{eq:orthintegral1}, and~\eqref{eq:symplecintegral1} and the discussions surrounding them): 
\begin{align}
Q_k^U(s) &= \sum_{\sigma, \tau \in S_k} \text{tr}(\sigma A_s) \,\text{tr}(\tau^{-1} B_s) \, \text{Wg}^U(\tau \sigma^{-1}, D) \\
Q_k^O(s) &= \sum_{\mathfrak{m}, \mathfrak{n} \in P_2(2k)} \text{tr}(\Delta_{\mathfrak{m}}A_s) \,\text{tr}(\Delta_{\mathfrak{n}}B_s) \, \text{Wg}^O(\sigma_{\mathfrak{n}} \sigma_{\mathfrak{m}}^{-1}, D) \\
Q_k^O(s) &= \sum_{\mathfrak{m}, \mathfrak{n} \in P_2(2k)} \text{tr}(\Delta'_{\mathfrak{m}}A_s \textbf{J}^{\otimes k}) \,\text{tr}(\Delta'_{\mathfrak{n}}\textbf{J}^{T\otimes k} B_s) \, \text{Wg}^\text{Sp}(\sigma_{\mathfrak{n}} \sigma_{\mathfrak{m}}^{-1}, D/2)\,.
\end{align}
Noting that $|\text{tr}(A_s)|\,, |\text{tr}(B_s)|\,, |\text{tr}(\Delta_{\mathfrak{e}}A_s)|,\, |\text{tr}(\Delta_{\mathfrak{e}}B_s)|,\, |\text{tr}(\Delta'_{\mathfrak{e}}A_s \textbf{J}^{\otimes k})|,\, |\text{tr}(\Delta'_{\mathfrak{e}}\textbf{J}^{T\otimes k} B_s)| =1$, we can leverage the expressions for $Q_k^U(s), Q_k^O(s), Q_k^{\text{Sp}}(s)$ above to find the following inequalities:
\begin{align}
\label{eq:QkUbd1}
|Q_k^U(s) - P_k(s)| &\leq |\text{Wg}^U(\mathds{1},D) - D^{-k}| + \sum_\sigma \sum_{\tau \not = \mathds{1}} |\text{tr}(\sigma A_s)| \,|\text{tr}(\tau^{-1} B_s)| \, |\text{Wg}^U(\tau \sigma^{-1}, D)| \nonumber \\
& \qquad \qquad \qquad \qquad \qquad \qquad \qquad + \sum_{\sigma \not = \mathds{1}} |\text{tr}(\sigma A_s)| \, |\text{Wg}^U(\sigma^{-1}, D)|\,, \\
\label{eq:QkObd1}
|Q_k^O(s) - P_k(s)| &\leq |\text{Wg}^O(\sigma_{\mathfrak{e}},D) - D^{-k}| +  \sum_{\mathfrak{m}} \sum_{\mathfrak{n} \not = \mathfrak{e}} |\text{tr}(\Delta_{\mathfrak{m}}A_s)| \,|\text{tr}(\Delta_{\mathfrak{n}}B_s)| \, |\text{Wg}^O(\sigma_{\mathfrak{n}} \sigma_{\mathfrak{m}}^{-1}, D)| \nonumber \\
& \qquad \qquad \qquad \qquad \qquad \qquad \qquad + \sum_{\mathfrak{m} \not = \mathfrak{e}} |\text{tr}(\Delta_{\mathfrak{m}}A_s)| \, |\text{Wg}^O(\sigma_{\mathfrak{m}}^{-1}, D)|\,, \\
\label{eq:QkSpbd1}
|Q_k^\text{Sp}(s) - P_k(s)| &\leq |\text{Wg}^{\text{Sp}}(\sigma_{\mathfrak{e}},D/2) - D^{-k}| +  \sum_{\mathfrak{m}} \sum_{\mathfrak{n} \not = \mathfrak{e}} |\text{tr}(\Delta_{\mathfrak{m}}'A_s \textbf{J}^{\otimes k})| \,|\text{tr}(\Delta_{\mathfrak{n}}\textbf{J}^{T\otimes k} B_s)| \, |\text{Wg}^\text{Sp}(\sigma_{\mathfrak{n}} \sigma_{\mathfrak{m}}^{-1}, D/2)| \nonumber \\
& \qquad \qquad \qquad \qquad \qquad +\sum_{\mathfrak{m}\neq \mathfrak{e}}|\text{tr}(\Delta_{\mathfrak{m}}'A_s\textbf{J}^{\otimes k})| \, |\text{Wg}^\text{Sp}(\sigma_{\mathfrak{m}}^{-1}, D/2)|\,.
\end{align}
To derive these inequalities, we have repeatedly used the triangle inequality, along with a convenient rearranging of the sums. This is directly analogous to the derivation of \eqref{eq:threeterms1} (in fact, \eqref{eq:QkUbd1} is the same as \eqref{eq:threeterms1}).  To simplify the inequalities above, we will need that $|\text{tr}(\sigma A_s)| \leq 1$, $|\text{tr}(\Delta_{\mathfrak{n}}A_s)| \leq 1$, $|\text{tr}(\Delta'_{\mathfrak{n}} A_s \textbf{J}^{\otimes k})| \leq 1$.  The arguments for these three inequalities are given below:
\begin{enumerate}
\item In the unitary case,
$|\text{tr}(\sigma A_s)| \leq 1$ since it is the absolute value of a product of Hilbert-Schmidt inner products of normalized density matrices. 
\item In the orthogonal case, we bound $|\text{tr}(\Delta_{\mathfrak{m}}A_s)|$ as follows.  Recalling that $A_s$ is a tensor product of $k$ density matrices, we can write each constituent density matrix in its respective eigenbasis so that we have
\begin{align}
A_s &= \sum_{j_1, j_2,...,j_k} p_{j_1}^{(1)} p_{j_2}^{(2)} \cdots p_{j_k}^{(k)} |\phi_{j_1}^{(1)}\rangle \langle \phi_{j_1}^{(1)}| \otimes |\phi_{j_2}^{(2)}\rangle \langle \phi_{j_2}^{(2)}| \otimes \cdots \otimes |\phi_{j_k}^{(k)}\rangle \langle \phi_{j_k}^{(k)}| \nonumber \\
&= \sum_{j_1, j_2,...,j_k} p_{j_1}^{(1)} p_{j_2}^{(2)} \cdots p_{j_k}^{(k)} \Phi_{j_1 j_2 \cdots j_k}\,.
\end{align}
We have introduced $\Phi_{j_1 j_2 \cdots j_k}$ to compress the notation slightly.  Now notice that
\begin{align}
|\text{tr}(\Delta_{\mathfrak{m}}A_s)| \leq \sum_{j_1, j_2,...,j_k} p_{j_1}^{(1)} p_{j_2}^{(2)} \cdots p_{j_k}^{(k)} \,|\text{tr}(\Delta_{\mathfrak{m}}\Phi_{j_1 j_2 \cdots j_k})|\,.
\end{align}
Here $\text{tr}(\Delta_{\mathfrak{m}}\Phi_{j_1 j_2 \cdots j_k})$ is just a product of terms of the form $\langle \phi_{j_\ell}^{\ell}| \phi_{j_{\ell'}}^{\ell'}\rangle$, $\langle \phi_{j_\ell}^{\ell}|(| \phi_{j_{\ell'}}^{\ell'}\rangle)^*$, $(\langle \phi_{j_\ell}^{\ell}|)^*| \phi_{j_{\ell'}}^{\ell'}\rangle$, or $(\langle \phi_{j_\ell}^{\ell}| \phi_{j_{\ell'}}^{\ell'}\rangle)^*$\,; these are just inner products of pure states.  But for any two pure (normalized) states $|\phi\rangle, |\phi'\rangle$, we have $| \langle \phi | \phi'\rangle | \leq 1$, and so $|\text{tr}(\Delta_{\mathfrak{m}}\Phi_{j_1 j_2 \cdots j_k})| \leq 1$. Accordingly,
\begin{equation}
|\text{tr}(\Delta_{\mathfrak{m}}A_s)| \leq \sum_{j_1, j_2,...,j_k} p_{j_1}^{(1)} p_{j_2}^{(2)} \cdots p_{j_k}^{(k)} = 1\,,
\end{equation}
thus establishing $|\text{tr}(\Delta_{\mathfrak{m}}A_s)| \leq 1$.
\item In the symplectic case, we consider $|\text{tr}(\Delta_{\mathfrak{m}}'A_s  \textbf{J}^{\otimes k})|$.  We can bound this with the same strategy as the orthogonal case above, since $\text{tr}(\Delta_{\mathfrak{m}}'\Phi_{j_1 j_2 \cdots j_k}  \textbf{J}^{\otimes k})$ is also a product of inner products of pure states (recall that $\textbf{J}$ is unitary).  Thus $|\text{tr}(\Delta_{\mathfrak{m}}'A_s  \textbf{J}^{\otimes k})| \leq 1$.
\end{enumerate}
As such, we have the needed bounds in all three cases.  Using the bounds in Eqn.'s~\eqref{eq:QkUbd1},~\eqref{eq:QkObd1} and~\eqref{eq:QkSpbd1}, the Cauchy-Schwarz inequality, and some basic properties 
of the Weingarten functions (i.e., we perform manipulations directly analogous to the derivation of~\eqref{eq:threeterms2}  in the proof of Lemma~\ref{thm:tvd}), we find
\begin{align}
\label{eq:QkUbd2}
|Q_k^U(s) - P_k(s)| &\leq |\text{Wg}^U(\mathds{1},D) - D^{-k}| + \sum_\sigma |\text{Wg}^U(\sigma, D)| \sum_{\tau \not = \mathds{1}} |\text{tr}(\tau^{-1} B_s)| + \sum_{\sigma \not = \mathds{1}} |\text{Wg}^U(\sigma, D)|\,, \\
\label{eq:QkObd2}
|Q_k^O(s) - P_k(s)| &\leq |\text{Wg}^U(\sigma_{\mathfrak{e}},D) - D^{-k}| +  \sum_{\mathfrak{m}} |\text{Wg}^O(\sigma_{\mathfrak{m}}, D)| \sum_{\mathfrak{n} \not = \mathfrak{e}} |\text{tr}(\Delta_{\mathfrak{n}}B_s)| + \sum_{\mathfrak{m} \not = \mathfrak{e}}  |\text{Wg}^O(\sigma_{\mathfrak{m}}, D)|\,, \\
\label{eq:QkSpbd2}
|Q_k^\text{Sp}(s) - P_k(s)| &\leq |\text{Wg}^U(\sigma_{\mathfrak{e}},D/2) - D^{-k}| +  \sum_{\mathfrak{m}} |\text{Wg}^\text{Sp}(\sigma_{\mathfrak{m}}, D/2)| \sum_{\mathfrak{n} \not = \mathfrak{e}}  |\text{tr}(\Delta_{\mathfrak{n}}'\textbf{J}^{T\otimes k} B_s)| \nonumber \\
&\qquad \qquad \qquad \qquad \qquad \qquad \qquad \qquad \qquad \qquad \qquad + \sum_{\mathfrak{m} \not = \mathfrak{e}} |\text{Wg}^\text{Sp}(\sigma_{\mathfrak{m}}, D/2)|\,.
\end{align}

We organize the above information by writing
\begin{align}
&\sum_{s} |Q_k^U(s) - P_k(s)| \leq c_1^U + c_2^U T^U \\
&\sum_{s} |Q_k^O(s) - P_k(s)| \leq c_1^O + c_2^O T^O \\
&\sum_{s} |Q_k^{\text{Sp}}(s) - P_k(s)| \leq c_1^{\text{Sp}} + c_2^{\text{Sp}} T^{\text{Sp}}
\end{align}
where
\begin{align}
&c_1^U = D^k \bigg(|\text{Wg}^U(\mathds{1},D) - D^{-k}| + \sum_{\sigma \not = \mathds{1}} |\text{Wg}^U(\sigma,D)|\bigg)\,, \,\, c_2^U = D^k \sum_{\sigma} |\text{Wg}^U(\sigma,D)|\,, \,\, T^U = \frac{1}{D^k}\sum_{\tau \not = \mathds{1}} \sum_s |\text{tr}(\tau^{-1} B_s)| \\
&c_1^O = D^k \bigg(|\text{Wg}^O(\sigma_{\mathfrak{e}},D) - D^{-k}| + \sum_{\mathfrak{m} \not = \mathfrak{e}} |\text{Wg}^O(\sigma_{\mathfrak{m}},D)|\bigg)\,, \,\, c_2^O = D^k \sum_{\mathfrak{m}} |\text{Wg}^O(\sigma_{\mathfrak{m}},D)|\,, \,\, T^O = \frac{1}{D^k}\sum_{\mathfrak{n} \not = \mathfrak{e}} \sum_s |\text{tr}(\Delta_{\mathfrak{n}}B_s)| \\
&c_1^{\text{Sp}} = D^k \bigg(|\text{Wg}^{\text{Sp}}(\sigma_{\mathfrak{e}},D) - D^{-k}| + \sum_{\mathfrak{m} \not = \mathfrak{e}} |\text{Wg}^{\text{Sp}}(\sigma_{\mathfrak{m}},D/2)|\bigg)\,, \,\, c_2^{\text{Sp}} = D^k \sum_{\mathfrak{m}} |\text{Wg}^{\text{Sp}}(\sigma_{\mathfrak{m}},D/2)|\,, \nonumber \\
& \qquad \qquad \qquad \qquad \qquad \qquad \qquad \qquad \qquad \qquad \qquad \qquad \qquad \qquad \qquad \qquad T^{\text{Sp}} = \frac{1}{D^k}\sum_{\mathfrak{n} \not = \mathfrak{e}} \sum_s |\text{tr}(\Delta'_{\mathfrak{n}}\textbf{J}^{T\otimes k} B_s)|\,.
\end{align}
So we need to appropriately bound the $c_1$'s, $c_2$'s and $T$'s.  We do so in turn. \\ \\
\textbf{Bounding the $c_1$\!'s.} We use the standard bounds on Weingarten functions (discussed in~\ref{App:reviewHaar}).
\begin{enumerate}
\item In the unitary case, for $k<\left(2^\ell/\sqrt{6}\right)^{4/7}$, we have $|\text{Wg}^U(\mathds{1},D) - D^{-k}| = \mathcal{O}(k^{7/2} D^{-(k+2)})$ by Lemma~\ref{lemm1AppB1} in~\ref{sec:IntroWeingerten}.  Additionally, $\sum_{\sigma \in S_{k}} |\text{Wg}^U(\sigma,D)| = \frac{(D-k)!}{D!} = D^{-k} + \mathcal{O}(k^2 D^{-(k+1)})$ by Lemma~\ref{lemm2AppB1} in~\ref{sec:IntroWeingerten}.  Then $c_1^U = \mathcal{O}(k^{7/2} D^{-2})$.
\item In the orthogonal case, for $k<\left(2^\ell/12\right)^{2/7}$, we have $|\text{Wg}^O(\sigma_{\mathfrak{e}},D) - D^{-k}| = \mathcal{O}(k^{7} D^{-(k+2)})$ by Lemma~\ref{lemm1AppB2} in~\ref{sec:IntroWeingerten2}.  We also have $\sum_{\mathfrak{m} \in P_2(2k)} |\text{Wg}^O(\sigma_{\mathfrak{m}},2D)| = \frac{(D-k)!!}{D!!} = D^{-k} + \mathcal{O}(k^2 D^{-(k+1)})$ by Lemma~\ref{lemm2AppB2} in~\ref{sec:IntroWeingerten2}.  Taken together, we find $c_1^{O} = \mathcal{O}(k^{7} D^{-2})$.
\item In the symplectic case, for $k<\left(2^{\ell}/6\right)^{2/7}$, we have $|\text{Wg}^\text{Sp}(\sigma_{\mathfrak{e}},D/2) - D^{-k}| = \mathcal{O}(k^{7/2} D^{-(k+2)})$ by Lemma~\ref{lemm1AppB3} in~\ref{sec:IntroWeingerten3}.  Furthermore, we have $\sum_{\mathfrak{m} \in P_2(2k)} |\text{Wg}^\text{Sp}(\sigma_{\mathfrak{m}},D/2)|  = \prod_{j=0}^{k-1} \frac{1}{D + 2j} = D^{-k} + \mathcal{O}(k^2 D^{-(k+1)})$ by Lemma~\ref{lemm2AppB3} in~\ref{sec:IntroWeingerten3}.  Altogether, $c_1^{\text{Sp}} = \mathcal{O}(k^{7/2} D^{-2})$.
\end{enumerate}
For $k<\left(2^\ell/12\right)^{2/7}$, all three conditions above are applicable.
$$$$
\textbf{Bounding the $c_2$\!'s.}  Here we can use the same bounds (i) $\sum_{\sigma \in S_{k}} |\text{Wg}^U(\sigma,D)| = D^{-k} + \mathcal{O}(k^2 D^{-(k+1)})$, (ii) $\sum_{\mathfrak{m} \in P_2(2k)} |\text{Wg}^O(\sigma_{\mathfrak{m}},2D)| = D^{-k} + \mathcal{O}(k^2 D^{-(k+1)})$, and (iii) $\sum_{\mathfrak{m} \in P_2(2k)} |\text{Wg}^\text{Sp}(\sigma_{\mathfrak{m}},D/2)|  =  D^{-k} + \mathcal{O}(k^2 D^{-(k+1)})$ that we used for the $c_1$'s.  Multiplying (i), (ii) and (iii) by $D^k$, we find:
\begin{enumerate}
    \item $c_2^{U} = 1 + \mathcal{O}(k^2 D^{-1})$.
    \item $c_2^{O} = 1 + \mathcal{O}(k^2 D^{-1})$.
    \item $c_2^{\text{Sp}} = 1 + \mathcal{O}(k^2 D^{-1})$.
\end{enumerate}
Finally, we need to bound the $T$'s.
\\ \\
\textbf{Bounding the $T$'s.} 
\begin{figure}
    \centering
    \includegraphics[width=4.5in]{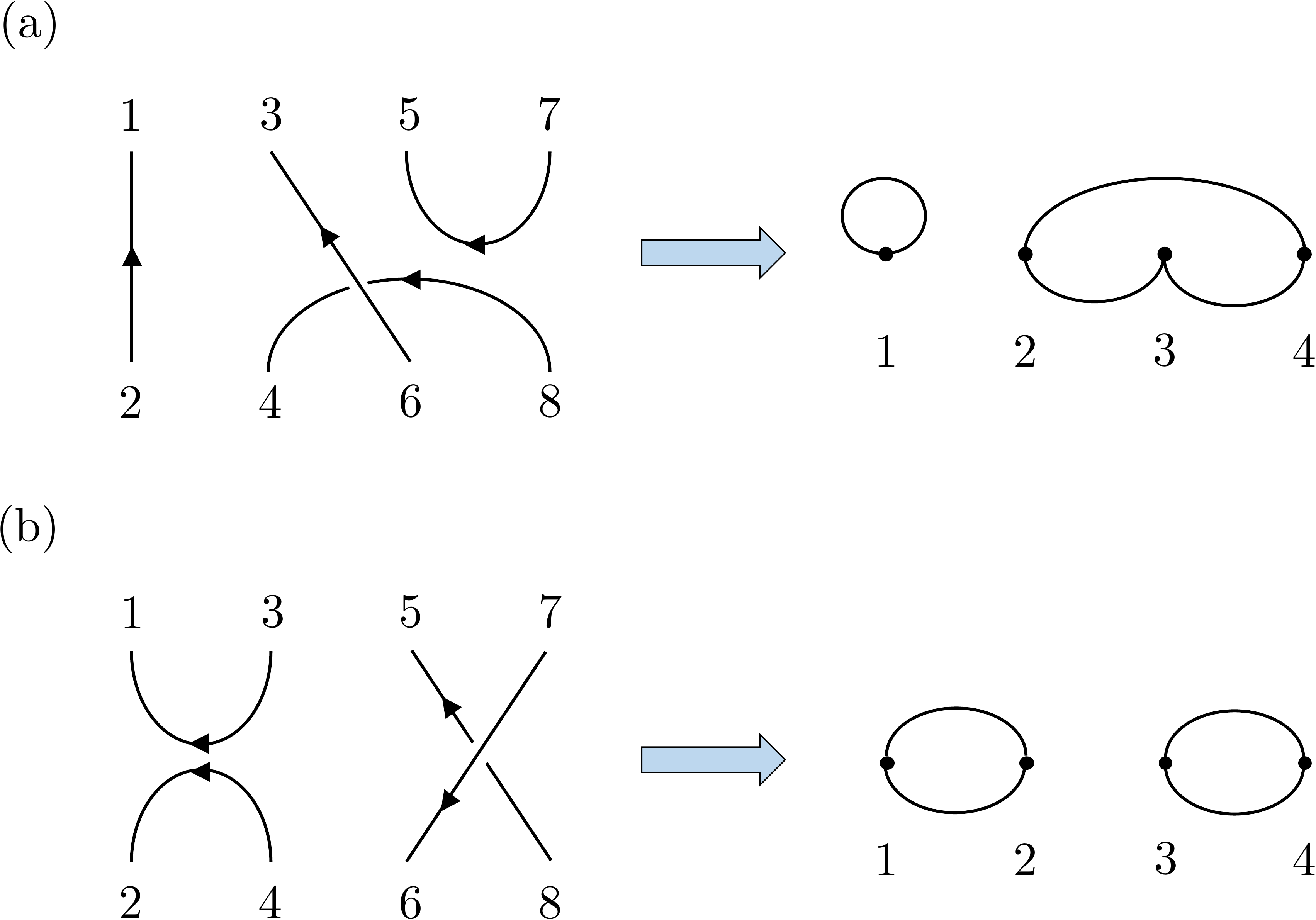}
    \caption{Two examples of the procedure that maps pair partitions $\mathfrak{m}$ to unoriented loops. (a) $\mathfrak{m}=\left\{12\right\}\left\{36\right\}\left\{48\right\}\left\{57\right\}$ corresponds to $(1)(234)$. (b) $\mathfrak{m}=\left\{13\right\}\left\{24\right\}\left\{58\right\}\left\{67\right\}$ corresponds to (12)(34).}
    \label{fig:ortholoop}
\end{figure}
\begin{enumerate}
\item The unitary case was already established in Lemma~\ref{thm:tvd} above, and we found the bound \newline $\frac{1}{D^k}\sum_s \sum_{\tau \not = \mathds{1}} |\text{tr}(\tau^{-1} B_s)| \leq \mathcal{O}\left(k^3/D\right)$.
\item Now we treat the orthogonal case. It is convenient to define the following $2k$ states:
\begin{align}
    \ket{\psi_{2i-1}}=\ket{y_{s_0s_1...s_i}^i},~\ket{\psi_{2i}}=\ket{y_{s_0s_1...s_i}^i}^*\,.\label{eq:psi def}
\end{align}
Then the trace is a product of state overlaps similar to the unitary case:
\begin{align}
    \left|{\rm tr}\left(\Delta_{\mathfrak{m}}B_s\right)\right|=\prod_{i=1}^k\left|\bra{\psi_{\mathfrak{m}(2i-1)}}\left(\ket{\psi_{\mathfrak{m}(2i)}}^*\right)\right|\lambda^i_{s_0s_1...s_i}\label{eq:Delta m string}
\end{align}
Now we define a procedure that relates each $\mathfrak{m}$ to loops in a $k$-vertex graph. Take $k$ vertices, and add an unoriented link between $i$ and $j$ if $\mathfrak{m}$ contains one of the following pairs: $\left\{2i-1,2j-1\right\}$, $\left\{2i-1,2j\right\}$, $\left\{2i,2j-1\right\}$, $\left\{2i,2j\right\}$. $\mathfrak{m}$ could contain two of these pairs. For example if $\mathfrak{m}$ contains $\left\{13\right\}\left\{24\right\}$, we should connect two links between $1,2$ in the graph. As is illustrated in Fig.~\ref{fig:ortholoop}, pictorially the graph is obtained by simply replacing a pair of points $2i-1,2i$ in the diagram representing $\Delta_{\mathfrak{m}}$ by a single point. This procedure leads to a graph with two unoriented links at each vertex, which therefore consists of unoriented loops. For example, $\mathfrak{m}=\left\{12\right\}\left\{36\right\}\left\{48\right\}\left\{57\right\}$ (see Figure~\ref{fig:ortholoop}(a)) corresponds to loops $(1)(234)$. Each trivial loop contributes $1$ in the product, and the contribution corresponding to a nontrivial loop is very similar to the unitary case with a permutation with the same cycle, with the only difference being that some states are transformed by complex conjugation, following Eqn.'s~\eqref{eq:psi def} and~\eqref{eq:Delta m string}.   For example, $\mathfrak{m}=\left\{12\right\}\left\{36\right\}\left\{48\right\}\left\{57\right\}$  corresponds to
\begin{align}
    \left|{\rm tr}\left(\Delta_{\mathfrak{m}}B_s\right)\right|=\left|\braket*{y^2_{s_0s_1s_2}}{y^3_{s_0...s_3}}\braket*{y^{2*}_{s_0s_1s_2}}{y^4_{s_0...s_4}}\braket*{y^3_{s_0...s_3}}{y^{4*}_{s_0...s_4}}\right|\prod_{i=1}^4\lambda^i_{s_0...s_4}
\end{align}
where we have denoted $\ket{y^{i*}_{s_0...s_i}}$ for $\left(\ket{y^i_{s_0...s_i}}\right)^*$ and similarly for the bras, for simplicity. Then we can apply the same loop decomposition procedure as in the proof of Lemma \ref{thm:tvd}. The only difference is the states in the matrix $M_{ji}$ in Eqn.~\eqref{eq:Mji} get complex conjugated in some cases. Since the complex conjugation does not change the orthonormality condition, the proof goes through, which gives the following inequality:
\begin{align}
    \sum_s\left|{\rm tr}\left(\Delta_{\mathfrak{m}}B_s\right)\right|\leq D^{k-\left\lfloor \frac{L_{\mathfrak{m}}}2\right\rfloor}
\end{align}
with $L_{\mathfrak{m}}$ the total length of nontrivial loops obtained in the procedure above. Carrying out the sum over $\mathfrak{m}$ we obtain
\begin{align}
    T^O=\frac{1}{D^k}\sum_{\mathfrak{m}\neq\mathfrak{e}}\sum_s\left|{\rm tr}\left(\Delta_{\mathfrak{m}}B_s\right)\right|\leq \frac{1}{D^k}\sum_{\mathfrak{m}\neq\mathfrak{e}}D^{k-\left\lfloor \frac{L_{\mathfrak{m}}}2\right\rfloor}\,.
\end{align}
The number of pairings of $2k$ items where $k-L$ of the pairings are trivial (i.e., a pairing of the form $\{2i-1,2i\}$) is given by
\begin{equation}
N^{\text{pair}}(2k,L) = \binom{k}{L} (2L-2)!! = \binom{k}{L} 2^{L-1}(L-1)!
\end{equation}
Thus
\begin{align}
    T^O\leq \frac{1}{D^k}\sum_{L=1}^k N^{\rm pair}(2k,L)D^{k-\left\lfloor \frac{L}2\right\rfloor}\leq \frac12\frac{(1+2k)\frac{4k^2}{D}}{1-\frac{4k^2}D}=\frac{4k^3}{D}+\frac{2k^2}{D}+\mathcal{O}\left(\frac{k^5}{D^2}\right)\,.
\end{align}

\item The symplectic case is very similar to the orthogonal case. We define $2k$ states
\begin{align}
    \ket{\psi_{2i-1}}=\ket{y^i_{s_0s_1...s_i}},~\ket{\psi_{2i}}=\textbf{J}\ket{y^i_{s_0s_1...s_i}}^*\,.
\end{align}
The trace becomes
\begin{equation}
|\text{tr}(\Delta_{\mathfrak{m}}'\textbf{J}^{T\otimes k} B)| = \prod_{i=1}^k\left|\bra{\psi_{\mathfrak{m}(2i-1)}}\left(\textbf{J}^T\ket{\psi_{\mathfrak{m}(2i)}}^*\right)\right|\lambda^i_{s_0s_1...s_i}\,.
\end{equation}
Following the same procedure as the orthogonal case, we can map each $\mathfrak{m}$ to a graph with unoriented loops. The remainder of the proof is identical, leading to the same upper bound
\begin{align}
    T^{Sp}\leq \frac{4k^3}{D}+\frac{2k^2}{D}+\mathcal{O}\left(\frac{k^5}{D^2}\right)
\end{align}
\end{enumerate}
$$$$
Putting all of the bounds on the $c_1$'s, $c_2$'s, and $T$'s together, we find
\begin{align}
   \|P_k-Q_k^U\|_1 &=\sum_{s}\left|P_k\left(s\right)-Q_k^U\left(s\right)\right|\leq \mathcal{O}(k^3/D) \\
    \|P_k-Q_k^O\|_1 &=\sum_{s}\left|P_k\left(s\right)-Q_k^O\left(s\right)\right|\leq \mathcal{O}(k^3/D)\\
    \|P_k-Q_k^{\text{Sp}}\|_1  &=\sum_{s}\left|P_k\left(s\right)-Q_k^{\text{Sp}}\left(s\right)\right|\leq \mathcal{O}(k^3/D) \,.
\end{align}
These are the desired bounds.
\end{proof}

Finally, using Claim~\ref{cl:derandomization}, we can upgrade the simple measurement QUALM setting of Lemma~\ref{lemm:symm} to prove the incoherent access QUALM setting of Theorem~\ref{thm:main2}.

\appendix

\section{Diagrams of quantum circuits and quantum channels}
\label{App:diagrams}
Here we explain our diagrammatic notation for quantum circuits and quantum channels used in the figures in the main text.  First, we consider standard circuit notation for a quantum state $|\psi\rangle$, with two unitaries $U_1, U_2$ acting on it:
\begin{align}
    \includegraphics[scale=.35, valign = c]{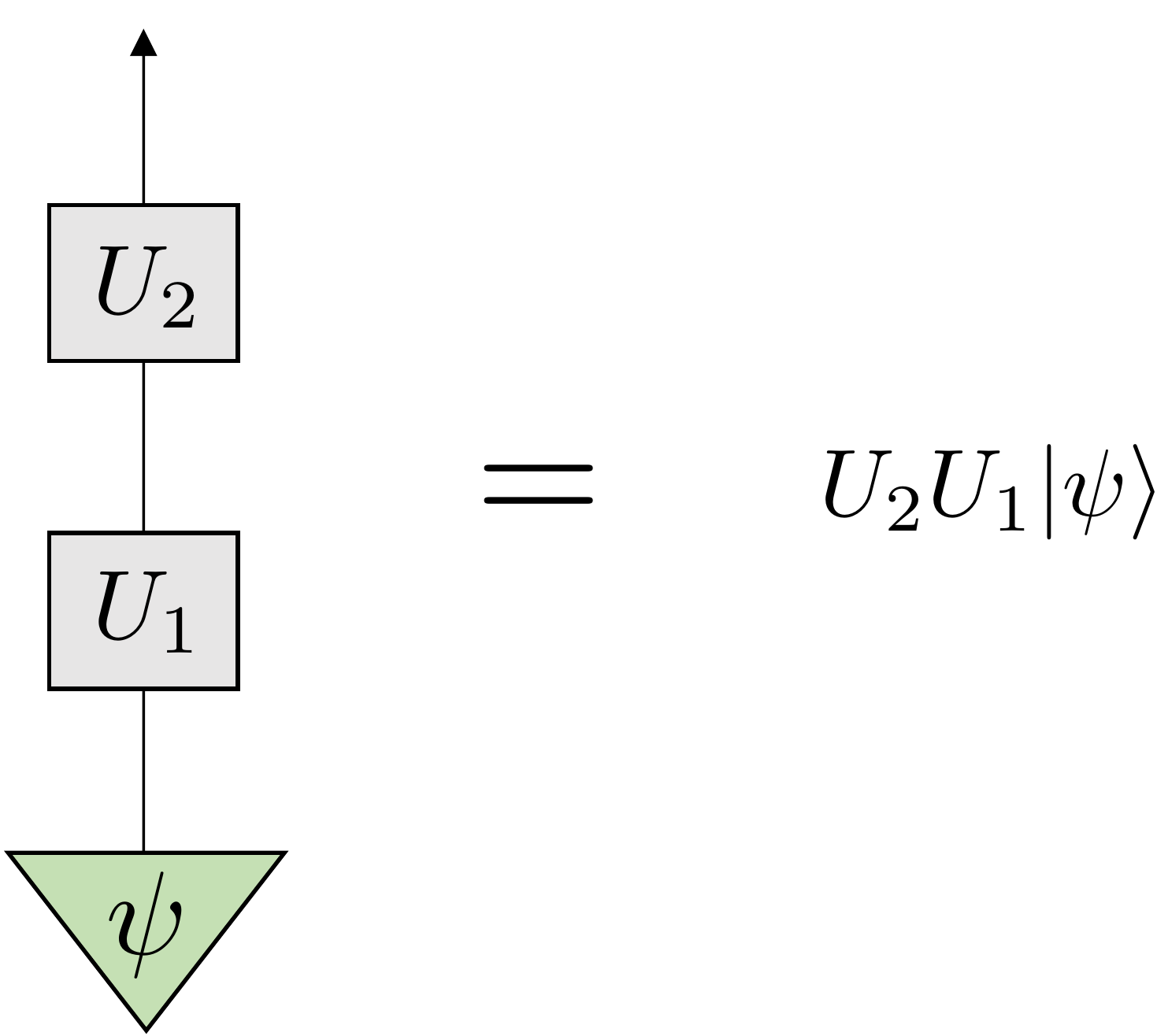}
\end{align}
This generalizes in the obvious way to any number of applied unitaries $U_1,U_2,U_3,...$\,.  Notice that the wire in the diagram is directed, as indicated by the arrow at the top.  The arrow designates the direction of matrix multiplication.  For instance, a matrix $M_{\alpha}^\beta$ is represented by a box with one in-line and one out-line, where the in-line corresponds to the $\alpha$ index and the out-line corresponds to the $\beta$ index.

Suppose that we want to consider the evolution of a density matrix $\rho$ instead of a pure state $|\psi\rangle$.  A diagrammatic representation analogous to the above is
\begin{align}
    \includegraphics[scale=.35, valign = c]{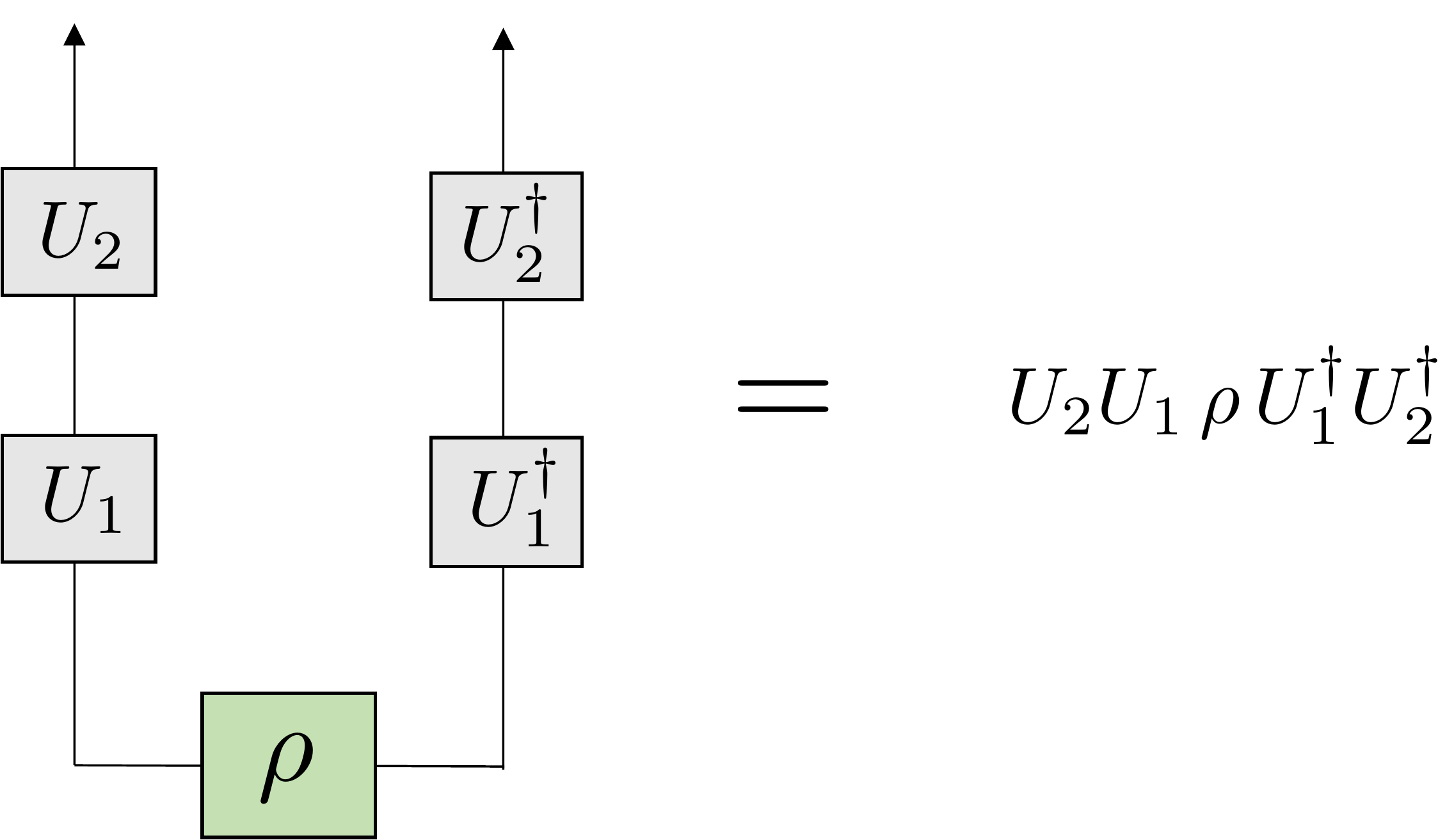}
\end{align}
\vskip.2cm
\noindent where $\rho$ has two wires since it is a matrix (which by definition has two indices).  We can likewise generalize the above to any sequence of applied unitaries.

It will be useful to simplify our diagrammatic notation slightly.  Consider the quantum channel $\mathcal{U}$ which acts on a density matrix $\rho$ by $\mathcal{U}[\rho] = U \rho U^\dagger$ for some unitary $U$.  We use the notation $\mathcal{U}_k[\rho] = U_k \rho U_k^\dagger$.   Since these channels maps matrices to matrices, they have two in-wires and two out-wires.  Thus we can depict $\mathcal{U}_2 \circ \mathcal{U}_1[\rho] = U_2 U_1 \rho U_1^\dagger U_2^\dagger$ by the diagram
\begin{align}
    \includegraphics[scale=.35, valign = c]{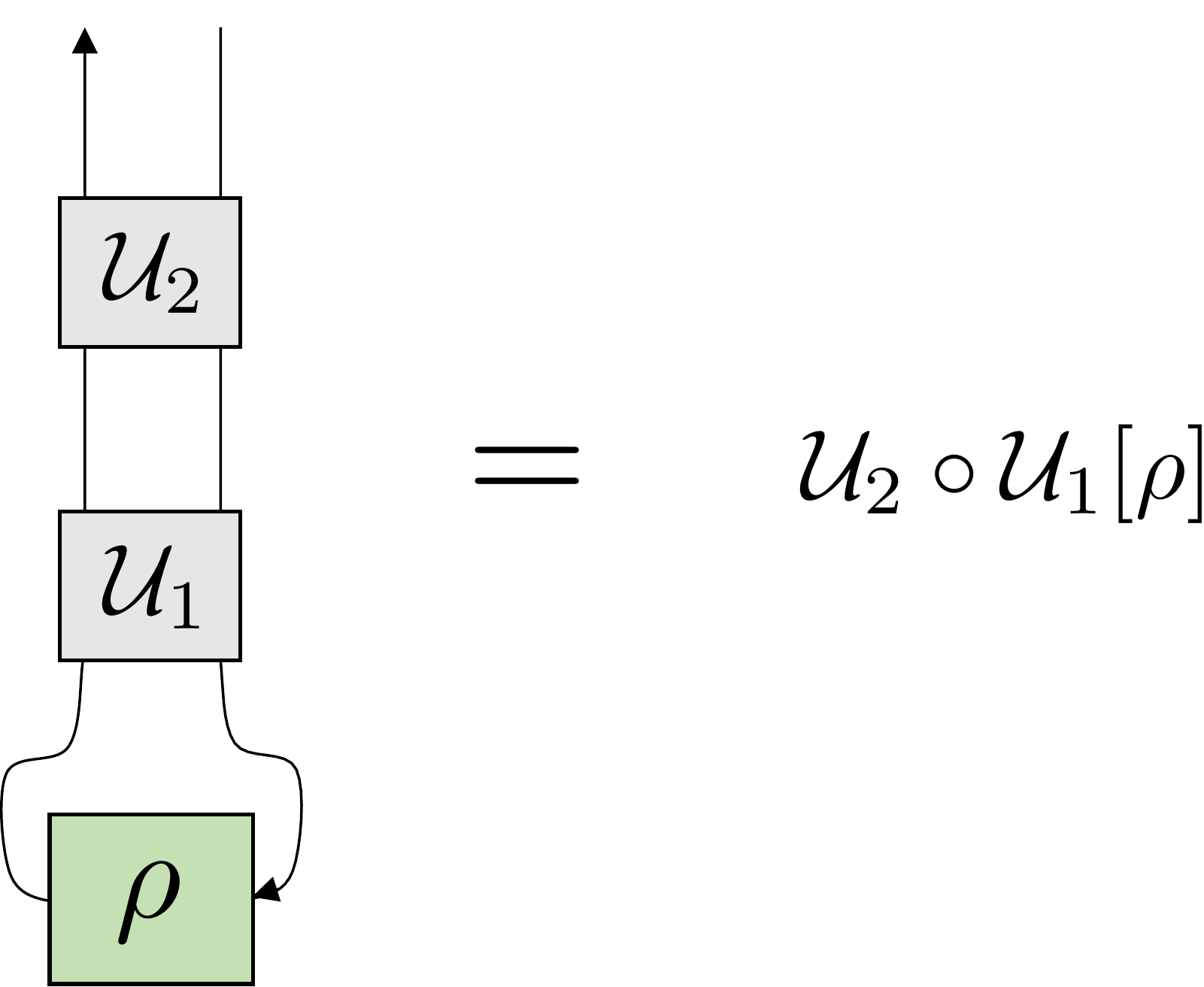}
\end{align}
\vskip.2cm
\noindent We make a further simplification by denoting pairs of wires (one on the left oriented upwards, and one on the right oriented downwards) by a single thick wire:
 \begin{align}
    \includegraphics[scale=.35, valign = c]{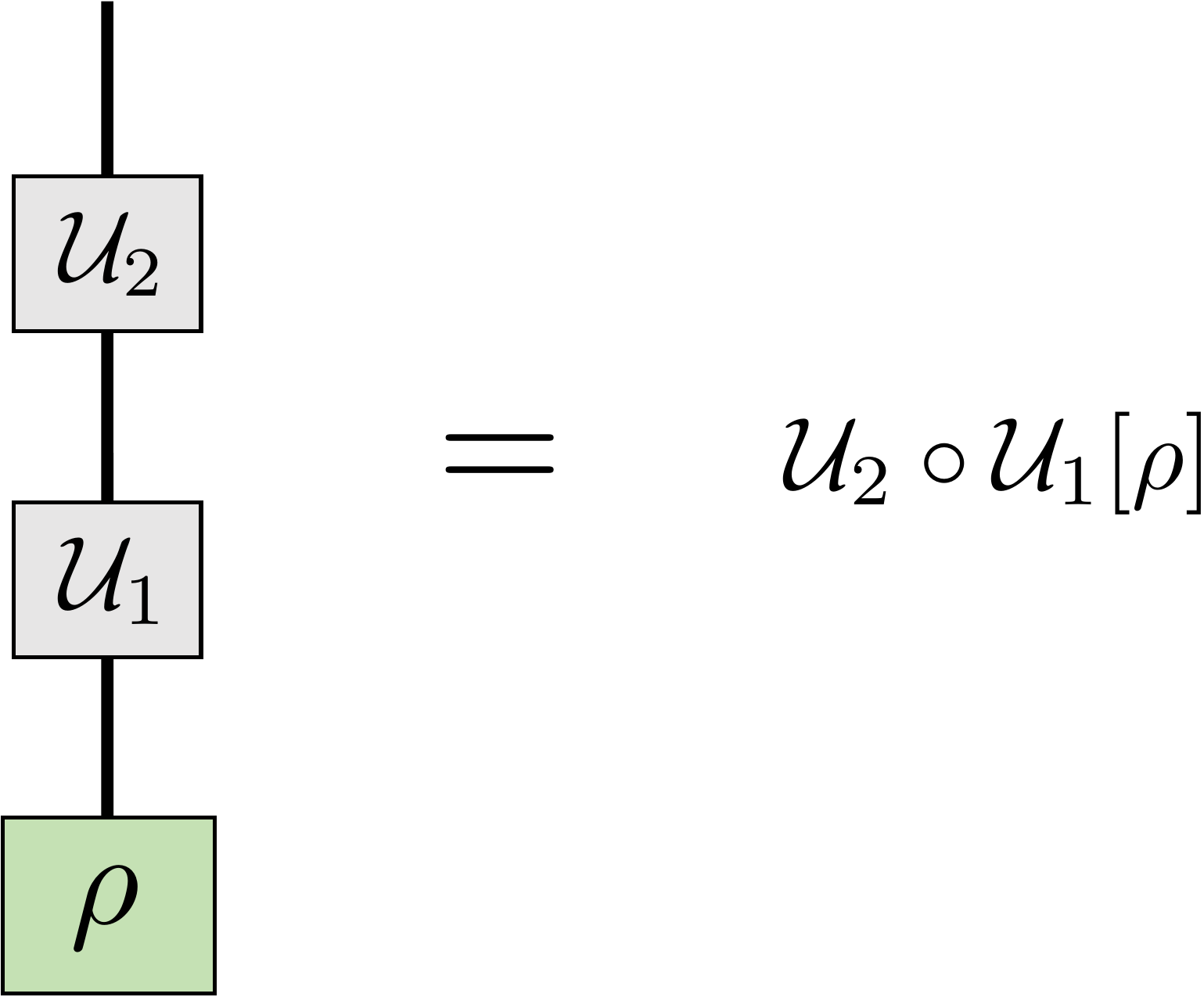}
\end{align}
\vskip.2cm
\noindent This generalizes to arbitrary superoperators $\mathcal{S}_1, \mathcal{S}_2,...$ as
\begin{align}
    \includegraphics[scale=.35,valign = c]{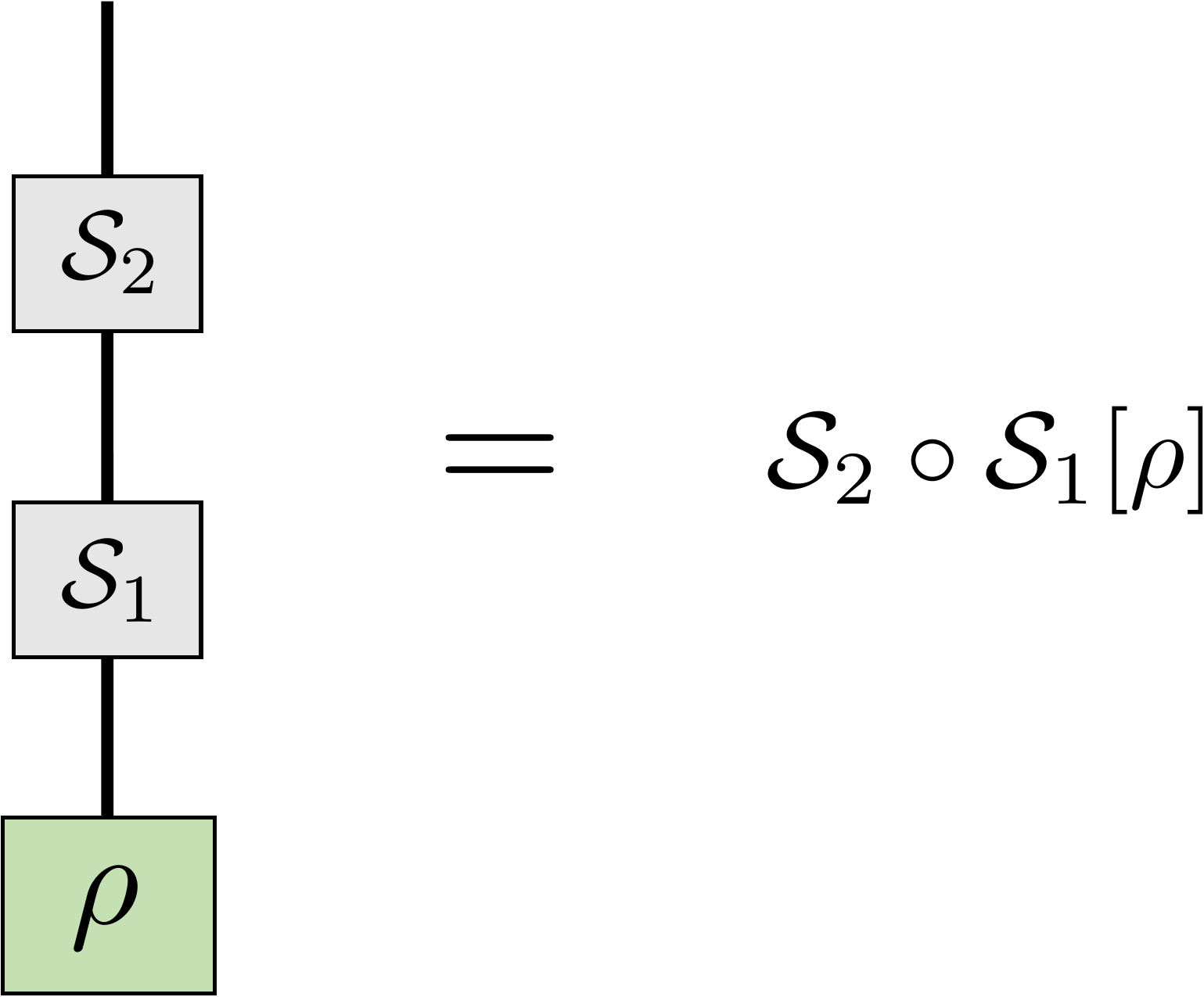}
\end{align}
\vskip.2cm
\noindent which is the primary diagrammatic notation used in the present paper.

There is an obvious generalization of the above ``thick line'' diagrams when we have multiple subsystems of our Hilbert space.  For instance, suppose $\fontH{H} \simeq \fontH{A} \otimes \fontH{B}$, and we have two quantum channels $\mathcal{C}_1, \mathcal{C}_2$ on $\fontH{H}$.  If we act these channels on $\rho_{\fontH{A}} \otimes \rho_{\fontH{B}}$, then this is depicted diagrammatically as
\begin{align}
    \includegraphics[scale=.35, valign = c]{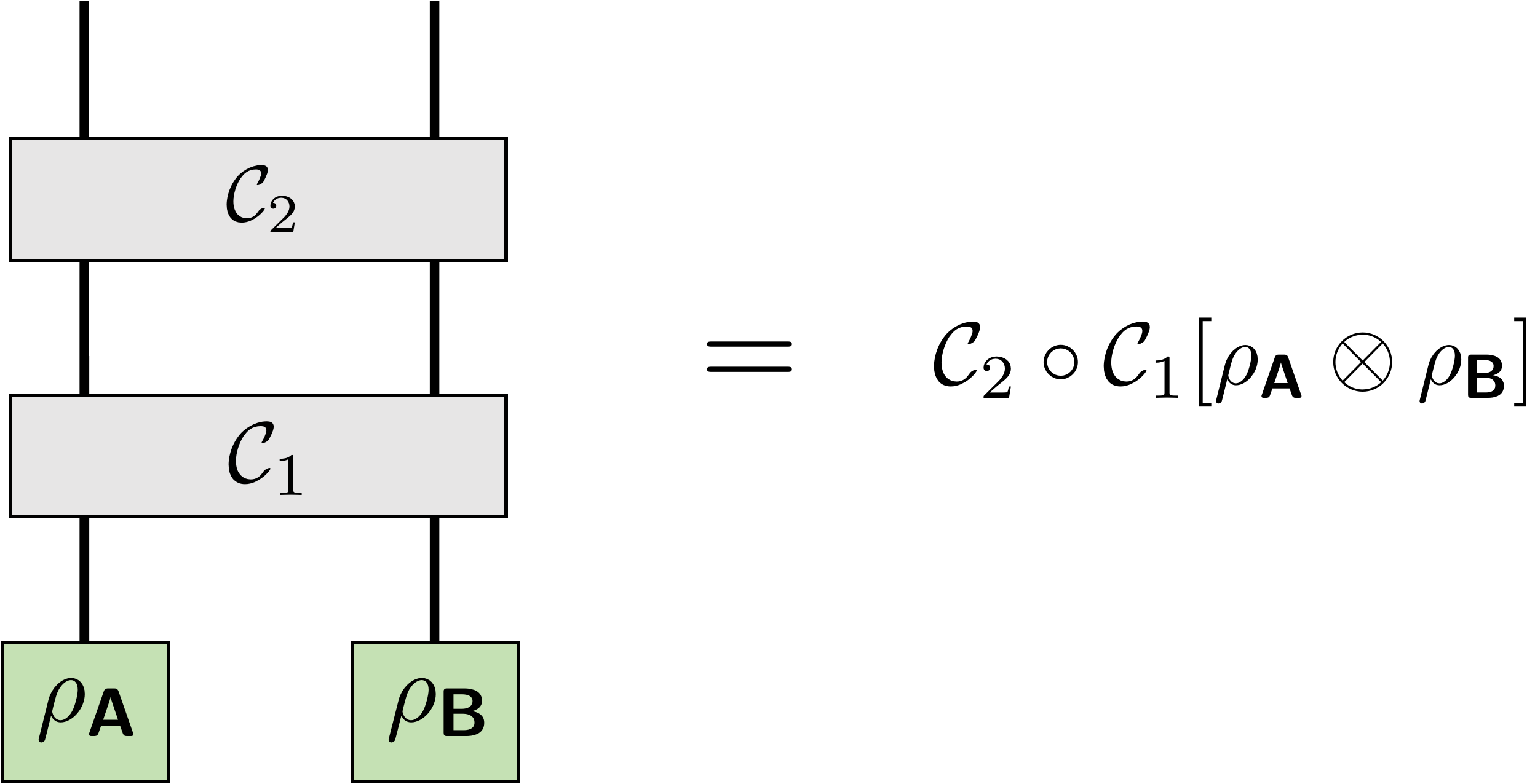}
\end{align}
\vskip.2cm
\noindent where evidently each thick line corresponds to a subsystem.  This of course works for $N$ subsystems, giving us as many thick lines.

We will use an additional piece of diagrammatic notation, which is often useful.  Again suppose $\fontH{H} \simeq \fontH{A} \otimes \fontH{B}$.  Say we have a collection of superoperators $\{\mathcal{S}_\alpha\}_\alpha$ acting on $\fontH{A}$, and a collection of superoperators $\{\mathcal{S}_\alpha'\}_\alpha$ acting on $\fontH{B}$.  Further suppose that the $\alpha$ index set is the same in each case, so that it is sensible to construct the channel
\begin{equation}
\label{E:alphasum1}
\sum_\alpha \mathcal{S}_{\alpha} \otimes \mathcal{S}_\alpha'
\end{equation}
acting on states $\rho_{\fontH{AB}}$ on $\fontH{A}\otimes \fontH{B}$.  We depict \eqref{E:alphasum1} acting on $\rho_{\fontH{AB}}$ by the diagram
\begin{align}
    \includegraphics[scale=.35, valign = c]{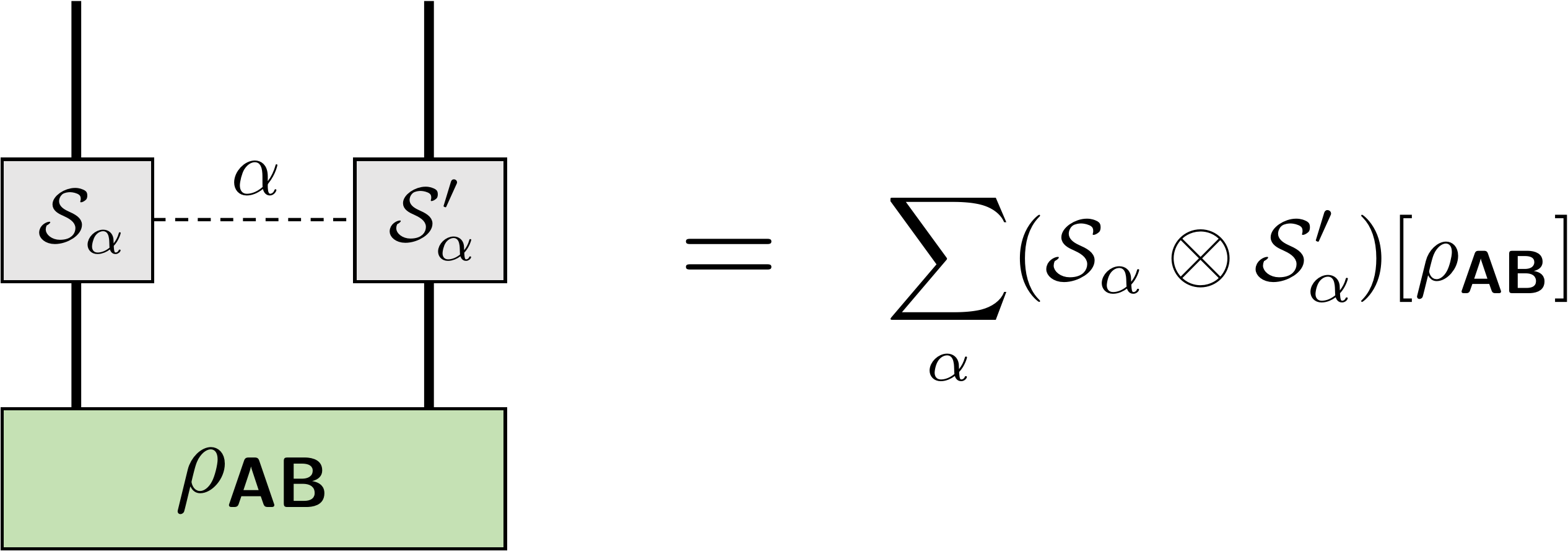}
\end{align}
\vskip.2cm
\noindent where the dotted line with the $\alpha$ above denotes a sum over $\alpha$.

\section{QUALM example: verification of quantum computation} 
\label{App:verification}
Here we give another example of how to present a familiar experimental task in the QUALM framework. As this example is more advanced and requires some non-negligible background, we delayed it until now. An underlying assumption of quantum computing is that general quantum computations require exponential time to simulate on a classical computer. This has a very interesting implication, namely that given a quantum system claimed to perform a certain quantum computational task, it is not clear how to verify that the output is actually correct. Sometimes one can check the end result easily (e.g., in Shor's factoring algorithm, checking that the integer factors are correct can be achieved by multiplying those factors); however, for the output of a general quantum algorithm, this is not the case.  The question is how to perform such a verification, despite the fact that simulating the system under investigation is assumed to take exponential time, and the correctness of the output cannot be directly checked (without having at hand a quantum computer known to work properly).   

In the works of~\cite{ABE1, BFK09, FitzsimonsKashefi, QPIP} a protocol to perform such a verification was provided, by an {\it almost} classical verifier. That is, the verifier has access to polynomial classical computing power, and also to a very small quantum computer composed of one or at most a few qubits which are fully controlled (i.e., whose behavior is fully understood). 

Let us recall in more detail the modus operandi of such an {\it almost} classical verification protocol \cite{QPIP}.  A `verifier', who has the ability to perform quantum computation on only $\mathcal{O}(1)$ qubits (and otherwise can only perform classical computation), is interacting with an agent called the `prover' which may have a quantum computer.  The verifier would like the prover to evaluate a chosen quantum circuit $U$ on a chosen computational basis state $|\psi_0\rangle$, and provide him the correct value of a specific output bit $b$; but the verifier does not trust the prover, and so he wants to be sure that the prover is not cheating. The verifier may interact with the prover via an elaborate protocol of communication. We then require that a successful verification protocol satisfies the following criteria:
\begin{itemize} 
\item If the prover does indeed have a quantum computer and is honest, then the output of the verification protocol is the correct output bit $b$ with probability $\geq 1 - \varepsilon$, for $\varepsilon$ small. 
\item For any prover, the probability of the verification protocol outputting the \textit{incorrect} bit $\bar{b}$ is $< \varepsilon$ (meaning that with probability $\geq 1 - \varepsilon$ the output is either the correct bit $b$, or $\textsf{REJECT}$). 
\end{itemize}

Now we formalize the above into a quantum task.  The idea is to have a lab oracle $\LO_{\text{honestQC}}$ which formalizes the action of a quantum computer. We think of a quantum computer as a hybrid quantum-classical entity: it applies quantum gates on its quantum register depending on classical inputs which it is given, which classically describe the quantum gates it needs to apply.  

To this end, we let $|\fontH{N}| = n$; this will consist of the quantum register of the quantum computer. We let the lab subsystem be $\fontH{L} \simeq \fontH{L}_1 \otimes \fontH{L}_2$ where $|\fontH{L}_2| = \mathcal{O}(1)$ and $|\fontH{L}_1| = \mathcal{O}(\text{poly}(n))$. $\fontH{L}_2$ will correspond to the few quantum qubits which the verifier may control, and $\fontH{L}_1$ will correspond to the register which serves as a classical communication channel between the verifier and the prover. 
We also set $|\fontH{W}| = \mathcal{O}(\text{poly}(n))$ -- this is the classical working register of the verifier. 

We can now define a lab oracle $\LO_{\text{honestQC}} = (\mathcal{E}_{\fontH{NL},\,\text{honest}}\,, |0\rangle \langle 0|^{\otimes n})$.  This lab oracle will play the role of an honest prover with a quantum computer.  The register $\fontH{L}_1$ will store the instructions that the verifier will send to the prover about the quantum gates to be implemented; the honest prover lab oracle superoperator examines the $\fontH{L}_1$ register and \textit{correctly and honestly} modifies the set of qubits in $\fontH{N} \otimes \fontH{L}_2$ according to the instructions.

The corresponding quantum task is as follows.  Let $\fontH{S}_{\text{in}}$ be a $k = \mathcal{O}(\text{poly}(n))$ subsystem of $\fontH{W}$. The input string in $\fontH{S}_{\text{in}}$ encodes a classical description of a polynomial size quantum circuit $U$, as well as an input $n$ qubit computational basis state $|\psi_0\rangle$ on which $U$ is to act.  Further, let $\fontH{S}_{\text{out}}$ be a two qubit subset of $\fontH{W}$; at the end of the protocol, this register will output one of three classical outputs: 
$0,1$ and $\textsf{REJECT}$.

We restrict the admissible set of gates $\mathcal{G}$ so that only a few types of operations are allowed: 
first, we allow a universal quantum gate set on the ``small quantum register'' $\fontH{L}_1$, where each of these gates is conditioned {\it classically} on one of the bits in $\fontH{W}$.  
We also allow measurements in the computational basis on $\fontH{L}_1$, 
whose classical output bit is written on a bit in $\fontH{W}$ (this can be suitably expressed as a superoperator). Finally, we allow a universal \textit{classical} gate set on $\fontH{L}_2 \otimes \fontH{W}$.

If $x_{|\psi_0\rangle, U}$ is a $k$-bit classical description of $|\psi_0\rangle$ and $U$, and $b = b(|\psi_0\rangle, U)$ is the desired correct output bit, then $f$ is a probabilistic function (see Remark~\ref{remark:altdef1}) defined by:
\begin{equation}
\label{E:verifierf1}
f : \mathscr{LO}(\fontH{N}, \fontH{L}) \times \{0,1\}^k \longrightarrow \mathcal{D}(\{0,1,\textsf{REJECT}\})
\end{equation}
where
\begin{equation}
f(\LO, x_{|\psi_0\rangle, U}) = \begin{cases}
b(|\psi_0\rangle, U) & \text{with probability }\geq 1 - \varepsilon\text{ if }\LO = \LO_{\text{honestQC}} \\
\bar{b} = 1 - b(|\psi_0\rangle, U) & \text{with probability }< \varepsilon\text{ for any }\LO.
\end{cases}\,.
\end{equation} 

Interestingly, this task can be achieved by a QUALM of polynomial QUALM complexity; this is the essence of the verification protocols of  \cite{ABE1, BFK09, FitzsimonsKashefi, QPIP}. It is a major open question whether this is possible if we require that all of $\fontH{L}$ is classical.

\section{Review of unitary, orthogonal, and symplectic matrix integrals}
\label{App:reviewHaar}

Here we review various facts on Haar unitary, orthogonal, and symplectic matrix integrals, as well as prove some useful lemmas.  Our review will be largely based on~\cite{Matsumoto1, CollinsMatsumoto1, Gu1}.

\subsection{Unitary matrix integrals}\label{sec:IntroWeingerten}

Here we review some facts about Haar unitary integrals, and the associated Weingarten functions.  Consider the integral
\begin{equation}
\label{E:mainhaar1}
\int_{U(D)} dU \, U_{i_1,j_1} U_{i_2,j_2} \cdots U_{i_k, j_k} \, \overline{U_{i_1', j_1'} U_{i_2', j_2'} \cdots U_{i_k', j_k'}} = \sum_{\sigma, \tau \in S_k} \delta_{\sigma(I),I'} \,\delta_{\tau(J),J'} \, \text{Wg}^U(\sigma \tau^{-1}, D)
\end{equation}
where $I= (i_1,...,i_k)$ and similarly for $I',J,J'$.  Above, the bars mean complex conjugation, so that $\overline{U_{ij}} = U_{ji}^\dagger$\,. We have additionally used the multi-index notation
\begin{equation}
\delta_{\sigma(I),I'} := \delta_{i_{\sigma(1)}, i_1'} \,\delta_{i_{\sigma(2)}, i_2'} \,\cdots\, \delta_{i_{\sigma(k)}, i_k'}\,.
\end{equation}
To compress notation further, let $U_{IJ}^{\otimes k} := U_{i_1,j_1} U_{i_2,j_2} \cdots U_{i_k, j_k}$.  Then we have, more compactly,
\begin{equation}
\label{E:Uweingartenidentity1}
\int_{U(D)} dU \, U_{IJ}^{\otimes k} U_{KL}^{\dagger \, \otimes k} = \sum_{\sigma, \tau \in S_k} \delta_{\sigma(I),L} \,\delta_{\tau(J),K} \, \text{Wg}^U(\sigma \tau^{-1}, D)\,,
\end{equation}
and accordingly
\begin{align}
\sum_{I,J,K,L}\int_{U(D)} dU \, U_{IJ}^{\otimes k} A_{JK} \, U_{KL}^{\dagger \, \otimes k}  B_{LI} &= \sum_{I,J}\sum_{\sigma, \tau \in S_k} A_{J, \tau(J)} \,B_{\sigma(I),I} \, \text{Wg}^U(\sigma \tau^{-1}, D) \\
&= \sum_{\sigma, \tau \in S_k} \text{tr}(\tau^{-1} A) \, \text{tr}(\sigma B) \,  \text{Wg}^U(\sigma \tau^{-1}, D)
\end{align}
where $A,B \in \mathcal{B}(\fontH{H}^{\otimes k})$.  Note that $\text{Wg}^{U}(\sigma,D)$ only depends on the conjugacy class of $\sigma$, i.e.~$\text{Wg}^{U}(\tau \sigma \tau^{-1},D) = \text{Wg}^{U}(\sigma, D)$ for any $\tau \in S_k$.  As such, $\text{Wg}^{U}(\sigma,D) = \text{Wg}^{U}(\sigma^{-1},D)$ and $\text{Wg}^{U}(\sigma \tau^{-1},D) = \text{Wg}^{U}( \tau^{-1} \sigma, D)$, for instance.

Two useful special cases of Eqn.~\eqref{E:Uweingartenidentity1} are
\begin{align}
\label{E:Haarmoment1}
&\int_{U(D)} dU \, U_{i_1 j_2} U^\dagger_{j_2 i_2} = \frac{1}{D} \, \delta_{i_1, i_2} \delta_{j_1, j_2} \\
\label{E:Haarmoment2}
&\int_{U(D)} dU\, U_{i_1 j_1} U_{i_2 j_2} U_{j_3 i_3}^\dagger U^\dagger_{j_4 i_4} = \frac{1}{D^2 - 1}\left(\delta_{i_1, i_3} \delta_{i_2, i_4} \delta_{j_1, j_3} \delta_{j_2, j_4} + \delta_{i_1, i_4} \delta_{i_2, i_3} \delta_{j_1, j_4} \delta_{j_2, j_3}\right)  \\
& \qquad \qquad \qquad \qquad \qquad \qquad \qquad - \frac{1}{D(D^2 - 1)} \left(\delta_{i_1, i_3} \delta_{i_2, i_4} \delta_{j_1, j_4} \delta_{j_2, j_3} + \delta_{i_1, i_4} \delta_{i_2, i_3} \delta_{j_1, j_3} \delta_{j_2, j_4}\right)\,. \nonumber
\end{align}

The Weingarten function $\text{Wg}^U$ can be expressed as an inverse of a more simple matrix.  Let
\begin{equation}
G^U(\sigma \tau^{-1}, D) := \text{tr}(\sigma \tau^{-1}) = D^{\#(\sigma \tau^{-1})}
\end{equation}
where $\#(\sigma \tau^{-1})$ is the number of cycles of the permutation $\sigma \tau^{-1}$.  We often also write
\begin{equation}
k - |\sigma \tau^{-1}| = \#(\sigma \tau^{-1})\,,
\end{equation}
where $|\sigma \tau^{-1}|$ can be regarded as a `distance' from $\sigma \tau^{-1}$ to the identity permutation.  We have the identity
\begin{equation}
\sum_{\tau} \text{Wg}^U(\sigma^{-1} \tau, D) \,G^U(\tau^{-1} \pi,D) = \delta_{\sigma, \pi}
\end{equation}
and so $G^U$ and $\text{Wg}^U$ are inverses as $k! \times k!$ matrices.  (Note that they are symmetric matrices.)  Perhaps a more illuminating notation is $\text{Wg}_D^U(\sigma, \tau) := \text{Wg}^U(\sigma^{-1} \tau, D)$ and $G_D^U(\tau, \pi) := G_D^U(\tau^{-1} \pi, D)$ so that
\begin{equation}
\sum_{\tau} \text{Wg}_D^U(\sigma, \tau) \,G_D^U(\tau, \pi) = \delta_{\sigma, \pi}\,.
\end{equation}
We also recall the useful bound:
\begin{CM3p2}\label{CM3p2}
For any $\sigma \in S_k$ and $D > \sqrt{6} k^{7/4}$,
\begin{align}
\label{eq:bound1}
\frac{1}{1 - \frac{k-1}{D^2}} \leq \frac{(-1)^{|\sigma|} D^{k + |\sigma|} \text{Wg}^U(\sigma,D)}{\prod_i \frac{(2 \ell_i - 2)!}{(\ell_i-1)! \ell_i!}} \leq \frac{1}{1 - \frac{6 k^{7/2}}{D^2}}
\end{align}
where the left-hand side inequality is valid for any $D \geq k$.  Note that $\sigma \in S_k$ has cycle type $(\ell_1,\ell_2,...)$.
\end{CM3p2}

Next, we present two Lemmas, which we have made use of in the text:

\begin{lemma}
\label{lemm1AppB1}
$|\text{Wg}^U(\mathds{1},D) - D^{-k}| \leq O(k^{7/2} D^{-(k+2)})$.
\end{lemma}
\begin{proof}
Follows directly from Eqn.~\eqref{eq:bound1}.
\end{proof}

\begin{lemma}
\label{lemm2AppB1}
$\sum_{\tau \in S_k} |\text{Wg}^U(\tau,D)| = \frac{(D-k)!}{D!}$\,.
\end{lemma}
\begin{proof}
We have $|\text{Wg}^U(\tau,D)|  = (-1)^{|\tau|} \text{Wg}^U(\tau,D)$, and so
\begin{equation}
\sum_{\tau \in S_k} |\text{Wg}^U(\tau,D)|  = \sum_{\tau \in S_k} (-1)^{|\tau|} \text{Wg}^U(\tau,D)\,.
\end{equation}
Consider the matrix
\begin{equation}
F_{\sigma \tau} = |\text{Wg}_D^U(\sigma, \tau)| = \text{Wg}^U(\sigma^{-1} \tau, D) (-1)^{|\sigma^{-1} \tau|} = \text{Wg}^U(\sigma^{-1} \tau, D) (-1)^{|\sigma|}(-1)^{|\tau|}\,.
\end{equation}
Letting $P_{\sigma \tau} = \delta_{\sigma \tau} (-1)^{|\sigma|}$, we have
\begin{equation}
F = P \text{Wg}_D^U P\,.
\end{equation}
Let $\vec{1}$ be the vector of all ones.  Since $F_{\sigma\tau}$ is only a function of $\sigma^{-1}\tau$, $\vec{1}$ is always an eigenvector of $F$:
\begin{equation}
F \cdot \vec{1} = \left(\sum_{\tau} F_{\tau}\right) \vec{1}\,.
\end{equation}
The eigenvalue $(\sum_{\tau} F_{\tau})$ is exactly the quantity we want to compute.  Note that $\vec{1}$ is likewise an eigenvector of $F^{-1} = P (\text{Wg}_D^U)^{-1} P = P G^U P$ with eigenvalue $(\sum_{\tau} F_{\tau})^{-1}$.  It follows that
\begin{equation}
\left(\sum_\tau F_\tau \right)^{-1} = \sum_\tau (-1)^{|\tau|} G^U(\tau, D) = D^k\sum_{\tau} (-D)^{-|\tau|} = \sum_{j=1}^k s(k,j) \, D^j = \frac{D!}{(D - k)!}
\end{equation}
where $s(k,j)$ are the Stirling numbers of the first kind, which count the number of permutations of length $k$ with $j$ cycles, times a factor of $(-1)^{k-j}$.  Thus we obtain
\begin{align}
    \sum_{\tau} |\text{Wg}^U(\tau,D)| = \frac{(D-k)!}{D!}
\end{align}
as desired.
\end{proof}

\subsection{Orthogonal matrix integrals}
\label{sec:IntroWeingerten2}

Before delving into the details of orthogonal matrix integrals, it is first important to understand pair partitions on $\{1,2,...,2k\}$, namely partitions into $k$ sets of size $2$.  The total number of such partitions is $\frac{(2k)!}{2^k k!} = (2k-1)!!$\,.  There are two useful ways to think about this: the first is the set of pair partitions $P_2(2k)$.  A pair partition $\mathfrak{m} \in P_2(2k)$ can be expressed as
\begin{equation}
\mathfrak{m} = \{\mathfrak{m}(1), \mathfrak{m}(2)\}\{\mathfrak{m}(3),\mathfrak{m}(4)\} \cdots \{\mathfrak{m}(2k-1),\mathfrak{m}(2k)\}
\end{equation}
where we pick the ordering convention $\mathfrak{m}(2i-1) < \mathfrak{m}(2i)$ for $1 \leq i \leq k$ and $\mathfrak{m}(1) <\mathfrak{m}(3) < \cdots < \mathfrak{m}(2k-1)$.  Alternatively, we can think of a pair partition as a member of the group $M_{2k}$ which we treat as a subgroup of $S_{2k}$ as follows.  We say that $\sigma \in M_{2k}$ if $\sigma(2i-1) < \sigma(2i)$ for $1 \leq i \leq k$ and $\sigma(1) <\sigma(3) < \cdots < \sigma(2k-1)$.  There is a natural group action of $M_{2k}$ on $P_{2}(2k)$ in the obvious way:
\begin{equation}
\sigma \cdot \mathfrak{m} = \{\sigma(\mathfrak{m}(1)), \sigma(\mathfrak{m}(2))\}\{\sigma(\mathfrak{m}(3)),\sigma(\mathfrak{m}(4))\} \cdots \{\sigma(\mathfrak{m}(2k-1)),\sigma(\mathfrak{m}(2k))\}\,.
\end{equation}
We can map between $P_2(2k)$ and $M_{2k}$ in a canonical manner, namely $\mathfrak{m} \mapsto \sigma_{\mathfrak{m}}$ where $\sigma_{\mathfrak{m}}(i) = \mathfrak{m}(i)$.

Now we explain what cycles and cycle types mean in the context of pairings.  These are often called the ``coset type'' of a pairing.  We begin by explaining coset types in terms of $P_2(2k)$, and then afterwards give a different (albeit equivalent) construction of coset types in terms of $M_{2k}$.  In the $P_2(2k)$ context, let $\mathfrak{m} \in P_2(2k)$ be a pairing.  We define the function
\begin{align}
f_{\mathfrak{m}}(i) = \begin{cases} \mathfrak{m}(2j) & \text{if }i = \mathfrak{m}(2j-1) \\
\mathfrak{m}(2j-1) & \text{if }i = \mathfrak{m}(2j)
\end{cases}\,.
\end{align}
In words, $f_{\mathfrak{m}}(i)$ maps $i$ to the integer it is paired with under $\mathfrak{m}$.  Let $\mathfrak{e}$ be the identity pairing $\{1,2\}\{3,4\} \cdots \{2k-1, 2k\}$.  Now we construct a sequence starting with $1$:
\begin{align}
(1, f_{\mathfrak{m}}(1), f_{\mathfrak{e}} \circ f_\mathfrak{m}(1),  f_{\mathfrak{m}} \circ f_{\mathfrak{e}} \circ f_\mathfrak{m}(1), f_{\mathfrak{e}} \circ f_{\mathfrak{m}} \circ f_{\mathfrak{e}} \circ f_\mathfrak{m}(1), ..., 1)
\end{align}
where the sequence terminates the first time we reach $1$ again.  If we treat the above $(\cdots)$ notation as defined cyclically, then we can omit the two copies of $1$.  Let us do this.  We call this new cyclic object
\begin{equation}
B_1 = (1, f_{\mathfrak{m}}(1), f_{\mathfrak{e}} \circ f_\mathfrak{m}(1),  f_{\mathfrak{m}} \circ f_{\mathfrak{e}} \circ f_\mathfrak{m}(1), f_{\mathfrak{e}} \circ f_{\mathfrak{m}} \circ f_{\mathfrak{e}} \circ f_\mathfrak{m}(1), ...)
\end{equation}
with length $b_1$.  If $B_1$ contains all of $\{1,2,...,2k\}$, then we are done.  Otherwise, we find the smallest $j$ in $\{1,2,...,2k\}$ that is not contained in $B_1$, and construct
\begin{equation}
B_2 = (j, f_{\mathfrak{m}}(j), f_{\mathfrak{e}} \circ f_\mathfrak{m}(j),  f_{\mathfrak{m}} \circ f_{\mathfrak{e}} \circ f_\mathfrak{m}(j), f_{\mathfrak{e}} \circ f_{\mathfrak{m}} \circ f_{\mathfrak{e}} \circ f_\mathfrak{m}(j), ...)
\end{equation}
which again terminates so that it is cyclic with the smallest possible period.  If $B_1$ and $B_2$ do not jointly contain all $\{1,2,...,2k\}$, then we proceed by constructing a $B_3$, and so on.  At the end of this procedure, we have a partition of $\{1,2,...,2k\}$ into $B_1 B_2 \cdots$ such that $b_1,b_2,...$ are all even.  Then the cycle type or coset type of $\mathfrak{m}$ in $M_{2k}$ is
\begin{equation}
(\mu_1, \mu_2, ...) = (b_1/2, b_2/2,...)\,.
\end{equation}
which we typically (re)order from largest to smallest. 

In terms of the $M_{2k}$, we can consider the following algorithm:  Let $\sigma \in M_{2k}$, and let $C_1 C_2 \cdots$ be its cycle decomposition (including cycles of length one) as an element of $S_{2k}$\,.  The cycles have lengths $\ell_1,\ell_2,...$ respectively.  To construct the cycle type of $\sigma$ as an element of $M_{2k}$, we consider all pairs of integers $2i-1,2i$.  If some such pair has the property that $2i-1$ and $2i$ belong to distinct cycles, then we ``combine'' those two cycles (simply by considering a set containing those two cycles).  We do this for all pairs $2i-1,2i$ when necessary (some pairs will already belong to the same cycle) and end with $M_{2k}$-cycles $\widetilde{C}_1 \widetilde{C}_2 \cdots$ where $\widetilde{C}_{i} = \{C_{i,1}, C_{i,2},...\}$ and by construction $\widetilde{C}_{i} \cap \widetilde{C}_{j} = \emptyset$ for $i \not = j$.  The lengths of these $M_{2k}$-cycles are $\widetilde{\ell}_1, \widetilde{\ell}_2,...$ where $\widetilde{\ell}_i := \ell_{i,1} + \ell_{i,2} + \cdots$.  Note that all the $\widetilde{\ell}_i$'s are by construction even, and so we say that the $M_{2k}$-cycle type of $\sigma$ is
\begin{equation}
(\mu_1, \mu_2, ...) = (\widetilde{\ell}_1/2, \widetilde{\ell}_2/2,...)
\end{equation}
which similar to above, we typically (re)order from largest to smallest.  Said more simply, we take the cycle type of $\sigma \in S_{2k}$ and ``connect'' cycles so that all pairs $2i-1,2i$ lie within the same ``connected'' cycle.  These connected cycles necessarily have even length, and so we can divide these lengths by two.  It can be checked that $\mathfrak{m}$ and its counterpart $\sigma_{\mathfrak{m}}$ have the same cycle types when we process each with its appropriate algorithm.

With our knowledge of pair partitions at hand, we can now delve into the Weingarten calculus for orthogonal matrix integrals.  Let $O_{IJ} = O_{i_1, j_1} O_{i_2, j_2} \cdots O_{i_{2k}, j_{2k}}$ where $I = (i_1,i_2,...,i_{2k})$ and similarly for $J$.  Note that in contrast to our notation before, $I$ and $J$ are multi-indices with length $2k$ (instead of $k$).  We will adjust this back to $k$ shortly.  At present, we have
\begin{equation}
\label{eq:orthintegral1}
\int_{O(D)} dO \, O_{IJ} = \sum_{\mathfrak{m}, \mathfrak{n} \in P_2(2k)} \Delta_{\mathfrak{m}}(I) \, \Delta_{\mathfrak{n}}(J) \, \text{Wg}_D^O(\mathfrak{m}, \mathfrak{n})
\end{equation}
where
\begin{equation}
\left[\Delta_{\mathfrak{m}}\right]_{i_1i_3...i_{2k-1}}^{i_2i_4...i_{2k}} = \prod_{s=1}^k \delta_{i_{\mathfrak{m}(2s-1)}, i_{\mathfrak{m}(2s)}}\,.
\end{equation}
Since there is a canonical map from $P_{2}(2k)$ to $M_{2k}$ taking $\mathfrak{m} \mapsto \sigma
_{\mathfrak{m}}$, we can also write
\begin{equation}
\text{Wg}_D^O(\mathfrak{m}, \mathfrak{n}) = \text{Wg}_D^O(\sigma_\mathfrak{m}, \sigma_\mathfrak{n}) = \text{Wg}^O(\sigma_\mathfrak{m}^{-1} \sigma_\mathfrak{n}, D)
\end{equation}
similar to the unitary case from before.

For our purposes, it will be useful to reorganize Eqn.~\eqref{eq:orthintegral1} slightly.  Let $I = (i_1,...,i_k)$, $J = (j_1,...,j_k)$, $I' = (i_{1}',...,i_{k}')$ and $J' = (j_{1}',...,j_{k}')$.  Note that now, each of $I,J,I',J'$ is a multi-index of length $k$.  We denote by $II'$ the combined multi-index $II' = (i_1,i_1'...,i_{k},i_k')$ and similarly for $JJ'$.  Then we have
\begin{equation}
\label{eq:orthintegral2}
\int_{O(D)} dO \, O_{IJ} \, O_{J'I'}^T = \sum_{\mathfrak{m}, \mathfrak{n} \in P_2(2k)} \Delta_{\mathfrak{m}}(II') \, \Delta_{\mathfrak{n}}(JJ') \, \text{Wg}_D^O(\mathfrak{m}, \mathfrak{n})
\end{equation}
and accordingly
\begin{equation}
\label{eq:orthintegral3}
\sum_{I,J,I',J'}\int_{O(D)} dO \, O_{IJ} \, A_{JJ'}\,O_{J'I'}^T \, B_{I'I} = \sum_{\mathfrak{m}, \mathfrak{n} \in P_2(2k)} \text{tr}(\Delta_{\mathfrak{n}}A) \, \text{tr}(\Delta_{\mathfrak{m}}B) \, \text{Wg}_D^O(\mathfrak{m}, \mathfrak{n})\,.
\end{equation}
Note that $\text{tr}(\Delta_{\mathfrak{n}}A)$ computes the pairwise contraction the $2k$ of indices of $A_{JJ'}$ corresponding to $\mathfrak{n}$, and $\text{tr}(\Delta_{\mathfrak{m}}B)$ is defined similarly.

Using Eqn.~\eqref{eq:orthintegral2}, the first two (nontrivial) moments of the orthogonal ensemble are
\begin{align}
\label{E:orthogmoment1}
& \int_{O(D)} dO \, O_{i_1 j_2} O^T_{j_2 i_2} = \frac{1}{D} \, \delta_{i_1 ,i_2} \delta_{j_1, j_2} \\
&\int_{O(D)} dO\, O_{i_1 j_1}O_{i_2 j_2}O_{j_3 i_3}^T O^T_{j_4 i_4} \label{E:orthogmoment2}\\
&\qquad = \frac{D+1}{D(D-1)(D+2)} 
\big( \delta_{i_1,i_2} \delta_{i_3,i_4} \delta_{j_1,j_2} \delta_{j_3,j_4} + \delta _{i_1,i_3} \delta_{i_2,i_4} \delta_{j_1,j_3} \delta_{j_2,j_4} + \delta_{i_1,i_4} \delta_{i_2,i_3} \delta_{j_1,j_4} \delta_{j_2,j_3} \big) \nonumber \\
&\qquad \qquad -\frac{1}{D(D-1)(D+2)} \big(\delta _{i_1,i_3} \delta _{i_2,i_4} \delta _{j_1,j_4} \delta_{j_2,j_3} + \delta _{i_1,i_2} \delta _{i_3,i_4} \delta _{j_1,j_4} \delta_{j_2,j_3} + \delta _{i_1,i_4} \delta _{i_2,i_3} \delta _{j_1,j_3} \delta_{j_2,j_4}\nonumber \\
&\qquad \qquad \qquad \qquad \qquad \qquad \,\,\,\, +\delta _{i_1,i_2} \delta _{i_3,i_4} \delta _{j_1,j_3} \delta_{j_2,j_4}+\delta _{i_1,i_4} \delta _{i_2,i_3} \delta _{j_1,j_2} \delta_{j_3,j_4}+\delta _{i_1,i_3} \delta _{i_2,i_4} \delta _{j_1,j_2} \delta_{j_3,j_4} \big)\,.\nonumber
\end{align}

Similar to the unitary case, we can express $\text{Wg}^O$ as the inverse of a simpler matrix.  Let
\begin{equation}
G_D^O(\sigma_{\mathfrak{m}}, \sigma_{\mathfrak{n}}) := D^{\#_O(\sigma_{\mathfrak{m}}^{-1} \sigma_{\mathfrak{n}})}
\end{equation}
where $\#_O(\sigma_{\mathfrak{m}}^{-1} \sigma_{\mathfrak{n}})$ is the number of $M_{2k}$-cycles of $\sigma_{\mathfrak{m}}^{-1} \sigma_{\mathfrak{n}}$.  Note that $G_D^O$ is a $(2k-1)!! \times (2k-1)!!$ matrix.  We often also write
\begin{equation}
k - |\sigma_{\mathfrak{m}}^{-1} \sigma_{\mathfrak{n}}|_O = \#_O(\sigma_{\mathfrak{m}}^{-1} \sigma_{\mathfrak{n}})\,,
\end{equation}
where $|\sigma_{\mathfrak{m}}^{-1} \sigma_{\mathfrak{n}}|_O$ can be regarded as a `distance' from $\sigma_{\mathfrak{m}}^{-1} \sigma_{\mathfrak{n}}$ to the identity pairing.  We have the identity
\begin{equation}
\sum_{\mathfrak{n} \in P_2(2k)} \text{Wg}_D^O(\sigma_{\mathfrak{m}}, \sigma_{\mathfrak{n}}) \,G_D^O(\sigma_{\mathfrak{n}}, \sigma_{\mathfrak{p}}) = \delta_{\mathfrak{m},\mathfrak{p}}\,,
\end{equation}
and so evidently $\text{Wg}^O$ and $G^O$ are inverses.

Similar to the unitary case, we have the following useful bound:
\begin{CM4p11}
For any $\sigma_{\mathfrak{m}} \in M_{2k}$ and $D > 12 k^{7/2}$,
\begin{align}
\label{eq:bound2}
\frac{1 - \frac{24 k^{7/2}}{D}}{1 - \frac{144 k^7}{D^2}} \leq \frac{(-1)^{|\sigma_{\mathfrak{m}}|_O} D^{k + |\sigma_{\mathfrak{m}}|_O} \text{Wg}^O(\sigma_{\mathfrak{m}},D)}{\prod_i \frac{(2 \mu_i - 2)!}{(\mu_i-1)! \mu_i!}} \leq \frac{1}{1 - \frac{144 k^{7}}{D^2}}
\end{align}
where $\sigma_{\mathfrak{m}}$ has $M_{2k}$-cycle type $(\mu_1,\mu_2,...)$. 
\end{CM4p11}
\noindent This bound is slightly weaker than Eqn.~\eqref{eq:bound1} in the unitary case, but it will suffice for our purposes.

We will now present two Lemmas analogous to the unitary case:
\begin{lemma}
\label{lemm1AppB2}
$|\text{Wg}^O(\sigma_\mathfrak{e},D) - D^{-k}| \leq O(k^{7} D^{-(k+2)})$. \end{lemma}
\begin{proof}
Follows directly from Eqn.~\eqref{eq:bound2}.
\end{proof}

\begin{lemma}
\label{lemm2AppB2}
$\sum_{\mathfrak{m} \in P_2(2k)} |\text{Wg}^O(\sigma_{\mathfrak{m}},D)| = \frac{(D-2k)!!}{D!!}$\,.
\end{lemma}
\begin{proof}
We have $|\text{Wg}^O(\tau_{\mathfrak{m}},D)|  = (-1)^{|\tau_{\mathfrak{m}}|_O} \text{Wg}^O(\tau_{\mathfrak{m}},D)$, and so
\begin{equation}
\sum_{\mathfrak{m} \in P_{2}(2k)}|\text{Wg}^O(\tau_{\mathfrak{m}},D)|  = \sum_{\mathfrak{m} \in P_{2}(2k)}(-1)^{|\tau_{\mathfrak{m}}|_O} \text{Wg}^O(\tau_{\mathfrak{m}},D)\,.
\end{equation}
Let us consider the matrix
\begin{equation}
F_{\mathfrak{m} \mathfrak{n}} = |\text{Wg}_D^O(\mathfrak{m}, \mathfrak{n})| = \text{Wg}^O(\sigma_{\mathfrak{m}}^{-1} \sigma_{\mathfrak{n}}, D) (-1)^{|\sigma_{\mathfrak{m}}^{-1} \sigma_{\mathfrak{n}}|_O} = \text{Wg}^O(\sigma_{\mathfrak{m}}^{-1} \sigma_{\mathfrak{n}}, D) (-1)^{|\sigma_{\mathfrak{m}}|_O}(-1)^{| \sigma_{\mathfrak{n}}|_O}\,.
\end{equation}
Defining $P_{\mathfrak{m} \mathfrak{n}} = \delta_{\mathfrak{m} \mathfrak{n}} (-1)^{|\sigma_{\mathfrak{m}}|_O}$, we automatically have the relation
\begin{equation}
F = P \text{Wg}_D^O P\,.
\end{equation}
As before, we let $\vec{1}$ be the vector of all ones.  Noting that $F_{\mathfrak{m} \mathfrak{n}}$ is only a function of $\sigma_{\mathfrak{m}}^{-1} \sigma_{\mathfrak{n}}$, the vector $\vec{1}$ is always an eigenvector of $F$:
\begin{equation}
F \cdot \vec{1} = \left(\sum_{\mathfrak{m}} F_{\sigma_{\mathfrak{m}}}\right) \vec{1}\,.
\end{equation}
We would like to compute the eigenvalue $(\sum_{\mathfrak{m}} F_{\sigma_{\mathfrak{m}}})$.  Since $\vec{1}$ is also an eigenvector of $F^{-1} = P (\text{Wg}_D^O)^{-1} P = P G^O P$ with eigenvalue $(\sum_{\mathfrak{m}} F_{\sigma_{\mathfrak{m}}})^{-1}$, we have
\begin{equation}
\left(\sum_{\mathfrak{m}} F_{\sigma_{\mathfrak{m}}} \right)^{-1} = \sum_{\mathfrak{m}}(-1)^{|\sigma_{\mathfrak{m}}|_O} G^O(\sigma_{\mathfrak{m}}, D) = \sum_{j=1}^k \widetilde{s}(k,j) \, D^j = \frac{D!!}{(D-2k)!!}
\end{equation}
where $\widetilde{s}(k,j)$ counts the number of $2k$-pairings with $j$ $M_{2k}$-cycles, times a factor of $(-1)^{k-j}$.  We note that $\widetilde{s}(k,j) = \frac{2^k}{2^j}\,s(k,j)$.  Finally, we find
\begin{align}
    \sum_{\mathfrak{m} \in P_2(2k)} |\text{Wg}^O(\sigma_{\mathfrak{m}},D)| = \frac{(D-2k)!!}{D!!}
\end{align}
which is the claimed result.
\end{proof}

\subsection{Symplectic matrix integrals}
\label{sec:IntroWeingerten3}

Finally, we treat the symplectic case.  This is rather similar to the orthogonal case, with one additional ingredient: the $D \times D$ matrix
\begin{equation}
\textbf{J} = \begin{bmatrix}
\textbf{0}_{D/2} & \mathds{1}_{D/2} \\ -\mathds{1}_{D/2} & \textbf{0}_{D/2}
\end{bmatrix}\,,
\end{equation}
where we note that $\textbf{J}^T = - \textbf{J}$ and we take $D$ to be even.  Here, $\textbf{0}_{D/2}$ is the $D/2 \times D/2$ matrix of all zeros.

To avoid extensive duplication of the previous sections, we cut to the bare essentials.  As before, let $I = (i_1,...,i_k)$, $J = (j_1,...,j_k)$, $I' = (i_{k}',...,i_{k}')$ and $J' = (j_{1}',...,j_{k}')$.  Then we have
\begin{equation}
\label{eq:symplecintegral1}
\int_{\text{Sp}(D/2)} dS \, S_{IJ} \, S_{J'I'}^T = \sum_{\mathfrak{m}, \mathfrak{n} \in P_2(2k)} \Delta_{\mathfrak{m}}'(II') \, \Delta_{\mathfrak{n}}'(JJ') \, \text{Wg}_{D/2}^{\text{Sp}}(\mathfrak{m}, \mathfrak{n})
\end{equation}
where
\begin{align}
\left[\Delta_{\mathfrak{m}}'\right]_{i_1i_3...i_{2k-1}}^{i_2i_4...i_{2k}} = \prod_{s=1}^k \textbf{J}_{i_{\mathfrak{m}(2s-1)}, i_{\mathfrak{m}(2s)}}\,.
\end{align}
Letting $\textbf{J}_{MN} = \textbf{J}_{m_1, n_1} \textbf{J}_{m_2, n_2} \cdots \textbf{J}_{m_k, n_k}$, we have
\begin{align}
\label{eq:symplecintegral2}
&\sum_{I,J,K,L,M,N}\int_{\text{Sp}(D/2)} dS \, S_{IJ} \, A_{JK}\,\textbf{J}_{KL} \,S_{LM}^T \, \textbf{J}_{MN}^T \, B_{NI}  \nonumber \\
& \qquad \qquad \qquad \qquad \qquad  = \sum_{\mathfrak{m}, \mathfrak{n} \in P_2(2k)} \text{tr}(\Delta_{\mathfrak{n}}'A \textbf{J}^{\otimes k}) \, \text{tr}(\Delta_{\mathfrak{m}}'\textbf{J}^{T\otimes k} \ B) \, \text{Wg}_{D/2}^\text{Sp}(\mathfrak{m}, \mathfrak{n})\,.
\end{align}
Here $\text{tr}(\Delta_{\mathfrak{n}}' A  \textbf{J}^{\otimes k})$ computes the pairwise contraction the $2k$ of indices of $A \textbf{J}^{\otimes k}$ corresponding to $\mathfrak{n}$, with an additional $\textbf{J}^{\otimes k}$ sandwiched in the middle.  The quantity $\text{tr}(\Delta_{\mathfrak{m}}'\textbf{J}^{T\otimes k} B)$ is defined similarly.

Using Eqn.~\eqref{eq:symplecintegral1} to compute the first two (nontrivial) moments of the symplectic ensemble, we find
\begin{align}
\label{E:sympmoment1}
&\int_{\text{Sp}(D/2)} dS \, S_{i_1 j_1} S_{j_2 i_2}^\dagger = \frac{1}{D} \, \delta_{i_1, i_2} \delta_{j_1, j_2} \\
&\int_{\text{Sp}(D/2)} dS\, S_{i_1 j_1} S_{i_2 j_2} S_{j_3 i_3}^\dagger S^\dagger_{j_4 i_4} \label{E:sympmoment2}\\
&\qquad = \frac{D-1}{D(D+1)(D-2)} 
\big( \textbf{J}_{i_1,i_2} \textbf{J}_{i_3,i_4} \textbf{J}_{j_1,j_2} \textbf{J}_{j_3,j_4} + \delta_{i_1,i_3} \delta_{i_2,i_4} \delta_{j_1,j_3} \delta_{j_2,j_4} + \delta_{i_1,i_4} \delta_{i_2,i_3} \delta_{j_1,j_4} \delta_{j_2,j_3} \big)\nonumber \\
&\qquad \qquad -\frac{1}{D(D+1)(D-2)} \big(\delta_{i_1,i_3} \delta_{i_2,i_4} \delta_{j_1,j_4} \delta_{j_2,j_3} - \textbf{J}_{i_1,i_2} \textbf{J}_{i_3,i_4} \delta_{j_1,j_4} \delta_{j_2,j_3} + \delta_{i_1,i_4} \delta_{i_2,i_3} \delta_{j_1,j_3} \delta_{j_2,j_4}\nonumber \\
& \qquad \qquad \qquad \qquad \qquad \qquad \quad +\textbf{J}_{i_1,i_2} \textbf{J}_{i_3,i_4} \delta_{j_1,j_3} \delta_{j_2,j_4} - \delta_{i_1,i_4} \delta_{i_2,i_3} \textbf{J}_{j_1,j_2} \textbf{J}_{j_3,j_4} + \delta_{i_1,i_3} \delta_{i_2,i_4} \textbf{J}_{j_1,j_2} \textbf{J}_{j_3,j_4} \big) \nonumber
\end{align}
where the expression for~\eqref{E:sympmoment2} comes from~\cite{CHJ1}.

Similar to the unitary and orthogonal cases, we have the following useful bound:
\begin{CM4p10}
For any $\sigma_{\mathfrak{m}} \in M_{2k}$ and $D > 6 k^{7/2}$, we have
\begin{align}
\label{eq:bound3}
\frac{1}{1 - \frac{k-1}{(D/2)^2}} \leq \frac{D^{k + |\sigma_{\mathfrak{m}}|_{\text{Sp}}} |\text{Wg}^\text{Sp}(\sigma_{\mathfrak{m}},D/2)|}{\prod_i \frac{(2 \mu_i - 2)!}{(\mu_i-1)! \mu_i!}} \leq \frac{1}{1 - \frac{6 k^{7/2}}{(D/2)^2}}
\end{align}
where $\sigma_{\mathfrak{m}}$ has $M_{2k}$-cycle type $(\mu_1,\mu_2,...)$ and $|\sigma_{\mathfrak{m}}|_{\text{Sp}} = |\sigma_{\mathfrak{m}}|_{O}$.
\end{CM4p10}
\noindent We also have
\begin{equation}
\label{eq:orthogsymplecrelation1}
\text{Wg}^{\text{Sp}}(\sigma_{\mathfrak{m}}, D/2) = (-1)^k \epsilon(\sigma_{\mathfrak{m}}) \text{Wg}^O(\sigma_{\mathfrak{m}}, -D)
\end{equation}
where $\epsilon(\sigma_{\mathfrak{m}})$ is the signature of $\sigma_{\mathfrak{m}}$ as an element of $S_{2k}$\,.

We will now present two Lemmas analogous to the unitary and orthogonal cases:
\begin{lemma}
\label{lemm1AppB3}
$|\text{Wg}^\text{Sp}(\sigma_\mathfrak{e},D/2) - D^{-k}| \leq \mathcal{O}(k^{7/2} D^{-(k+2)})$
\end{lemma}
\begin{proof}
This follows directly from Eqn.'s~\eqref{eq:bound3} and~\eqref{eq:orthogsymplecrelation1}.
\end{proof}

\begin{lemma}
\label{lemm2AppB3}
$\sum_{\mathfrak{m} \in P_2(2k)} |\text{Wg}^\text{Sp}(\sigma_{\mathfrak{m}},D/2)| = \prod_{j=0}^{k-1} \frac{1}{D + 2j}$\,.
\end{lemma}
\begin{proof}
We have $|\text{Wg}^{\text{Sp}}(\tau_{\mathfrak{m}},D/2)| = |\text{Wg}^{O}(\tau_{\mathfrak{m}},-D)|$.  Since for large enough $D$, $\text{Wg}^{O}(\tau_{\mathfrak{m}},D/2) \sim C_{\mathfrak{m}}(-1)^{|\tau_{\mathfrak{m}}|_O} D^{-k-|\tau_{\mathfrak{m}}|_O}$ for $C_{\mathfrak{m}} > 0$, it follows that $(-1)^{|\tau_{\mathfrak{m}}|_O} \text{Wg}^{O}(\tau_{\mathfrak{m}},D/2) > 0$ and so $|\text{Wg}^{O}(\tau_{\mathfrak{m}},D/2)| = (-1)^{|\tau_{\mathfrak{m}}|_O} \text{Wg}^{O}(\tau_{\mathfrak{m}},D/2)$.  Similarly, $\text{Wg}^{O}(\tau_{\mathfrak{m}},-D) \sim C_{\mathfrak{m}}(-1)^{|\tau_{\mathfrak{m}}|_O} (-D)^{-k-|\tau_{\mathfrak{m}}|_O}$ and so we find

\noindent $(-1)^{|\tau_{\mathfrak{m}}|_O} \text{Wg}^{O}(\tau_{\mathfrak{m}},-D) \sim C_{\mathfrak{m}} (-D)^{-k-|\tau_{\mathfrak{m}}|_O}$.  To make both sides positive, we can multiply by $(-1)^{-k-|\tau_{\mathfrak{m}}|_O}$, and get
\begin{equation}
|\text{Wg}^{\text{Sp}}(\tau_{\mathfrak{m}},D/2)| = |\text{Wg}^{O}(\tau_{\mathfrak{m}},-D)| = (-1)^k \, \text{Wg}^{O}(\tau_{\mathfrak{m}},-D) \,.
\end{equation}
Now consider
\begin{equation}
G_{-D}^{O}(\sigma_{\mathfrak{m}}, \sigma_{\mathfrak{n}}) = (-D)^{\#_O(\sigma_{\mathfrak{m}}^{-1} \sigma_{\mathfrak{n}})}
\end{equation}
Similar to before, we have
\begin{equation}
\sum_{\mathfrak{n} \in P_2(2k)} \text{Wg}_{-D}^O(\sigma_{\mathfrak{m}}, \sigma_{\mathfrak{n}}) \,G_{-D}^O(\sigma_{\mathfrak{n}}, \sigma_{\mathfrak{p}}) = \delta_{\mathfrak{m},\mathfrak{p}}\,.
\end{equation}
Then similar to our previous proofs,
\begin{equation}
\text{Wg}_{-D}^{O} \cdot \vec{1} = \left(\sum_{\mathfrak{m} \in P_2(2k)} \text{Wg}^{O}(\tau_{\mathfrak{m}}, -D)\right) \, \vec{1}
\end{equation}
and thus
\begin{equation}
\left(\sum_{\mathfrak{m} } \text{Wg}^{O}(\tau_{\mathfrak{m}}, -D)\right)^{-1} = \sum_{\mathfrak{m}} G^{O}(\tau_{\mathfrak{m}}, -D) = \sum_{\mathfrak{m}} (-D)^{\#_O(\tau_{\mathfrak{m}})} = \sum_{j=1}^k \widetilde{s}(k,j) \, (-D)^j\,.
\end{equation}
This is clearly equivalent to
\begin{equation}
\prod_{j=0}^{k-1} (-D - 2j) = (-1)^k \, \prod_{j=0}^{k-1} (D + 2j) \,.
\end{equation}
Finally, we have
\begin{align}
\sum_{\mathfrak{m} \in P_2(2k)} |\text{Wg}^{\text{Sp}}(\tau_{\mathfrak{m}},D/2)| = (-1)^k \sum_{\mathfrak{m} \in P_2(2k)} \text{Wg}^{O}(\tau_{\mathfrak{m}},-D) = \prod_{j=0}^{k-1} \frac{1}{D + 2j}
\end{align}
as claimed.
\end{proof}

\bibliographystyle{apa}

\end{document}